%% file: main.tex
\documentclass[a4paper,fleqn,oneside]{cas-sc}

\usepackage[expansion=false]{microtype}
\usepackage{enumitem}
\usepackage{algorithm}
\usepackage{algpseudocode}
\usepackage{amsthm}
\usepackage[numbers]{natbib}
\usepackage{cleveref}

\graphicspath{{figures/}}

\theoremstyle{plain}
\newtheorem{theorem}{Theorem}
\newtheorem{corollary}{Corollary}

\newtheorem{proposition}{Proposition}
\theoremstyle{remark}
\newtheorem{remark}{Remark}

\newcommand{\R}{\mathbb{R}}
\newcommand{\E}{\mathbb{E}}
\newcommand{\Var}{\mathrm{Var}}
\newcommand{\Cov}{\mathrm{Cov}}
\newcommand{\N}{\mathcal{N}}
\newcommand{\dt}{\Delta t}
\newcommand{\Lop}{\mathcal{L}}
\newcommand{\as}{\overset{\mathrm{a.s.}}{\longrightarrow}}
\newcommand{\inlaw}{\overset{d}{\longrightarrow}}

\begin{document}

\let\WriteBookmarks\relax
\def\floatpagepagefraction{1}
\def\textpagefraction{.001}

\shorttitle{Weak-Form Recovery of Stochastic Generators}

\shortauthors{Eshwar R.~A. and G.~V.~Honnavar}

\title[mode=title]{Weak-Form Recovery of Stochastic Generators}

\author[1]{Eshwar~R.~A}
\cormark[1]
\ead{eshwarra5@gmail.com}
\affiliation[1]{organization={Department of Computer Science Engineering, PES University (EC Campus)},
            city={Bengaluru},
            postcode={560100},
            country={India}}

\author[2]{Gajanan~V.~Honnavar}
\ead{gajanan.v@qnulabs.com}
\affiliation[2]{organization={QNu Labs Pvt Ltd, Centenary Building, M.G.\ Road},
            city={Bengaluru},
            postcode={560025},
            country={India}}

\cortext[1]{Corresponding author.}

\begin{abstract}
Data-driven discovery of governing equations is well developed for
deterministic systems but much less so for stochastic ones. We present a
sparse-symbolic algorithm recovering both the drift $b(x)$ and diffusion
$a(x)$ of a stochastic differential equation from a single trajectory,
replacing the time-indexed test functions of weak-form SINDy with
state-dependent spatial kernels. Fixed-width time-window weighting induces a
persistent endogeneity bias (for the Ornstein--Uhlenbeck process it
erases the drift coefficient entirely, whatever the window shape), while spatial kernels avoid
this via a martingale strong-law argument linked to the
estimating-function theory of Bibby and S{\o}rensen. We establish
consistency, a $T^{-1/2}$ central-limit rate, a finite-step correction
linking drift and diffusion estimation, and a correction for additive
Gaussian observation noise. The method is validated on three canonical
SDEs, eight further systems identified from single trajectories with
elastic-net-stabilized sparse regression, and a dedicated
support-recovery experiment.
\end{abstract}


\maketitle

\section{Introduction}

\subsection{What We Want to Learn and Why It Matters}

Stochastic differential equations (SDEs) describe dynamical systems that evolve
under a combination of deterministic forces and intrinsic random fluctuations.
Examples include the velocity of a particle in a fluid, the folding and
unfolding of a protein, the firing rate of a neural population driven by
synaptic noise, species concentrations in a noisy reaction network, and
asset prices under market volatility. In each case the dynamics are
modelled by an It\^{o} SDE,
\begin{equation}
  dX_t = b(X_t)\,dt + \sigma(X_t)\,dW_t,
  \label{eq:sde_intro}
\end{equation}
where $W_t$ is a standard Wiener process modelling the random forcing.
Two functions drive the state $X_t$. The \emph{drift} $b(x)$ is the
deterministic velocity the state would have without noise, and the
\emph{diffusion} $\sigma(x)$ is the local noise intensity, which may vary
with position. Together they fix the statistics of the process, including
its relaxation rate to equilibrium and the dwell time in a metastable
state.

We want to learn $b(x)$ and $\sigma(x)$ in symbolic form from a discretely
sampled trajectory $X_{t_0}, X_{t_1}, \ldots, X_{t_N}$.

\subsection{The Noise Problem in Increment-Based Identification}
\label{sec:noiseproblem}

\paragraph{The classical estimator and its noise problem.}
The standard route to identifying the SDE directly from data is the
Kramers--Moyal expansion, on which stochastic
SINDy~\cite{boninsegna2018,callaham2021} is built.  The drift and
diffusion are read off as the first two conditional moments of the
increments,
\begin{align}
  b(x)
  &= \lim_{\dt\to0}\frac{1}{\dt}\,
     \E\!\left[\,X_{t+\dt}-x \,\middle|\, X_t = x\right],
  \label{eq:km_drift}\\
  a(x)
  &= \lim_{\dt\to0}\frac{1}{\dt}\,
     \E\!\left[(X_{t+\dt}-x)^2 \,\middle|\, X_t = x\right],
  \label{eq:km_diff}
\end{align}
and at finite sampling each data point supplies a noisy estimate of these
limits by dropping the expectation and the limit: the increment ratio
$\Delta X_n/\dt$ is regressed against the library as a sample of $b(X_{t_n})$,
and $(\Delta X_n)^2/\dt$ as a sample of $a(X_{t_n})$.

The difficulty is the $1/\dt$ that these formulas place in the denominator.
Substituting the Euler--Maruyama increment $\Delta X_n = b(X_{t_n})\dt +
\sigma(X_{t_n})\xi_n\sqrt{\dt}$ into the drift estimator gives
\begin{equation}
  \frac{\Delta X_n}{\dt}
  = b(X_{t_n}) + \frac{\sigma(X_{t_n})}{\sqrt{\dt}}\,\xi_n,
  \label{eq:km_noise}
\end{equation}
so the noise attached to each regression target has variance
$\sigma(X_{t_n})^2/\dt$, which \emph{diverges} as the sampling interval
shrinks. Sampling faster therefore makes each individual drift estimate
noisier, and the same $1/\dt$ factor affects the diffusion estimator.

\paragraph{Binning is a partial remedy, not a cure.}
The usual practical workaround is to bin the state
space and average the noisy increment ratios over all visits to each bin
before regressing.  Averaging $m$ samples in a bin reduces the noise variance
from $\sigma^2/\dt$ to $\sigma^2/(m\dt)$, which makes the regression workable.
But binning leaves the $1/\dt$ structure in place and only averages it out
spatially, at a cost. It needs many visits per bin, which in practice limits it
to one or two dimensions. It discards resolution within each bin. And it
assumes the trajectory samples the invariant density well in every bin, which
fails most near potential barriers, where visits are rare.

\subsection{Our Algorithm and What Sets It Apart}

\paragraph{Our reformulation removes the amplification at the source.}
We avoid forming the noise-amplifying ratio at all.  Instead of dividing the
increment by $\dt$ to estimate a pointwise derivative and then averaging, we
multiply the raw increment by a state-dependent weighting function $K_j$ and
\emph{sum} along the trajectory,
\begin{equation}
  B_j = \sum_{n} K_j(X_{t_n})\,\Delta X_n
      = \underbrace{\sum_n K_j(X_{t_n})\,b(X_{t_n})\,\dt}_{\text{signal}}
      + \underbrace{\sum_n K_j(X_{t_n})\,\sigma(X_{t_n})\,\xi_n\sqrt{\dt}}_{\text{noise}}.
  \label{eq:weak_intro}
\end{equation}
The noise term is now a sum of zero-mean independent contributions, each of
order $\sqrt{\dt}$, so its magnitude grows only as the square root of the
number of terms while the signal grows linearly.  The fluctuations therefore
average out at the $1/\!\sqrt{N}$ rate, and because the increment is never
divided by $\dt$, the variance does not blow up as the sampling rate increases.  The \emph{process}-noise variance of the weak-form statistic therefore stays bounded as $\dt\to0$ (as does that of any estimator that averages over visits, such as binning, which however must discretise the state space); bias from \emph{observation} noise is a separate matter, handled by the Gaussian correction of \cref{sec:noiserobust}.  This weak-projection
reformulation, which we develop fully in \cref{sec:weakproj}, applies the
same trajectory-averaging principle that Weak SINDy~\cite{messenger2021}
uses for deterministic systems, extended here to the full stochastic model
including the diffusion term.

\paragraph{A subtlety: fixed-width time windows versus state-dependent kernels.}
The weighting used in deterministic weak-form methods is a function of
\emph{time}, $\varphi_j(t_n)$. Long trajectories force it to be used in tiled
windows of fixed temporal width, so that the number of test functions grows with
the trajectory, and on a stochastic equation this is asymptotically biased:
more data does not remove the coefficient error. (We do not claim state-dependence
is \emph{necessary}, only that fixed-width time windows fail and fixed spatial
kernels succeed.) Predictability does not explain the difference. A
deterministic weight $\varphi_j(t_n)$ is $\mathcal{F}_{t_n}$-measurable, just
like $K_j(X_{t_n})$, so both make the per-step noise term zero-mean by the same
one-line tower-property argument. What separates them is the
\emph{signal-to-noise ratio per regression equation} as $T$ grows.  A fixed-width temporal window accumulates only a bounded number
of observations, so its signal and its noise stay of fixed size and its noise
contribution never averages down; the window-level cross term then stays
nonzero from block to block, and stacking more windows preserves it.  Widening
the window does not cure this: for OU with $\varphi\equiv1$ on a single
window of width $L=T$, the response is $X_T-X_0=O(1)$ and the regressor
$\int_0^TX\,dt=O(\sqrt T)$ because $X$ has mean zero, so the drift signal is
of the same order as the noise and the estimate still tends to $0$
(\cref{prop:ou_temporal}).  More observations per equation help only if the regressor's weighted average
does not vanish. A fixed set of spatial kernels satisfies this: the weights are
nonnegative and local, so each kernel is revisited arbitrarily often, its signal
grows linearly in $T$ and its noise like $\sqrt T$, and a martingale strong-law
argument drives the noise contribution to zero. This is a known fact about
discretely observed diffusions, namely that predictable weights with growing
information per equation give consistent estimating functions
(Bibby and S{\o}rensen~\cite{bibby1995,kessler1999}). What we add is the
statement of when it holds for test functions built for trajectory-averaged
sparse regression, and a quantification of how badly the windowed-temporal
choice fails when it does not:

\begin{theorem}[Endogeneity of fixed-width temporal projections, informal]
\label{thm:endogeneity_informal}
Let $\varphi_j(t_n)$, $j=1,\ldots,M_{\mathrm{win}}$, be non-constant temporal test
functions of fixed real-time width $L$ (independent of $T$), tiled across
$[0,T]$ so that $M_{\mathrm{win}}=T/L$, applied to the weak projection of a scalar
It\^{o} SDE. Under the lag-covariance nondegeneracy stated in
Theorem~\ref{thm:endogeneity}, the corresponding ordinary least-squares
estimator does not converge to $c^*$ as $T\to\infty$ at fixed $\Delta t$
and fixed $L$; instead it converges to a generally incorrect,
window-dependent limit. The numerator and information matrix in the
normal equations are sums of stationary window-level contributions and
therefore both grow linearly in $T$. Their ratio consequently approaches
a finite constant rather than vanishing. The theorem does \emph{not}
claim that the limiting bias is independent of $L$: its dependence on
window width is determined by the chosen temporal test-function
normalisation and by the window-level covariance structure. The numerical
window-width audit below shows that, for the OU construction used here,
the limiting coefficient error is nearly unchanged over the tested range.
\end{theorem}

The mechanism has two parts. First, within a single window, later states
$X_{t_{n'}}$ ($n'>n$) depend on the earlier Brownian innovation $\xi_n$ through
the SDE recursion, so the window's regressor is correlated with its own
residual, an SDE analogue of simultaneous-equations endogeneity. Because the
window has fixed width $L$, its sample size $L/\Delta t$ is fixed and this
correlation does not shrink with $T$. Second, the dynamics are stationary, so
every window contributes the same population law, and under the non-cancellation
assumption of Theorem~\ref{thm:endogeneity} its expectation is nonzero. Stacking
$M_{\mathrm{win}}=T/L$ windows therefore keeps that cross term, and the
numerator and information matrix of the normal equations both grow linearly in
$T$. A \emph{fixed} set of $M$ state-dependent kernels avoids this: each is
revisited an ever-growing number of times, its signal grows linearly in $T$ and
its noise like $\sqrt T$, and the martingale strong law, not a cross-step
cancellation, drives the noise contribution to zero
(\cref{sec:weakproj,sec:kernels}). In this sense the spatial kernels realise the
martingale estimating-function principle of Bibby and
S{\o}rensen~\cite{bibby1995} within a trajectory-averaged sparse-regression
setting. Figure~\ref{fig:endogeneity} shows the contrast: the tiled-boxcar
variant plateaus at a coefficient error of $\approx1.1$ however much data it
gets, while the state-weighted algorithm converges to the truth.

\paragraph{Distinguishing properties.}
The algorithm maps trajectory data to a sparse symbolic model
$(\hat{b}, \hat{a})$ through a single shared regression on state-weighted
trajectory averages (\cref{sec:algorithm}). Compared with prior stochastic
identification methods, it avoids the per-point $1/\dt$ variance growth of
increment-ratio regression without state-space binning, recovers drift and
diffusion jointly from one weighted pass over the data, and returns an
interpretable polynomial model rather than a black-box surrogate.  At
fixed sampling interval the drift and diffusion estimators are strongly
consistent for the Euler--Maruyama coefficients, while the drift estimating
equation has the $T^{-1/2}$ central-limit rate established in
\cref{thm:clt}; no separate diffusion-coefficient CLT is claimed.

We validate the algorithm on three stochastic systems: the
Ornstein--Uhlenbeck (OU) process, a double-well Langevin system, and a system
with multiplicative diffusion. For each we report recovered spectral gaps,
Kramers escape rates, stationary densities and position-dependent relaxation
timescales.

\subsection{Relation to Existing Work}
\label{sec:related_intro}

Four lines of work are closest to ours.

\textit{Stochastic SINDy}~\cite{boninsegna2018,callaham2021} extends the
SINDy framework~\cite{brunton2016} to SDEs using Kramers--Moyal increment
statistics, establishing infinite-data consistency and, in later
work~\cite{wanner2024}, sharper finite-data drift/diffusion estimates.
These methods are interpretable and produce symbolic models, but each
regression row is formed from a single time-step increment, so estimation
variance grows as $\dt\to0$; the practical remedy is to bin the state space
and average within bins, which is effective in low dimension but discards
spatial resolution and degrades sharply as dimension grows.

\textit{Weak SINDy}~\cite{messenger2021} achieves trajectory-level averaging
for deterministic ODEs and PDEs via temporal test functions, substantially
reducing noise sensitivity relative to pointwise finite differences. Its
temporal-projection machinery does not carry over to SDEs without
modification: applied directly, it does not identify the diffusion
coefficient $a(x)$, and, as \cref{sec:endogeneity} shows, its fixed-width
temporal test functions can produce a persistent coefficient-level bias rather
than a bias that vanishes with more data.

\textit{Operator/generator-estimation methods}, in particular
gEDMD~\cite{klus2020}, already recover the drift and diffusion of an SDE
from data via a state-dependent Galerkin projection of the Koopman
generator, and the stochastic occupation kernel method~\cite{wells2024}
recovers drift and diffusion from a reproducing-kernel integral (occupation
kernel) formulation. Both are, like the present method, weak/integral
formulations that avoid pointwise derivative estimation; the present
algorithm differs in combining this projection principle with an explicit
sparse polynomial library, a single-trajectory Weak-SINDy-style averaging
formulation, and the quantified temporal-vs-spatial bias analysis of
\cref{sec:endogeneity}, rather than a Galerkin basis of general observables
or a fully nonparametric RKHS diffusion estimate.
\textit{Langevin regression}~\cite{callaham2021} likewise couples
Kramers--Moyal-type averaging with a finite-time-step (adjoint
Fokker--Planck) bias correction and sparse model selection, addressing the
same finite-$\dt$ contamination of the diffusion estimate that
\cref{sec:bias} treats here via a drift-informed correction tied directly
to the shared design matrix of the spatial-kernel projection.

\textit{Martingale estimating functions}~\cite{bibby1995,kessler1999}
established, in the statistics literature on discretely observed diffusions,
that state-dependent (spatial) weights yield unbiased, asymptotically
normal estimating equations for scalar diffusion coefficients, in contrast
to weights built from the naive discretised likelihood. The endogeneity
analysis of \cref{sec:endogeneity} is this principle applied to, and made
quantitative within, a Weak-SINDy-style sparse-regression setting; as
\cref{sec:kernels} shows, the relevant distinction is not measurability
(both temporal and spatial weights are equally predictable in the narrow,
per-step sense) but whether the sample size entering each regression
equation grows without bound as $T\to\infty$.

We do not claim that weak or state-dependent SDE estimators are new: gEDMD, the
occupation kernel method and the martingale estimating-function literature
already provide unbiased ones. Our contributions are (i)~an explanation of why
the temporal test functions of weak-form ODE/PDE methods fail on stochastic data,
namely a persistent coefficient-level asymptotic bias for fixed-width temporal
projections, with a window-width plateau in the OU benchmark, which we have not
seen stated in this form, and (ii)~an algorithm that couples this diagnosis to a
single shared design matrix for joint sparse symbolic drift and diffusion
recovery, usable on a single trajectory and with the finite-step bias correction.

\section{Background and Setup}

\subsection{Sparse Identification of Nonlinear Dynamics}

The SINDy framework~\cite{brunton2016} identifies governing ODEs by
expressing the right-hand side as a sparse linear combination of library
functions.  Given trajectory data, one constructs a feature matrix
$\Theta(X)=[1,X_1,X_2,X_1^2,\ldots]$ and a vector of state
derivatives $\dot X$, then solves
\begin{equation}
  \hat{c} = \arg\min_{c}
  \|\dot X - \Theta(X)\,c\|_2^2 + \lambda\|c\|_1.
  \label{eq:sindy}
\end{equation}
This has been extended to implicit dynamics~\cite{kaheman2020},
PDEs~\cite{rudy2017}, and model selection via information
criteria~\cite{mangan2017}.  For noisy observations
$\tilde{X}_{t_n}=X_{t_n}+\eta_n$, the finite-difference derivative
has variance $\Var[(\tilde{X}_{t+\dt}-\tilde{X}_t)/\dt]\propto\sigma_\eta^2/\dt^2$,
diverging as $\dt\to0$.

\subsection{Weak SINDy for Deterministic Systems}

Messenger and Bortz~\cite{messenger2021} address noise amplification for
deterministic systems through Galerkin projection.  For an ODE
$\dot{X}=F(X)$, projecting onto a temporal test function $\varphi_j(t)$
and integrating by parts yields
\begin{equation}
  -\int_0^T X_t\,\varphi'_j(t)\,dt + [X\varphi_j]_0^T
  = \int_0^T F(X_t)\,\varphi_j(t)\,dt,
  \label{eq:weaksindy}
\end{equation}
aggregating information over the entire trajectory and averaging noise at
rate $1/\!\sqrt{N}$.  When applied to a stochastic equation, two problems
arise: (i)~the stochastic integral $\int\varphi_j(t)\sigma(X_t)\,dW_t$
does not vanish; and (ii)~temporal test functions introduce the
endogeneity bias quantified in \cref{sec:endogeneity}, which is the
Weak-SINDy-specific instance of the general state-dependent-weight
requirement established for discretely observed diffusions by Bibby and
S{\o}rensen~\cite{bibby1995}.

\subsection{The Generator and the Quantities It Encodes}
\label{sec:spectral}

The two functions we identify, $b(x)$ and $\sigma(x)$, together define the
infinitesimal generator of the diffusion~\cref{eq:sde_intro},
\begin{equation}
  \Lop f(x) = b(x)\cdot\nabla f(x)
             + \tfrac{1}{2}\,a(x):\nabla^2 f(x),
  \qquad a(x)=\sigma(x)\sigma(x)^\top,
  \label{eq:generator}
\end{equation}
which governs the evolution of expectations: $\tfrac{d}{dt}\E[f(X_t)\mid
X_0=x]=\Lop f(x)$.  The stationary density $\pi(x)$ satisfies the
Fokker--Planck equation $\Lop^\dagger\pi=0$; for a one-dimensional system
with $a(x)>0$,
\begin{equation}
  \pi(x) \propto \frac{1}{a(x)}
  \exp\!\Bigl(2\!\int_0^x \frac{b(y)}{a(y)}\,dy\Bigr).
  \label{eq:fp}
\end{equation}
Once $b(x)$ and $a(x)$ are identified, three practically important
quantities follow directly without further simulation.

\textit{Relaxation timescale.}
The rate at which the system forgets its initial condition is controlled
by the spectral gap $\lambda_1$ of $\Lop$ in $L^2(\mu)$:
\begin{equation}
  \bigl\|\mathrm{Law}(X_t) - \mu\bigr\|_{\mathrm{TV}}
  \le C\,e^{-\lambda_1 t},
  \label{eq:spectralgap}
\end{equation}
so $1/\lambda_1$ is the characteristic relaxation time.
For the Ornstein--Uhlenbeck process $dX_t = -\theta X_t\,dt + \sigma_0\,dW_t$,
$\lambda_1 = \theta$; recovering the coefficient $c_x = -\theta$ from
data directly gives the relaxation rate.

\textit{Escape rates from metastable states.}
In multi-stable systems the spectral gap is related to inter-well
transition rates.  For the double-well potential
$V(x) = -x^2/2 + x^4/4$, Kramers' formula~\cite{kramers1940,hanggi1990}, \cref{eq:kramers_dw}, gives
the mean escape time as
\begin{equation}
  \tau_{\mathrm{Kramers}}
  = \frac{2\pi}{\sqrt{|V''(0)|\,V''(\pm1)}}
    \exp\!\Bigl(\frac{2\Delta V}{\sigma_0^2}\Bigr)
  = \sqrt{2}\,\pi\,\exp\!\Bigl(\frac{1}{2\sigma_0^2}\Bigr),
  \label{eq:kramers_dw}
\end{equation}
with barrier height $\Delta V = V(0) - V(\pm1) = 1/4$.  Errors in the
identified polynomial coefficients propagate directly to errors in
$\tau_{\mathrm{Kramers}}$ through the potential curvatures and barrier height.

\textit{Position-dependent mixing.}
For systems with multiplicative diffusion, the local relaxation rate at
position $x$ is $a(x)/|b'(x)|$.  Getting the state-dependent structure of
$a(x)$ right is needed for the local dynamics, not only the marginal
distribution.

\section{Theoretical Framework}
\label{sec:theory_framework}

\subsection{Standing Assumptions}
\label{sec:assumptions}

\begin{enumerate}[label=\textbf{A\arabic*.},leftmargin=*]
\item\label{ass:ergodic}
  \textbf{(Geometric ergodicity.)} The SDE~\cref{eq:sde_intro} has a unique
  invariant probability measure $\mu$ with smooth, strictly positive
  Lebesgue density $\pi(x)$.  The associated Markov semigroup is
  geometrically ergodic: there exist $C>0$ and $\rho\in(0,1)$ such
  that for all bounded measurable $g$ and all $x\in\R$,
  $|\E^x[g(X_t)] - \int g\,d\mu| \le C\|g\|_\infty\rho^t$.
\item\label{ass:regularity}
  \textbf{(Regularity.)} The drift $b$ and diffusion $\sigma$ are locally
  Lipschitz with linear growth: $|b(x)|+|\sigma(x)|\le C(1+|x|)$.  The
  kernels $K_1,\ldots,K_M$ are bounded and Lipschitz, the library functions
  $f_1,\ldots,f_K$ are locally Lipschitz (monomials are not bounded), and the
  products $K_jf_k$ are bounded and Lipschitz, as holds for Gaussian kernels
  with polynomial libraries.  The true coefficient vector $c^*$
  satisfies $b(x)=\Theta(x)c^*$ exactly.
\item\label{ass:em_ergodic}
  \textbf{(Ergodicity of the Euler--Maruyama chain.)} At the step size
  $\dt$ used, the Euler--Maruyama chain~\cref{eq:em} has a unique invariant
  probability measure $\mu_{\dt}$ and is geometrically ergodic.
\end{enumerate}

Assumption~\ref{ass:ergodic} (a property of the SDE) is satisfied by all three benchmark systems:
the OU process through its explicit Gaussian transition kernel, and the
double-well and multiplicative systems through Foster--Lyapunov criteria
with $V(x)=1+x^2$~\cite{pavliotis2014,gardiner2009}.

\begin{remark}[Scope of Assumption~\ref{ass:em_ergodic}]
\label{rmk:em_ergodicity}
Assumption~\ref{ass:ergodic} concerns the continuous-time SDE, whereas the
proofs of Theorems~\ref{thm:consistency} and~\ref{thm:clt} work with the
\emph{discrete-time} Euler--Maruyama chain, so they need
Assumption~\ref{ass:em_ergodic}, which does not follow from
Assumption~\ref{ass:ergodic}.  It holds for the two benchmarks with
globally Lipschitz drift and at most linearly growing diffusion: for the OU
process the one-step map contracts when $|1-\theta\dt|<1$, and for the
multiplicative system ($b(x)=-2x$, $\sigma(x)=\tfrac12\sqrt{1+x^2}$)
the Lyapunov function $V(x)=1+x^2$ satisfies
$\E[V(X_{n+1})\mid X_n=x]\le(1-\gamma\dt)V(x)+C\dt$ with
$\gamma>0$ for small $\dt$, which gives geometric ergodicity by the
standard Foster--Lyapunov argument~\cite{meyn2009,mattingly2002}.
It does \emph{not} hold for the cubic double-well drift $b(x)=x-x^3$.
There $|x-\dt\,x^3|$ grows like $\dt|x|^3$ for large $|x|$, the explicit
scheme is not contractive outside any bounded set, and for every $\dt>0$
the chain diverges with positive probability~\cite{mattingly2002,
hutzenthaler2011}.  We therefore do \emph{not} claim that
Theorems~\ref{thm:consistency} and~\ref{thm:clt} cover the
Euler-simulated double well.  Its results are empirical: at
$\dt=0.002$--$0.005$ no simulated trajectory left a bounded region, and
the horizons are many orders of magnitude shorter than the divergence
time, but this is an observation about our runs, not a theorem.  A
tamed or implicit scheme would restore Assumption~\ref{ass:em_ergodic} for
such drifts; we have not tested the estimator on data generated that way.
\end{remark}

\subsection{The Weak Projection of a Stochastic Flow}
\label{sec:weakproj}

Let $\{X_{t_0},\ldots,X_{t_N}\}$ be discrete observations of the
scalar It\^{o} diffusion~\cref{eq:sde_intro} at times $t_n = n\dt$.  We assume
the true drift and diffusion can be expressed as sparse linear combinations
of a known feature library:
\begin{equation}
  b(x) = \Theta(x)\,c^*, \qquad a(x) = \Theta(x)\,d^*,
  \label{eq:library}
\end{equation}
where $\Theta(x)=[f_1(x),\ldots,f_K(x)]$ and $c^*,d^*\in\R^K$ are sparse.

The Euler--Maruyama discretisation of~\cref{eq:sde_intro} is
\begin{equation}
  \Delta X_n = b(X_{t_n})\,\dt + \sigma(X_{t_n})\,\xi_n\,\sqrt{\dt},
  \label{eq:em}
\end{equation}
where $\xi_n=(W_{t_{n+1}}-W_{t_n})/\!\sqrt{\dt}\sim\N(0,1)$ are i.i.d.\
and $\xi_n\perp\mathcal{F}_{t_n}$.  For a general test function
$\psi_j:\R\times[0,T]\to\R$, the \emph{weak projection} is
\begin{align}
  S_j &= \sum_{n=0}^{N-1}\psi_j(X_{t_n},t_n)\,\Delta X_n \notag\\
      &= \underbrace{\sum_n \psi_j(X_{t_n},t_n)\,b(X_{t_n})\,\dt}_{\text{drift term}}
       + \underbrace{\sum_n \psi_j(X_{t_n},t_n)\,\sigma(X_{t_n})\,\xi_n\,\sqrt{\dt}}_{\text{stochastic term}}.
  \label{eq:weakproj}
\end{align}
For the regression to be valid, the stochastic term must have zero mean.
Whether it is satisfied, and whether the regression is asymptotically
unbiased, depends on the choice of test-function family, as we now show.

\subsection{Endogeneity of Temporal Test Functions}
\label{sec:endogeneity}

Whether the weak projection~\cref{eq:weakproj} produces a consistent
regression depends entirely on the test-function family.  We first show
that the natural choice, a function of time alone tiled into windows
of fixed width as deterministic Weak SINDy does on long trajectories,
fails, and then quantify how it fails.
\emph{Both} temporal and spatial weights make the per-step stochastic
contribution individually zero-mean: a deterministic weight
$\varphi_j(t_n)$ and a state-dependent weight $K_j(X_{t_n})$ are each
trivially $\mathcal{F}_{t_n}$-measurable, so predictability in this
narrow, per-step sense does not distinguish them. What distinguishes
them is whether the \emph{signal-to-noise ratio of each regression equation}
grows without bound as $T\to\infty$ (this needs both a growing number of
observations and a weighted regressor whose average does not vanish), and this is
the precise form, for a trajectory-averaged sparse-regression setting,
of the estimating-function principle of Bibby and
S{\o}rensen~\cite{bibby1995} and Kessler and
S{\o}rensen~\cite{kessler1999} for discretely observed diffusions.

Specialising~\cref{eq:weakproj} to $\psi_j(X_{t_n},t_n)=\varphi_j(t_n)$ and
substituting $b(x)=\Theta(x)c^*$ puts the projection in matrix form
$S = A^{\mathrm{temp}}c^* + Z^{\mathrm{temp}}$, with design matrix and
stochastic residual
\begin{equation}
  A^{\mathrm{temp}}_{jk}=\sum_n\varphi_j(t_n)\,f_k(X_{t_n})\,\dt,
  \qquad
  Z^{\mathrm{temp}}_j=\sqrt{\dt}\sum_n\varphi_j(t_n)\,\sigma(X_{t_n})\,\xi_n.
  \label{eq:temp_decomp}
\end{equation}
We take $\varphi_j$ to be supported on a window of fixed real-time width
$L$ (independent of $T$), the standard construction for applying
weak-form test functions to a long trajectory, so that the number of
windows $M_{\mathrm{win}}=T/L$ grows as $T$ grows while each window's own sample size
$L/\dt$ does not. As $\E[Z_j^{\mathrm{temp}}]=0$ term by term (shown
below), the estimator's deviation from $c^*$ is governed by
$(A^{\mathrm{temp}})^\top Z^{\mathrm{temp}}$, which is generically
nonzero \emph{in expectation}: this is the cross term we must bound, and
its scaling with $T$, not any failure of the marginal zero-mean
property, is what determines consistency.

The reason $\E[(A^{\mathrm{temp}})^\top Z^{\mathrm{temp}}]\ne0$ is a
\emph{cross-step} correlation internal to each window.  Each individual
noise term $\varphi_j(t_n)\sigma(X_{t_n})\xi_n$ is zero-mean, since
$\xi_n$ is independent of the current state.  But through the
Euler--Maruyama recursion the shock $\xi_m$ is embedded into every later
state $X_{t_n}$ with $n>m$, so it simultaneously contributes to
$Z^{\mathrm{temp}}$ at row $m$ and corrupts the regressor value
$f_k(X_{t_n})$ at all later rows $n>m$ \emph{within the same window}.
Crucially, this correlation does not require the absence of mixing.  The
analysis below therefore states an explicit lag-mixing bound under which the
covariance between the innovation at time $m$ and a later regressor
$f_k(X_{t_{m+\ell}})$ decays with the lag $\ell$.  What survives this decay
is the contribution from \emph{short} lags ($\ell=O(1)$ in step count), which
is not removed by long-range mixing.  The resulting window-level covariance
can therefore remain nonzero as $T$ grows; its precise dependence on the
window width is encoded in the lag sum $q_L$ in Theorem~\ref{thm:endogeneity}
and is not asserted to be universal. This is a close cousin of the endogeneity
that plagues simultaneous-equations regression in econometrics, here
induced by the SDE dynamics and compounded by the fixed-sample-size
structure of temporal windowing, rather than by reverse causality alone.

\begin{theorem}[Endogeneity of fixed-width temporal projections]
\label{thm:endogeneity}
Let $\psi_j(X_{t_n},t_n)=\varphi_j(t_n)$, $j=1,\ldots,M_{\mathrm{win}}$, be non-constant
purely temporal test functions, each supported on a window of fixed
real-time width $L$ independent of $T$ and tiled so that $M_{\mathrm{win}}=T/L$, and let
$\hat{c}^{\,\mathrm{temp}}$ be the ordinary least-squares estimator formed
from the stacked temporal weak projection~\cref{eq:weakproj}. In addition
to the standing assumptions of \cref{sec:assumptions}, assume that the
process is stationary after burn-in and that the window-level sequence is
ergodic. Assume also the lag-covariance bound
\begin{equation*}
  \left|\operatorname{Cov}\!\left[\sigma(X_{t_m})\xi_m,\,
f_k(X_{t_{m+\ell}})\right]\right|
  \le C_k\sqrt{\dt}\,e^{-\lambda(\ell-1)\dt},\qquad \ell\ge1,
\end{equation*}
for the relevant library functions, and assume that the resulting window-level
normal-equation matrix is nonsingular. Finally, suppose that for at least
one coefficient direction the corresponding lagged covariance sum does not
cancel exactly. Then, with fixed $\dt$ and $L$ as $T\to\infty$:
\begin{enumerate}[label=\textup{(\roman*)}]
  \item Each stepwise stochastic term has zero marginal mean:
    $\E[\varphi_j(t_n)\sigma(X_{t_n})\xi_n]=0$.
  \item For a window $r$, define its design row $a_r$ and stochastic
    residual $z_r$ by collecting the terms of~\cref{eq:temp_decomp} over
    that window. The block ergodic theorem gives
    \[
      \frac1M(A^{\mathrm{temp}})^\top A^{\mathrm{temp}}\;\as
      G_L:=\E[a_0^\top a_0],\qquad
      \frac1M(A^{\mathrm{temp}})^\top Z^{\mathrm{temp}}\;\as
      q_L:=\E[a_0^\top z_0].
    \]
    Thus the OLS estimator has the almost-sure limit
    \[
      \hat c^{\,\mathrm{temp}}\;\as\;
      c^*+G_L^{-1}q_L.
    \]
    Unless $q_L=0$, the limit is not $c^*$.
  \item The cross term $q_L$ is generically nonzero because
    $$
    \operatorname{Cov}[\sigma(X_{t_m})\xi_m,f_k(X_{t_{m+1}})]
    =\E[f_k'(X_{t_m})\sigma(X_{t_m})^2]\sqrt{\dt}+O(\dt),
    $$
    while the displayed lag bound makes the remaining contributions
    summable over $\ell$. The observation that the shock enters the next
    state at order $\sqrt{\dt}$ is therefore retained explicitly; what
    mixing changes is the number of effective future lags, not the existence
    of the short-lag term. Exact cancellation of the full lag sum is a
    nongeneric exceptional case.
\end{enumerate}
For the Ornstein--Uhlenbeck process the lag sum can be evaluated explicitly
(\cref{app:endogeneity}), giving $c_\infty^{\mathrm{temp}}=0$: the
windowed-temporal estimator loses the drift coefficient entirely. The theorem
does not assert that the limiting bias is independent of $L$; that dependence
is determined by the temporal test-function shape and the window-level
information matrix.
\end{theorem}

The proof is organised at the level of the tiled windows rather than by
counting all past--future pairs as independent.  The window-level ergodic
limit gives the exact $T\to\infty$ normal-equation limit, while the
lag-covariance expansion identifies the source of the nonzero population
cross term.  Each individual lag contributes only $O(\dt)$ to the
window-level sum and vanishes as $\dt\to0$; the $O(1)$ cross term $q_L$
arises from accumulating $O(1/(\lambda\dt))$ such lags across the mixing
window, under the geometric lag bound that makes the tail beyond the mixing
time summable.  The dependence on $L$ is left to the window-level normal
equations and is checked directly in the window-width experiment below.
For the OU example, the closed-form calculation is stronger: it gives the
exact limiting coefficient $0$.

Figure~\ref{fig:endogeneity} tests this on all three benchmark systems
($\dt=0.005$, $R=8$ trajectories, $M=40$ centres, 25 seeds; the temporal
estimator uses tiled boxcar windows, $\varphi_j=\mathbf 1$ on window $j$; the
error plotted is the relative error of the coefficient of $x$, with
$c_x^*=-1,\,+1,\,-2$ for OU, double well and multiplicative diffusion).  For
every system the temporal-window estimate stays at a relative error of about
one however long the trajectory is (median at $T=400$: $1.05$, $0.99$ and $1.07$),
i.e.\ $\hat c_x$ is driven to $\approx0$ and the drift term is lost, while the
spatial-kernel estimate decays with $T$ (median relative error at $T=400$: $0.036$,
$0.024$, $0.016$).  For OU this is exactly the limit
$c_\infty^{\mathrm{temp}}=0$ of Proposition~\ref{prop:ou_temporal}, valid for
every window shape; the slight excess of the OU plateau over $1$ we attribute to
finite-sample effects at $R=8$, $T\le400$, since a long-run simulation at
$\dt=0.005$, $T=2\times10^5$ gives $|\hat c|\le0.01$.  For the double well and the
multiplicative system the theorem only asserts a nonzero limit; that the
limit is again essentially zero is an empirical finding that we have not proved.

\begin{figure}[tp]
  \centering
  \includegraphics[width=\textwidth]{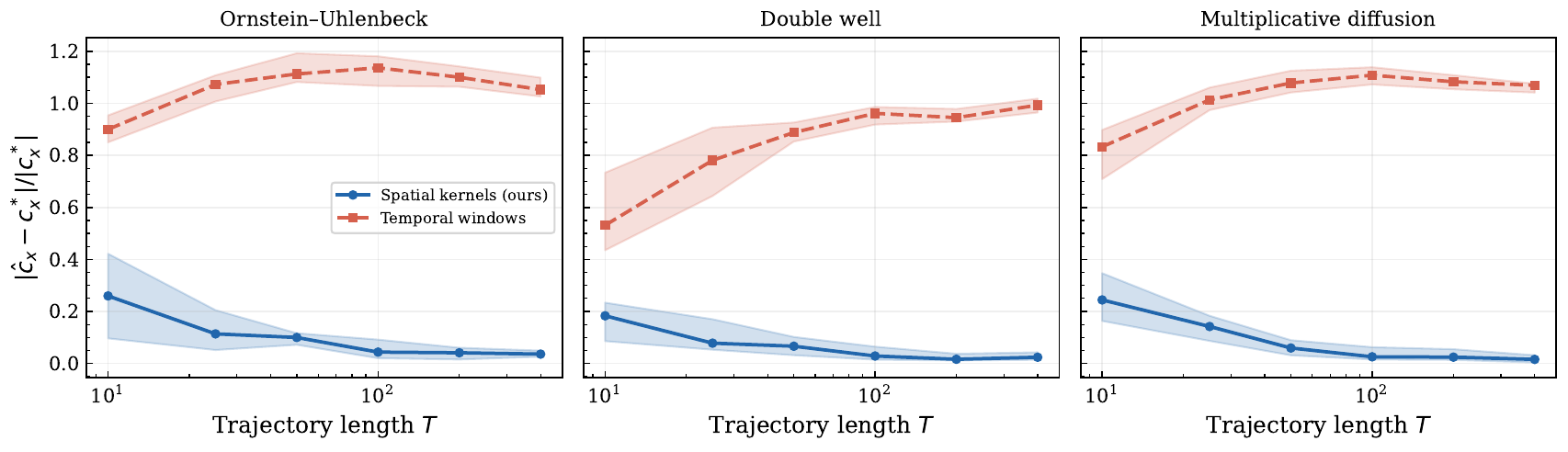}
  \caption{Empirical confirmation of the endogeneity bias
    (Theorem~\ref{thm:endogeneity}) on the three benchmark systems
    ($\dt=0.005$, $R=8$ trajectories, $M=40$ centres, 25 seeds; centres and
    bandwidth $h=0.25$ for OU and double well, $h=0.27$ with centres on
    $[-2.8,2.8]$ for the multiplicative system).  Relative error
    $|\hat{c}_x - c_x^*|/|c_x^*|$ of the coefficient of $x$ vs.\ trajectory length
    $T$.  Spatial-kernel OLS (blue) decays with $T$; temporal-window OLS with
    tiled boxcar windows (red) stays near $1$, i.e.\ $\hat c_x\to0$, for all three
    systems (for OU this is the limit of Proposition~\ref{prop:ou_temporal}; for the
    other two it is an empirical observation).  Shaded bands show interquartile
    ranges across the 25 repetitions.}
  \label{fig:endogeneity}
\end{figure}

\paragraph{Window-width audit.}
The theorem establishes a nonzero coefficient-level limit whenever
the lag-sum $q_L$ does not vanish; it does not by itself imply a universal
window-width plateau.  We therefore test the window-width dependence directly, on all three systems, at $T=400$, $\dt=0.005$, with $R=8$
trajectories and 25 independent repetitions, varying the number of tiled
windows $M_{\mathrm{win}}\in\{50,100,200,400,800\}$ and hence the window width
$L=400/M_{\mathrm{win}}\in\{8,4,2,1,0.5\}$.  The median relative coefficient error
stays between $1.03$ and $1.10$ for OU, $0.97$ and $0.99$ for the double well and
$1.05$ and $1.09$ for the multiplicative system across the entire range of $L$
(the median recovered $\hat c_x$ is $0.03$--$0.10$, $0.01$--$0.03$ and
$0.09$--$0.17$ respectively, against true values of magnitude $1$, $1$ and $2$).  Thus the empirical
plateau is insensitive to the tested window width, over both sides of each
system's relaxation time (marked), in agreement with the asymptotic picture.
\begin{figure}[tp]
  \centering
  \includegraphics[width=\textwidth]{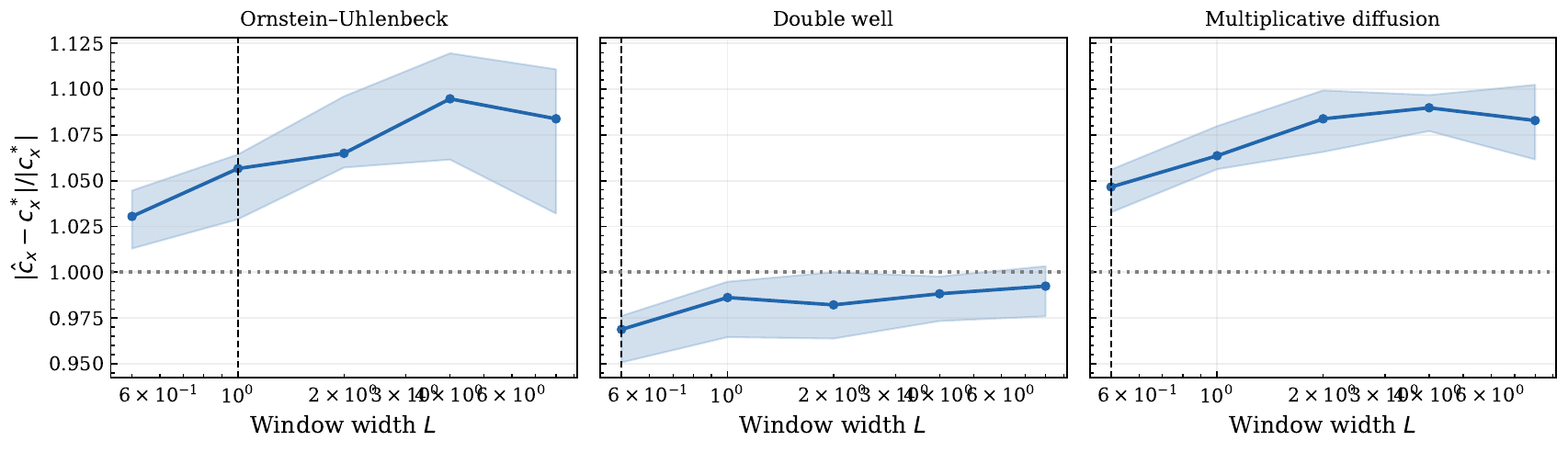}
  \caption{Temporal-window-width audit (tiled boxcar windows) of the endogeneity bias
    on the three benchmark systems ($T=400$, $\dt=0.005$, $R=8$, 25 repetitions).
    The relative coefficient error $|\hat c_x-c_x^*|/|c_x^*|$ of the temporal-window
    estimator stays near one across $L\in\{0.5,1,2,4,8\}$ for every system.
    The dashed vertical line marks the linearised relaxation time ($1$ for OU,
    $1/2$ for the double well and the multiplicative system).  The result supports
    the theorem's statement that changing the window width does not remove the
    coefficient-level bias.}
  \label{fig:temporalwidth}
\end{figure}

\subsection{Consistency of Spatial Kernels}
\label{sec:kernels}

The diagnosis suggests the cure: use a fixed, finite set of test functions
that the trajectory revisits an unbounded number of times as $T$ grows.
Weights that depend on the current state, $K_j(X_{t_n})$, do this.
(State-dependence does not remove the bias by itself, since a deterministic time
weight is equally $\mathcal F_{t_n}$-measurable; what matters is the
signal-to-noise ratio per regression equation, which grows with $T$ for kernels.)  We place $M$ centres $x_1,\ldots,x_M$ equally spaced
over the
observed state range and set
\begin{equation}
  K_j(x) = \exp\!\left(-\frac{(x-x_j)^2}{2h^2}\right),
  \qquad j=1,\ldots,M,
  \label{eq:kernel}
\end{equation}
with bandwidth $h>0$, taking $\psi_j(X_{t_n},t_n) = K_j(X_{t_n})$.  Each
kernel weights a step by how close the current state is to its centre:
steps near $x_j$ contribute fully, distant steps negligibly.

Because $K_j(X_{t_n})$ is a function of the current state alone, it is
$\mathcal{F}_{t_n}$-measurable, and the It\^{o} construction makes $\xi_n$
independent of $\mathcal{F}_{t_n}$.  The tower property then gives
$\E[K_j(X_{t_n})\sigma(X_{t_n})\xi_n]=0$ at \emph{every} step, so the
stochastic residual is zero-mean term by term, exactly as it was for
the temporal weight in \cref{sec:endogeneity}: this marginal property
alone is therefore \emph{not} what distinguishes the two constructions
(\cref{thm:unbiased} below is, in this narrow sense, a restatement of
part~(i) of \cref{thm:endogeneity}).

\begin{theorem}[Zero-mean spatial martingale residual]
\label{thm:unbiased}
Let $\psi_j(X_{t_n},t_n)=K_j(X_{t_n})$ be a spatial Gaussian kernel.
The stochastic contribution to the weak projection satisfies
$\E[Z_j^{\mathrm{spatial}}]=0$. Consequently,
$\E[B_j]=\E[(Ac^*)_j]$ for each $j$.
\end{theorem}

The theorem is deliberately an \emph{unconditional} mean statement.  It
must not be read as
$\E[B_j\mid\{X_{t_n}\}]=(Ac^*)_j$: conditioning on the entire observed
path reveals future states that contain earlier innovations $\xi_n$, so the
full-path conditioning sigma-algebra does not preserve the one-step
martingale-difference property.  The tower-property argument applies at
the one-step filtration level,
$\E[K_j(X_{t_n})\sigma(X_{t_n})\xi_n\mid\mathcal{F}_{t_n}]=0$, and is
given in \cref{app:endogeneity}. What \cref{thm:unbiased} does \emph{not}
by itself establish is that $\hat{c}\to c^*$, since OLS bias is governed by
$\E[A^\top Z]$, not by $\E[Z]$ alone (exactly the distinction the preceding
section turned on). The argument that closes this gap for spatial kernels
is not a claim that cross-step covariances vanish, but the sample-size
mechanism used in the proof of \cref{thm:consistency}: the $M$ kernel
centres are \emph{fixed} and do not grow with $T$, so as the ergodic
trajectory lengthens, each kernel is revisited an unboundedly growing
number of times. The quadratic variation of $Z_j^{\mathrm{spatial}}$,
viewed as a martingale in $n$, is
$\langle Z^{\mathrm{spatial}}_j\rangle_N = \dt\sum_n K_j^2(X_{t_n})
\sigma^2(X_{t_n}) = O(N\dt)=O(T)$, so the $L^2$ martingale strong law
gives $Z_j^{\mathrm{spatial}}/N\to0$ almost surely as $N\to\infty$ ---
regardless of the cross-step covariance structure. The signal $A_{jk}/T\to\bar A_{jk}\neq0$ grows linearly while the noise grows as
$\sqrt T$.  The fixed-width temporal windows of \cref{sec:endogeneity} lack this:
there the analogous per-window martingale never has its own sample size grow, so
its relative fluctuation never shrinks, and the nonzero short-lag population
contribution is reproduced across the growing number $M_{\mathrm{win}}=T/L$ of
windows.  (Growing the window is no remedy for mean-zero regressors; see the
single-wide-window check after \cref{prop:ou_temporal}.)

The bandwidth $h$ controls the resolution--variance tradeoff: smaller $h$
captures finer spatial structure but uses fewer observations per regression
row, increasing variance; larger $h$ reduces variance at the cost of spatial
resolution.  In our experiments we use $h=0.22$ for the OU and double-well
systems (state range $\approx\pm2.5$, overlap ratio $h/\Delta x_c\approx2.2$)
and $h=0.27$ for the multiplicative system, whose heavier-tailed trajectories
span a wider state range.

\subsection{The Drift Identification System}
\label{sec:drift}

With the spatial kernel test functions established, we can now
assemble the drift regression system explicitly.
Multiplying both sides of~\cref{eq:em} by $K_j(X_{t_n})$ and summing
over all steps:
\begin{equation}
  \underbrace{\sum_n K_j(X_{t_n})\,\Delta X_n}_{=:\,B_j}
  = \sum_n K_j(X_{t_n})\,b(X_{t_n})\,\dt
    + \sum_n K_j(X_{t_n})\,\sigma(X_{t_n})\,\xi_n\,\sqrt{\dt}.
  \label{eq:drift_proj}
\end{equation}

\paragraph{Interpretation of $A$ and $B$.}
The right-hand side of~\cref{eq:drift_proj} separates into a
drift contribution and a stochastic contribution.  Substituting
$b(x)=\Theta(x)c^*=\sum_k c_k^* f_k(x)$ into the drift term and
factoring:
\begin{equation}
  A_{jk} = \sum_{n=0}^{N-1} K_j(X_{t_n})\,f_k(X_{t_n})\,\dt.
  \label{eq:design}
\end{equation}
The matrix $A\in\R^{M\times K}$ is the \emph{Gram matrix} of the kernel
test functions against the library functions, integrated along the
trajectory with time step $\dt$.  Its $(j,k)$ entry approximates the
ergodic average $\bar{A}_{jk} = \int K_j(x) f_k(x)\,\mu(dx)$ up to
fluctuations that vanish by the ergodic theorem as $T\to\infty$
(see \cref{thm:consistency}).  Geometrically, $A_{jk}$ measures the
weighted co-occurrence of kernel $K_j$ and library term $f_k$ over the
observed state space: if the trajectory spends most of its time near
centre $x_j$ and library function $f_k$ is large there, then $A_{jk}$
will be large.

The vector $B\in\R^M$ is the \emph{kernel-weighted increment}:
\begin{equation}
  B_j := \sum_n K_j(X_{t_n})\,\Delta X_n.
  \label{eq:drift_system_B}
\end{equation}
Each $B_j$ aggregates all observed increments $\Delta X_n$, weighted
by how close the current state is to centre $x_j$.  Near centre $x_j$,
increments are driven primarily by $b(x_j)\dt$; far from $x_j$, those
increments receive negligible weight.  Taken together, the vector $B$
encodes how the process drifts near each kernel centre, averaged over
all visits.

Since the noise term has zero mean by Theorem~\ref{thm:unbiased}, taking
expectations gives $\E[B_j] = \sum_k c_k^* A_{jk}$, or in matrix form
$\E[B] = Ac^*$.  The resulting population-level linear system is
\begin{equation}
  B \approx Ac,
  \label{eq:drift_system}
\end{equation}
which can be solved by any standard linear regression method to recover
the drift coefficient vector $c^*$.  Because the noise term
$Z^{\mathrm{spatial}}$ is a martingale difference sequence (by
Theorem~\ref{thm:unbiased}), OLS applied to~\cref{eq:drift_system}
is strongly consistent by the ergodic convergence proved in
\cref{thm:consistency}: as $T\to\infty$, $\hat{c}\to c^*$ almost surely.

\subsection{Diffusion Identification via Quadratic Variation}
\label{sec:diffusion}

The drift system~\cref{eq:drift_system} exploits the first-order
statistics of the increments.  To identify the diffusion coefficient
$a(x)=\sigma(x)^2$, we exploit the \emph{second-order} statistics via
the quadratic variation.

\paragraph{Why squared increments estimate $a$.}
By the definition of the It\^{o} quadratic variation,
$[X]_t=\int_0^t a(X_s)\,ds$ almost surely: the accumulated squared
variation of the path is determined entirely by the diffusion coefficient
along the trajectory.  At the discrete level, squaring the
Euler--Maruyama increment exactly (keeping $\xi_n^2$, not replacing it by
its conditional mean) gives
\begin{equation}
  (\Delta X_n)^2
  = \bigl(b(X_{t_n})\,\dt + \sigma(X_{t_n})\,\xi_n\,\sqrt{\dt}\bigr)^2
  = a(X_{t_n})\,\xi_n^2\,\dt
    + 2\,b(X_{t_n})\,\sigma(X_{t_n})\,\xi_n\,\dt^{3/2}
    + b(X_{t_n})^2\,\dt^2.
  \label{eq:qv_pathwise}
\end{equation}
Writing $\xi_n^2 = 1 + (\xi_n^2-1)$ splits the leading term into its
conditional mean and a zero-mean fluctuation:
\begin{equation}
  (\Delta X_n)^2
  = a(X_{t_n})\,\dt
  + \underbrace{a(X_{t_n})\,(\xi_n^2-1)\,\dt}_{\text{Q.V.\ fluctuation}}
  + 2\,b(X_{t_n})\,\sigma(X_{t_n})\,\xi_n\,\dt^{3/2}
  + b(X_{t_n})^2\,\dt^2.
  \label{eq:qv_split}
\end{equation}
Since $\E[\xi_n^2-1\mid\mathcal{F}_{t_n}]=0$ and
$\E[\xi_n\mid\mathcal{F}_{t_n}]=0$, taking conditional expectations of
\cref{eq:qv_split} leaves only the first and last terms:
$\E[(\Delta X_n)^2 \mid X_{t_n}] = a(X_{t_n})\dt + O(\dt^2)$, exactly as
before. The middle two terms of \cref{eq:qv_split} do not bias $Q_j$ in
expectation, but, unlike the drift regression's noise term, which
contributes only to variance by design, the quadratic-variation
fluctuation term $a(X_{t_n})(\xi_n^2-1)\dt$ is \emph{also} $O(\dt)$, the
same order as the signal itself (whereas the drift-cross term
$2b\sigma\xi_n\dt^{3/2}$ remains sub-leading, $O(\dt^{3/2})$). It is
therefore the dominant source of finite-sample variance in the diffusion
estimator $\hat{d}$ and must be kept explicit in any variance accounting
for $Q_j$; we carry it through the remainder of this section and return
to its implication for the asymptotic variance of $\hat d$ at the end of
\cref{sec:bias}.

\paragraph{The diffusion regression system.}
Multiplying the squared increment by the kernel $K_j(X_{t_n})$ and
summing over steps:
\begin{equation}
  Q_j := \sum_n K_j(X_{t_n})\,(\Delta X_n)^2
  \;\longrightarrow\;
  \int_0^T K_j(X_t)\,a(X_t)\,dt
  = \sum_k d_k\,A_{jk}
  \quad\text{as }\dt\to0,
  \label{eq:diffusion_system}
\end{equation}
where in the last equality we used the library expansion
$a(x) = \Theta(x)d^* = \sum_k d_k^* f_k(x)$ and the definition
of $A_{jk}$ from~\cref{eq:design}.  This gives the second linear system
\begin{equation}
  Q \approx Ad,
\end{equation}
with exactly the same design matrix $A$ as the drift system.  The
shared $A$ is a consequence of using the same kernel test functions
and library for both identifications: the kernel-library Gram matrix
encodes the geometry of the trajectory in state space, and that geometry
is the same whether one is projecting the drift or the diffusion.
Practically, this means both $A$ and the kernel evaluations need to be
computed only once, halving the data-access cost.

For a multi-dimensional state space $\R^d$, each entry $a_{pq}$ of the
diffusion tensor has its own sparse coefficient vector $d^{(pq)}$, and the
identification system is $Q^{(pq)}\approx Ad^{(pq)}$ where
$Q_j^{(pq)} = \sum_n K_j(X_{t_n})(\Delta X_n)_p(\Delta X_n)_q$.  This
requires $d(d+1)/2$ linear solves, all sharing the design matrix $A$.

\subsection{Finite-Time-Step Bias Correction}
\label{sec:bias}

The diffusion identification system~\cref{eq:diffusion_system} is exact
only as $\dt\to0$: at any finite step size, the squared increment
contains a systematic positive contribution from the drift that inflates
the estimate of $a$.  We now quantify and remove this bias.

\paragraph{Source of the bias.}
Starting from the exact expansion~\cref{eq:qv_split}, which we repeat
here with its four terms labelled:
\begin{equation}
  (\Delta X_n)^2 = \underbrace{a(X_{t_n})\,\dt}_{\text{signal}}
  + \underbrace{a(X_{t_n})(\xi_n^2-1)\,\dt}_{\text{Q.V.\ fluctuation}}
  + \underbrace{2\,b(X_{t_n})\,\sigma(X_{t_n})\,\xi_n\,\dt^{3/2}}_{
    \text{drift--diffusion cross term}}
  + \underbrace{b(X_{t_n})^2\,\dt^2}_{\text{drift-squared bias}}.
  \label{eq:bias_expand}
\end{equation}
The four terms have distinct origins and distinct statistical properties.
The first term is the signal we want. The second (Q.V.\ fluctuation) and
third (drift--diffusion cross) terms have zero conditional mean, by
$\E[\xi_n^2-1\mid\mathcal{F}_{t_n}]=0$ and
$\E[\xi_n\mid\mathcal{F}_{t_n}]=0$ respectively, so neither contributes
to the \emph{expectation} of $Q_j$; they do, however, contribute to its
\emph{variance}: the Q.V.\ fluctuation term dominates, being $O(\dt)$
rather than $O(\dt^{3/2})$, and we return to this below. The fourth
term, $b(X_{t_n})^2\dt^2$, is the only one with nonzero conditional
mean: it is \emph{always non-negative} and therefore introduces a
systematic upward bias into $Q_j$. Multiplying by $K_j$ and summing, the
expected value of $Q_j$ is
\begin{equation}
  \E[Q_j] = \sum_k d_k\,A_{jk}
  + \underbrace{\sum_n \E\bigl[K_j(X_n)\,b(X_n)^2\bigr]\,\dt^2}_{
    \text{drift-squared bias}}.
  \label{eq:bias}
\end{equation}
This bias is $O(\dt^2)$ per step and $O(N\dt^2) = O(T\dt)$ in total,
so it shrinks as $\dt\to0$ for fixed $T$, but is non-negligible at
practical time steps (e.g., $\dt=0.002$, $T=200$).

\paragraph{Two-step correction procedure.}
We remove the bias by leveraging the drift estimate already obtained
from the first regression system.  The procedure is:
\begin{enumerate}[label=\textup{(\arabic*)}]
  \item \textbf{Estimate the drift.}  Solve $B\approx Ac$ to obtain
    $\hat{c}$, and form the estimated drift function
    $\hat{b}(x) = \Theta(x)\hat{c}$.
  \item \textbf{Subtract the bias.}  Form the corrected quadratic
    variation vector
    \begin{equation}
      Q_j^{\mathrm{corr}} = Q_j
      - \sum_n K_j(X_n)\,\hat{b}(X_n)^2\,\dt^2,
      \label{eq:bias_corr}
    \end{equation}
    and solve the corrected system $Q^{\mathrm{corr}}\approx Ad$.
\end{enumerate}
\paragraph{Residual error after correction.}
Writing $\Delta b(x) := \hat{b}(x)-b(x) = \Theta(x)(\hat{c}-c^*)$, the
correction residual is
\begin{equation}
  \sum_n K_j(X_n)\bigl[b(X_n)^2-\hat{b}(X_n)^2\bigr]\dt^2
  = -\underbrace{2\sum_n K_j(X_n)\,b(X_n)\,\Delta b(X_n)\,\dt^2}_{
    \text{linear in }\Delta b}
  - \underbrace{\sum_n K_j(X_n)\,\Delta b(X_n)^2\,\dt^2}_{
    \text{quadratic in }\Delta b},
  \label{eq:correction_residual}
\end{equation}
using $\hat{b}^2-b^2=2b\Delta b+\Delta b^2$, i.e.\
$b^2-\hat{b}^2=-2b\Delta b-\Delta b^2$ (the cross term $-2b\Delta b$ has no definite sign).
Both terms are present; the correction is not purely quadratic in the
coefficient error, and we do not claim it is. Both residual terms
in \cref{eq:correction_residual} share the \emph{same} $\dt^2$ kernel-weighted
summation structure as the uncorrected drift-squared bias
of \cref{eq:bias}, so this factor cancels rather than adding an extra power
of $\dt$. Treating $b(X_n)$ as approximately constant in sign and
magnitude over the effective support of $K_j$, the dominant (linear) residual
term is approximately $2\,\Delta b/b$ times the uncorrected bias, i.e.\ the
correction shrinks the bias by a factor of order twice the relative
drift-coefficient error, $O(\|\hat c-c^*\|/\|c^*\|)=O(T^{-1/2})$ by
Theorem~\ref{thm:clt}, with no additional factor of $\dt$. The quadratic
term in \cref{eq:correction_residual} is smaller still, of relative order
$(\Delta b/b)^2=O(T^{-1})$. At $\dt=0.002$ for the multiplicative-diffusion
system (\cref{sec:results_mult}), the uncorrected relative bias in
$\hat d_{x^2}$ is 3.5\% against a drift-coefficient error of 3.9\%; the
predicted improvement factor $2\times3.9\%\approx7.8\%$ gives an expected
corrected bias of about $7.8\%\times3.5\%\approx0.3\%$, consistent with the
$0.5\%$ observed after correction (\cref{tab:bias_correction}).

\paragraph{Variance contribution of the quadratic-variation fluctuation.}
Returning to the fluctuation term identified in
\cref{eq:qv_split,eq:bias_expand}: because $a(X_{t_n})(\xi_n^2-1)\dt$ is
$O(\dt)$, the same order as the signal, it is the leading
contribution to $\Var[Q_j]$, dominating the $O(\dt^{3/2})$
drift--diffusion cross term. For Gaussian $\xi_n$,
$\Var[\xi_n^2-1]=2$, so this term contributes
$2\sum_n K_j^2(X_n)a(X_n)^2\dt^2$ to $\Var[Q_j]$, an $O(N\dt^2)=O(T\dt)$
quantity and the dominant finite-sample fluctuation in the diffusion target.
This establishes the source of the diffusion estimator's variance but does not
by itself provide a coefficient-level central limit theorem; a full diffusion
CLT incorporating this term exactly is left for future work.

\section{Algorithm}
\label{sec:algorithm}

\subsection{From the Weak Systems to Algorithm~\ref{alg:main}}
\label{sec:assembly}

The previous section produced three ingredients: the drift system
$B\approx Ac$~\cref{eq:drift_system}, the diffusion system $Q\approx Ad$
sharing the same design matrix~\cref{eq:diffusion_system}, and the
finite-step bias correction~\cref{eq:bias_corr}.  We now describe precisely
how these are assembled into the end-to-end procedure of
Algorithm~\ref{alg:main}, since the order of operations and the pooling
across trajectories are what make the method usable in practice.

\paragraph{Step 1: one kernel pass produces all three statistics.}
For a single trajectory, every quantity the method needs is a kernel-weighted
sum over the same left-endpoint states $\{X_{t_n}\}$.  Evaluating the kernel
matrix $K_{jn}=K_j(X_{t_n})$ and the library matrix $\Theta_{kn}=f_k(X_{t_n})$
once, the three statistics are read off as
\begin{equation}
  A_{jk} = \dt\sum_n K_{jn}\,\Theta_{kn},
  \qquad
  B_j = \sum_n K_{jn}\,\Delta X_n,
  \qquad
  Q_j = \sum_n K_{jn}\,(\Delta X_n)^2 .
  \label{eq:assembly_stats}
\end{equation}
Because $A$, $B$, and $Q$ all reuse the same kernel evaluations, a single
sweep over the trajectory, at cost $O(MNK)$, builds everything. No derivative
estimate is formed, so there is no per-point $1/\dt$ amplification, and the
averaging over visits is built into the statistic.

\paragraph{Step 2: pooling across trajectories by stacking.}
With $R$ independent trajectories $\{X_{t_n}^{(r)}\}$, each contributes its
own $(A^{(r)},B^{(r)},Q^{(r)})$ from~\cref{eq:assembly_stats}.  Rather than
average the per-trajectory design matrices, we \emph{stack} the rows,
\begin{equation}
  A_{\mathrm{stack}} =
  \begin{bmatrix} A^{(1)}\\[-1pt]\vdots\\[-1pt] A^{(R)}\end{bmatrix}
  \in\R^{RM\times K},
  \quad
  B_{\mathrm{stack}} =
  \begin{bmatrix} B^{(1)}\\[-1pt]\vdots\\[-1pt] B^{(R)}\end{bmatrix},
  \quad
  Q_{\mathrm{stack}} =
  \begin{bmatrix} Q^{(1)}\\[-1pt]\vdots\\[-1pt] Q^{(R)}\end{bmatrix}
  \in\R^{RM}.
  \label{eq:stacking}
\end{equation}
Stacking rather than averaging keeps every trajectory's projected equations
as separate constraints, which (i)~increases the effective number of
regression rows from $M$ to $RM$, improving conditioning when individual
trajectories explore only part of the state space, and (ii)~lets the
trajectory index serve as a natural grouping for cross-validation in
Step~4, so that autocorrelated rows from the same path are never split
across folds.

\paragraph{Step 3: column normalisation.}
The library functions $f_k$ have heterogeneous scales (e.g.\ $1$, $x$, $x^3$
differ by orders of magnitude over the observed range), which would let the
$\ell_1$ penalty in Step~4 act unequally across coefficients.  We therefore
rescale each column of $A_{\mathrm{stack}}$ to unit norm, solve the
regression in the normalised coordinates, and undo the scaling on the
recovered coefficients.  This makes the sparsity penalty scale-invariant and
is what allows a single $\lambda$ to be meaningful across the whole library.

\paragraph{Step 4: drift first, then bias-corrected diffusion.}
The drift and diffusion systems share $A_{\mathrm{stack}}$ but must be solved
in a definite order, because the finite-step correction to the diffusion
target~\cref{eq:bias_corr} requires the drift estimate.  Algorithm~\ref{alg:main}
therefore (a)~solves $B_{\mathrm{stack}}\approx A_{\mathrm{stack}}c$ for the
drift support $S_b$ and coefficients $\hat{c}$; (b)~forms the corrected
quadratic-variation vector
$Q_j^{\mathrm{corr}} = Q_j - \sum_n K_{jn}\,\hat{b}(X_{t_n})^2\,\dt^2$ using
the just-recovered $\hat{b}=\Theta\hat{c}$; and (c)~solves
$Q^{\mathrm{corr}}_{\mathrm{stack}}\approx A_{\mathrm{stack}}d$ for the
diffusion support $S_a$ and coefficients $\hat{d}$.  Both regressions use the
identical sparse-selection pipeline described next, so the diffusion solve
adds no new machinery beyond the correction step.

\subsection{Sparse Regression and Model Selection}

After building $A_{\mathrm{stack}}$ and $B_{\mathrm{stack}}$ by stacking
all $R$ trajectory contributions, we solve the elastic-net
problem~\cite{tibshirani1996,zou2005}
\begin{equation}
  \hat{c} = \arg\min_{c\in\R^K}
  \|A_{\mathrm{stack}}\,c - B_{\mathrm{stack}}\|_2^2
  + \lambda\rho\|c\|_1 + \tfrac12\lambda(1-\rho)\|c\|_2^2,
  \label{eq:lasso}
\end{equation}
and similarly for the diffusion system with $Q^{\mathrm{corr}}_{\mathrm{stack}}$.
Setting the mixing parameter $\rho=1$ recovers the plain LASSO, and we
search an eight-point grid of $\rho$ values ranging from 0.4
(ridge-leaning) to 1.0 (pure LASSO)\footnote{$\rho\in\{0.4,0.6,0.75,0.85,
0.93,0.97,0.99,1.0\}$.} jointly with
$\lambda$ rather than fixing $\rho=1$ a priori; both are chosen by $K$-fold
cross-validation (ElasticNetCV) with folds partitioned by \emph{trajectory
index} rather than by time step, since time-based
partitioning would leak temporal autocorrelation between folds and distort
model selection~\cite{mangan2017}.  The $\ell_2$ term is not cosmetic: the
monomial library entries are correlated under a near-Gaussian stationary
density regardless of sample size (e.g.\ $\mathrm{corr}(x,x^3)\approx0.77$
for any Gaussian-distributed $x$), and plain LASSO's selection becomes
unstable, arbitrarily splitting weight between or dropping one of two
simultaneously active correlated terms, precisely in the cubic-drift
systems of \cref{sec:benchmark_table} where $x$ and $x^3$ are both active.
Ridge shrinkage is the standard remedy for this kind of design
collinearity~\cite{zou2005} and $\rho=1$ remains available to
cross-validation whenever the data do not need it, so this is a strict
generalisation of the original LASSO step, not a different method. After
ElasticNetCV selects an initial support, OLS debiasing removes the
shrinkage introduced by the penalty, and iterated Sequential Thresholded
Least Squares (STLSQ)~\cite{brunton2016,kaheman2020} prunes residual
near-zero coefficients, so the reported support and coefficients are exact
and sparse in the original monomial coordinates regardless of which $\rho$
cross-validation selects internally.

The final support sets $S_b,S_a\subseteq\{1,\ldots,K\}$ yield the
identified generator
\begin{equation}
  \hat{\Lop}f = \hat{b}(x)f' + \tfrac{1}{2}\hat{a}(x)f'',
  \qquad
  \hat{b}(x) = \sum_{k\in S_b}\hat{c}_k f_k(x),\quad
  \hat{a}(x) = \sum_{k\in S_a}\hat{d}_k f_k(x),
  \label{eq:identified}
\end{equation}
which can be used directly for spectral analysis, stationary density
computation via~\cref{eq:fp}, escape rate estimation
via~\cref{eq:kramers_dw}, and analytical perturbation theory.

\subsection{Complete Pipeline}

The complete pipeline is summarised in Algorithm~\ref{alg:main}.  The
dominant cost is building $A_{\mathrm{stack}}$: $O(MNK)$ per trajectory,
linear in all three dimensions.  For our experimental settings
($M=50$, $N=50{,}000$, $K=5$, $R=120$), the full pipeline completes in
under two minutes on a standard multi-core workstation.  Formal proofs of
consistency (Theorem~\ref{thm:consistency}), asymptotic normality
(Theorem~\ref{thm:clt}), and noise robustness (Theorem~\ref{thm:noise}),
together with the spectral gap sensitivity analysis and identifiability
conditions, are collected in \cref{sec:theory}.

\begin{algorithm}
\caption{Weak Stochastic Generator Recovery (Spatial Gaussian Kernels)}
\label{alg:main}
\begin{algorithmic}[1]
\Require Trajectories $\{X_{t_n}^{(r)}\}_{n,r}$, library $\Theta(x)$,
         centres $\{x_j\}_{j=1}^M$, bandwidth $h$, time step $\dt$
\Ensure  Identified generator $\hat{\Lop}$ via $\hat{b}(x)$, $\hat{a}(x)$
\State Evaluate $K_j(X_{t_n})=\exp\!\bigl(-(X_{t_n}-x_j)^2/2h^2\bigr)$
       and $\Theta(X_{t_n})$ at all left-endpoint states.
\For{each trajectory $r=1,\ldots,R$}
  \State $A_{jk}^{(r)}
         \gets \sum_n K_j(X_{t_n}^{(r)})\,f_k(X_{t_n}^{(r)})\,\dt$
  \State $B_j^{(r)}
         \gets \sum_n K_j(X_{t_n}^{(r)})\,\Delta X_n^{(r)}$
  \State $Q_j^{(r)}
         \gets \sum_n K_j(X_{t_n}^{(r)})\,(\Delta X_n^{(r)})^2$
\EndFor
\State Stack: $A_{\mathrm{stack}}$, $B_{\mathrm{stack}}$,
       $Q_{\mathrm{stack}}$.  Normalise columns of $A_{\mathrm{stack}}$.
\State Solve $B_{\mathrm{stack}}=A_{\mathrm{stack}}\,c$ via ElasticNetCV
       (trajectory-grouped $K$-fold when $R>1$; contiguous time-block
       $K$-fold when $R=1$, \cref{sec:single_traj}; $\rho=1$ recovers
       plain LASSO)
       $+$ OLS debias $+$ STLSQ\@.
       Obtain $\hat{c}$, support $S_b$.
\State Compute $Q_j^{\mathrm{corr}}\gets Q_j
       -\sum_n K_j(X_n)\hat{b}(X_n)^2\dt^2$.
\State Solve $Q_{\mathrm{stack}}^{\mathrm{corr}}=A_{\mathrm{stack}}\,d$
       with the same pipeline.  Obtain $\hat{d}$, support $S_a$.
\State \textbf{return} $\hat{\Lop}$ via \cref{eq:identified}.
\end{algorithmic}
\end{algorithm}

\paragraph{Noisy-observation variant.}
For observed states $Y_n=X_{t_n}+\eta_n$ with i.i.d. Gaussian noise of
known variance $\nu=\sigma_\eta^2$, Algorithm~\ref{alg:main} is applied to
the noisy observations after replacing the drift and diffusion targets by the
exact corrected statistics \cref{eq:noise_drift_exact,eq:noise_diffusion_exact}
and the finite-step diffusion target by \cref{eq:noise_diffusion_full_correction}.
The observed design matrix is
$\widetilde A_{jk}=\dt\sum_n K_j(Y_n)f_k(Y_n)$.  Thus the clean-data
algorithm is recovered as the special case $\nu=0$, while the Gaussian
noise-aware variant remains a single shared-design-matrix procedure.

\section{Simulation Setup}
\label{sec:implementation}

All three benchmark SDEs are simulated using the Euler--Maruyama
scheme~\cite{higham2001} with $\dt=0.002$ over horizon $T=100$, giving
$N=50{,}000$ observations per trajectory and $R=120$ independent
realisations per system.  Initial conditions are drawn uniformly from
$[-3,3]$.  The polynomial library $\Theta(x)=[1,x,x^2,x^3,x^4]$ ($K=5$)
is used throughout; columns of $A_{\mathrm{stack}}$ are normalised to
unit $\ell_2$ norm before regression.  For the OU and double-well systems,
$M=50$ kernel centres are placed uniformly on $[-2.5,2.5]$ with $h=0.22$;
for the multiplicative system, centres span $[-2.8,2.8]$ with $h=0.27$.
ElasticNetCV searches 60 logarithmically spaced values
$\lambda\in[10^{-8},10^{-0.5}]$ jointly with mixing parameter
$\rho\in\{0.4,0.6,0.75,0.85,0.93,0.97,0.99,1.0\}$ over five
trajectory-grouped folds, followed by OLS debiasing and at most 20 STLSQ
iterations with relative threshold 0.25 (0.30 for the multiplicative
system).

\section{Dynamical Systems Applications}
\label{sec:results}

\subsection{Ornstein--Uhlenbeck Process: Spectral Gap Recovery}

The Ornstein--Uhlenbeck process
\begin{equation}
  dX_t = -\theta X_t\,dt + \sigma_0\,dW_t, \qquad
  \theta=1.0,\quad\sigma_0=0.7,
  \label{eq:ou}
\end{equation}
is the canonical test case for spectral gap estimation.  The spectral gap
equals $\lambda_1 = \theta = 1$, the autocorrelation decays as
$C(\tau)=e^{-\theta\tau}=e^{-\tau}$, and the stationary distribution is
$\pi_{\mathrm{OU}}\sim\mathcal{N}(0,\sigma_0^2/2\theta)=\mathcal{N}(0,0.245)$.

Algorithm~\ref{alg:main} recovers $\hat{c}_x=-0.980$, giving a spectral
gap estimate $\hat{\lambda}_1 = 0.980$ with error~2.0\%.  All other drift
coefficients are set to exactly zero by the sparse-selection pipeline; the
diffusion estimate is $\hat{d}_1=0.490$ (error~0.0\%).

\paragraph{Spectral gap error and CLT rate.}
By Theorem~\ref{thm:clt}, the standard error of $\hat{\lambda}_1$ decays
as $1/\!\sqrt{T}$.  At $T=100$ with $R=120$ trajectories ($T_{\mathrm{eff}}=12{,}000$),
the 2.0\% error is consistent with the predicted $1/\!\sqrt{T_{\mathrm{eff}}}$
scaling, as confirmed by Figure~\ref{fig:convergence}.  The relaxation
timescale $\tau_{\mathrm{relax}}=1/\lambda_1=1.0$
is recovered as $\hat{\tau}_{\mathrm{relax}}=1/0.980=1.020$, a
2.0\% overestimate.

The ElasticNetCV regularisation path (Figure~\ref{fig:lasso}, top-left)
shows a sharp elbow at $\alpha^*\approx1.6\times10^{-3}$ (CV selects
$\rho^*=1.00$, i.e.\ plain LASSO for this sub-problem), where the CV MSE
reaches its minimum.  The clean identification of a one-term drift from a
five-term library without manual thresholding demonstrates that the
grouped CV scheme correctly selects the true sparsity level.

\subsection{Double-Well Langevin System: Metastability and Kramers Rates}
\label{sec:results_dw}

The double-well system
\begin{equation}
  dX_t = (X_t - X_t^3)\,dt + \sigma_0\,dW_t, \qquad \sigma_0=0.5,
  \label{eq:doublewell}
\end{equation}
is the canonical model of metastable stochastic dynamics.  The potential
$V(x)=-x^2/2+x^4/4$ has stable equilibria at $x=\pm1$, an unstable
fixed point at $x=0$, and barrier height $\Delta V = 1/4$.  The Kramers
escape time~\cref{eq:kramers_dw} is
$\tau_{\mathrm{Kramers}} = \sqrt{2}\,\pi\exp(1/2\sigma_0^2) = \sqrt{2}\,\pi\exp(2) \approx 32.8$
time units.  The spectral gap $\lambda_1 \approx 1/\tau_{\mathrm{Kramers}}$
is exponentially sensitive to the barrier height $\Delta V$ and to the
noise amplitude $\sigma_0$.

\paragraph{Identified generator.}
Algorithm~\ref{alg:main} recovers $\hat{c}_x=+0.988$ (error~1.2\%) and
$\hat{c}_{x^3}=-0.976$ (error~2.4\%), with all other coefficients exactly
zero.  The diffusion estimate is $\hat{d}_1=0.250$ (error~0.1\%).

\paragraph{Kramers escape rate.}
The identified potential is
$\hat{V}(x) = -0.988x^2/2 + 0.976x^4/4$, with minima at
$\hat{x}_{\min} = \pm\sqrt{0.988/0.976} \approx \pm1.006$
and barrier height $\hat{\Delta V} \approx 0.250$.
The estimated Kramers time, from the same formula~\cref{eq:kramers_dw}
using the estimated diffusion level $\hat\sigma_0^2=\hat d_1=0.250$ and the
curvatures of $\hat V$ ($\hat V''(0)=-0.988$, $\hat V''(\pm\hat x_{\min})=1.976$), is
\[
  \hat{\tau}_{\mathrm{Kramers}}
  = \frac{2\pi}{\sqrt{0.988\times1.976}}
    \exp\!\Bigl(\frac{2\hat{\Delta V}}{\hat\sigma_0^2}\Bigr)
  \approx 4.497\,e^{2.0003} \approx 33.2,
\]
compared to the true value $\sqrt2\,\pi e^{2}\approx32.8$ of~\cref{eq:kramers_dw}.
The relative error is approximately 1.2\%.  The exponent is essentially
exact (the 1.2\% and 2.4\% errors in $\hat c_x$ and $\hat c_{x^3}$ act in
compensating directions for the barrier height, which deviates by only
0.015\%), so the residual error is carried by the curvature prefactor and
is comparable to the individual coefficient errors.

\paragraph{Dynamical implications.}
The bimodal stationary density with peaks at $x\approx\pm1$ is
reproduced with total variation $\mathrm{TV}=0.0092$
(Figure~\ref{fig:density}, centre).  The autocorrelation function
(Figure~\ref{fig:autocorr}, centre), shown over a $1.2$-second lag
window (short relative to the Kramers timescale of $\approx33$~s, so
only the fast intra-well relaxation is visible) tracks the true
system closely (true and recovered ACF both $\approx0.87$--$0.89$ at the
end of the window), confirming that the two-term cubic drift polynomial
correctly encodes the local potential curvature.  Probing the slower
inter-well mixing regime would require a correspondingly longer lag
window.

\subsection{Multiplicative Diffusion: Position-Dependent Relaxation}
\label{sec:results_mult}

The multiplicative diffusion system
\begin{equation}
  dX_t = -2X_t\,dt + \tfrac{1}{2}\sqrt{1+X_t^2}\,dW_t,
  \label{eq:mult}
\end{equation}
has state-dependent diffusion $a(x)=\tfrac{1}{4}(1+x^2)$ and is the most
demanding benchmark because it tests the method's ability to recover
position-dependent relaxation timescales, not merely the global spectral gap.

\paragraph{Local relaxation rate.}
The local relaxation rate at position $x$ is
\[
  \gamma(x) = \frac{|b'(x)|}{a(x)}
            = \frac{2}{\tfrac{1}{4}(1+x^2)}
            = \frac{8}{1+x^2}.
\]
Near the origin ($x\approx0$) the relaxation is fast ($\gamma(0)=8$);
at $x=\pm2$ it is four times slower ($\gamma(\pm2)=8/5$).  Recovering
the $x^2$ coefficient of $a(x)$ is therefore essential for capturing
the correct spatial variation of relaxation rates.

\paragraph{Necessity of bias correction.}
In a dedicated run isolating the correction
(Table~\ref{tab:bias_correction}; an independent realisation, so its
values differ slightly from the main-pipeline Table~\ref{tab:coefficients}),
the OLS estimate without the finite-step bias correction~\cref{eq:bias_corr} gives $\hat{d}_{x^2}\approx0.259$ (3.5\% error), which
overestimates the rate of diffusion growth with $|x|$ and compresses the
local relaxation timescale near $x=\pm2$ by the same proportion.  After
the two-step correction using the main-pipeline drift estimate
$\hat{b}(x)=-1.922x$ (error~3.9\%), the bias-corrected estimates of
Table~\ref{tab:coefficients} are $\hat{d}_{x^2}=0.251$ (0.6\% error) and
$\hat{d}_1=0.250$ (0.0\% error), giving recovered local rates
$\hat{\gamma}(x)=2\times1.922/[{\frac{1}{4}(0.250+0.251x^2)}]$; the
independent dedicated run gives $0.249$ (0.5\%) and $0.250$ (0.1\%).
This is dynamically accurate across the full observed state range.
Table~\ref{tab:bias_correction} provides a paired comparison of the
before- and after-correction coefficients for both the constant and
quadratic terms.

\paragraph{Total variation and autocorrelation.}
The stationary density has heavier tails than Gaussian and is reproduced
with $\mathrm{TV}=0.0093$.  The autocorrelation function matches the true
system across the full range of lags shown in Figure~\ref{fig:autocorr},
confirming that the bias-corrected estimate of $a(x)$ correctly captures
the spatially varying mixing rate.

\subsection{Function Recovery and Regularisation Paths}

Figure~\ref{fig:recovery} shows the recovered drift and diffusion functions
plotted against ground truth for all three systems over the range $[-2.5, 2.5]$.
Mean relative coefficient errors are 2.0\%, 1.2--2.4\%, and 3.9\% for the
three drifts, and 0.0\%, 0.1\%, and at most 0.6\% (on the active $x^2$
term) for the three diffusion functions.
Diffusion errors are uniformly lower than drift errors, consistent with
the higher effective signal-to-noise ratio of the quadratic variation
estimator relative to the increment estimator.

\begin{figure}[tp]
  \centering
  \includegraphics[width=\textwidth]{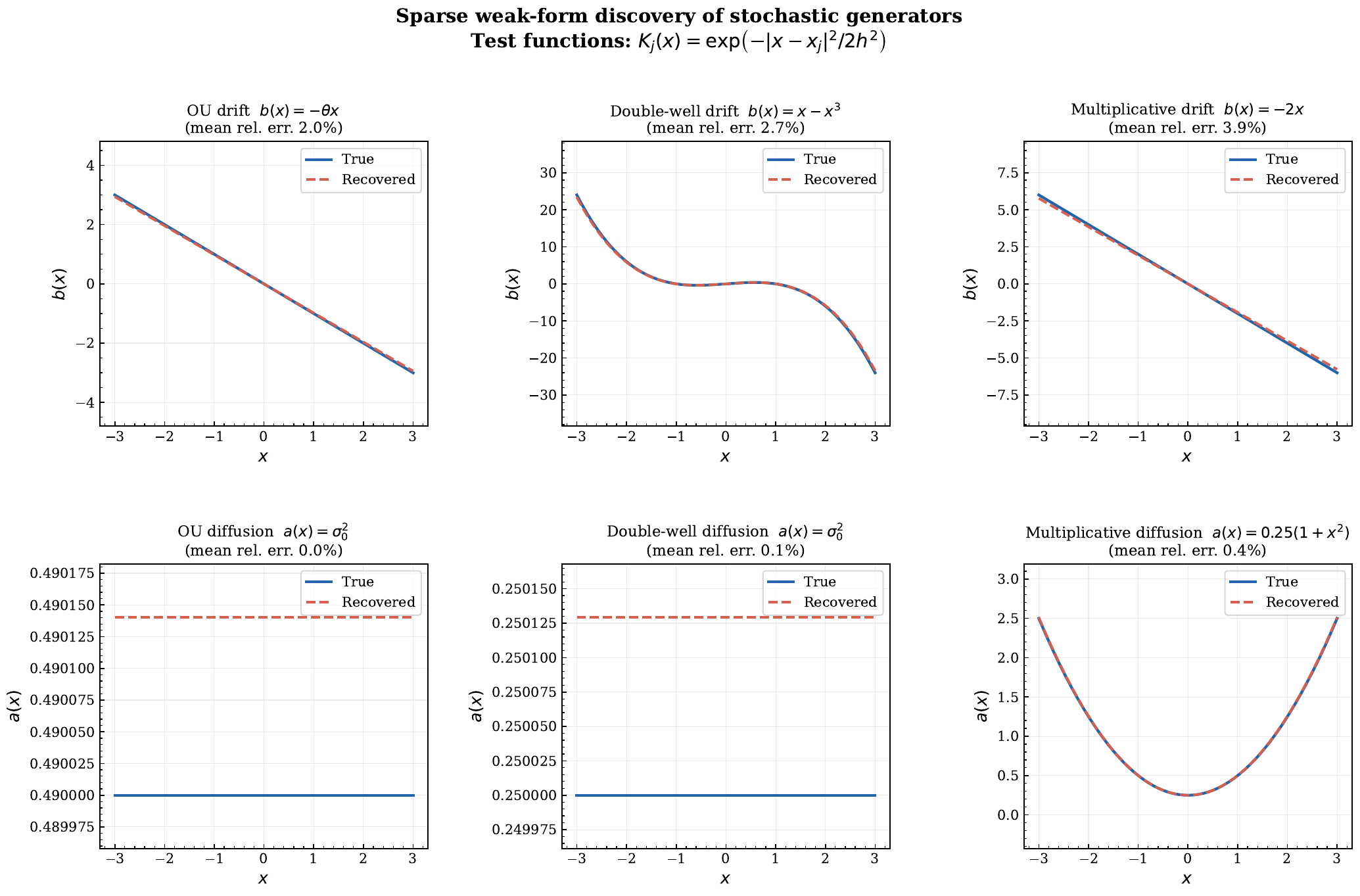}
  \caption{Recovered vs.\ true drift and diffusion functions for all three
    benchmark systems (evaluated on $[-2.5,2.5]$).
    Blue solid: ground truth; red dashed: estimates from
    Algorithm~\ref{alg:main}; shaded region: pointwise discrepancy.
    Annotated errors are the pointwise mean relative error of the
    reconstructed function over $[-2.5,2.5]$, which for the two-term
    double-well drift and multiplicative diffusion blends both active
    coefficients into a single function-level number and so differs from
    the per-coefficient relative errors of Table~\ref{tab:coefficients}.
    \textit{Top row (diffusion):} OU, $a(x)=0.490$, function err.\ 0.0\%;
    double-well, $a(x)=0.250$, 0.1\%;
    multiplicative, $a(x)=(1+x^2)/4$, 0.4\%
    (after bias correction; 0.0\% and 0.6\% on the individual constant and
    $x^2$ coefficients, Table~\ref{tab:coefficients}).
    \textit{Bottom row (drift):} OU, $b(x)=-\theta x$, 2.0\%;
    double-well, $b(x)=x-x^3$, 2.7\% (1.2\% and 2.4\% on the individual
    $x$ and $x^3$ coefficients);
    multiplicative, $b(x)=-2x$, 3.9\%.
    All recovered curves are visually indistinguishable from ground
    truth at the displayed scale.}
  \label{fig:recovery}
\end{figure}

Figure~\ref{fig:lasso} shows the ElasticNetCV regularisation paths for all
six sub-problems.  In four of the six panels the CV error exhibits a sharp
elbow separating the correct sparse identification regime from the
over-regularised regime, with the selected $\alpha^*$ falling at or just
past the elbow; for the OU and double-well diffusion sub-problems the
selected $\alpha^*$ instead sits at the lower edge of the search grid
($10^{-8}$), reflecting that these single-term diffusion fits are
essentially unregularised at any grid point rather than exhibiting a
distinct interior elbow. CV selects
the pure-LASSO end of the mixing-parameter grid ($\rho^*=1.00$) for five
of the six sub-problems and $\rho^*=0.97$ for the OU diffusion, consistent
with this three-system benchmark's library not being severely collinear;
\cref{sec:benchmark_table} is where the ridge component of the elastic net
does most of its work.

\begin{figure}[tp]
  \centering
  \includegraphics[width=\textwidth]{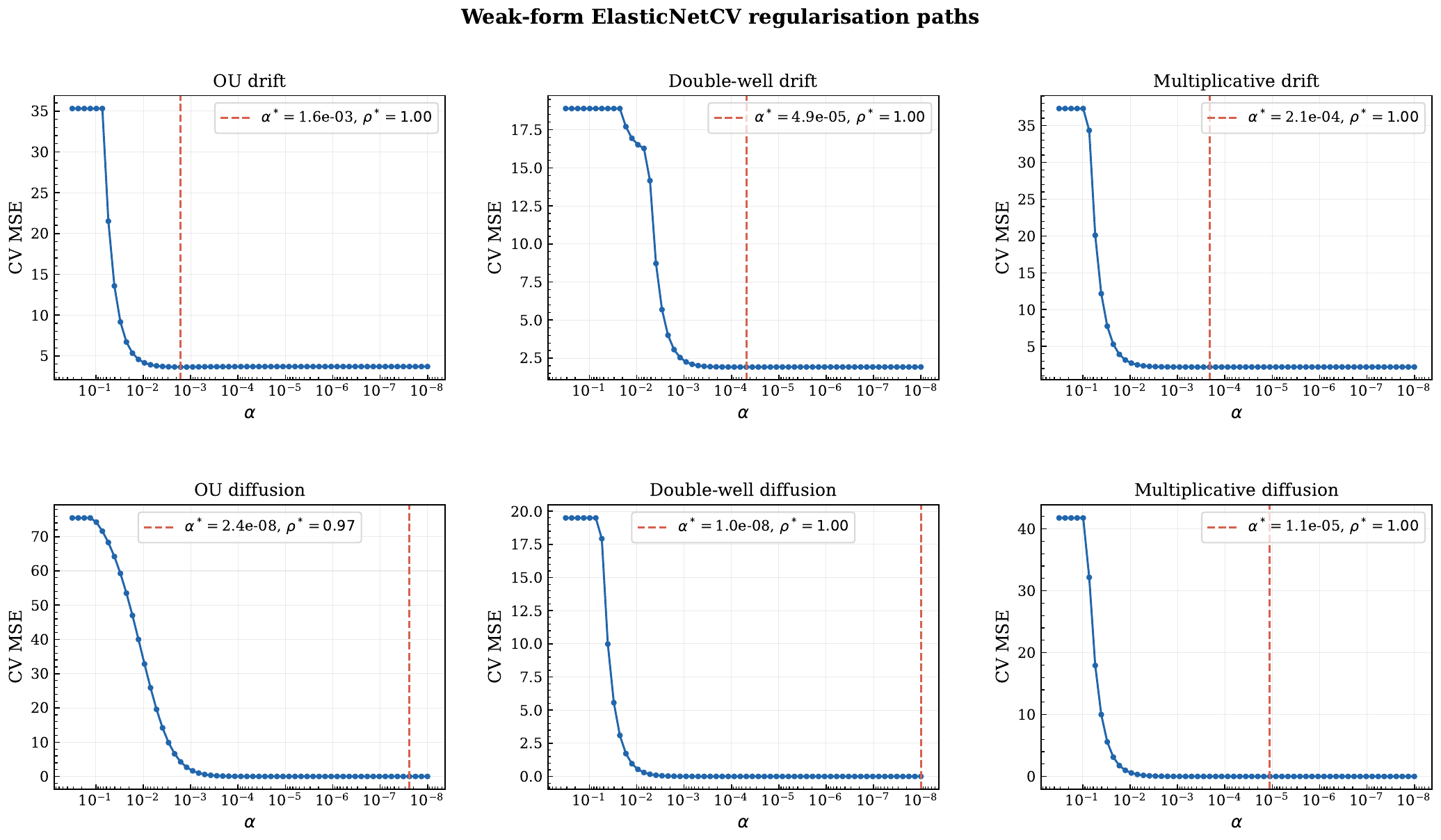}
  \caption{ElasticNetCV regularisation paths for all six sub-problems.  Each
    panel plots the mean cross-validated MSE, at the CV-selected mixing
    parameter $\rho^*$, over five trajectory folds as
    a function of regularisation strength $\alpha$ (decreasing left to right).
    Red dashed vertical lines mark the selected $\alpha^*$ and $\rho^*$.
    \textit{Top row (drift):} OU ($\alpha^*\approx1.6\times10^{-3}$),
    double-well ($\approx4.9\times10^{-5}$),
    multiplicative ($\approx2.1\times10^{-4}$).
    \textit{Bottom row (diffusion):} OU ($\approx2.4\times10^{-8}$),
    double-well ($\approx1.0\times10^{-8}$, the lower edge of the search
    grid), multiplicative ($\approx1.1\times10^{-5}$).
    The drift panels and the multiplicative-diffusion panel show a sharp
    interior elbow; the OU and double-well diffusion panels select
    $\alpha^*$ at or near the grid floor, consistent with these
    single-term diffusion fits needing little to no regularisation.}
  \label{fig:lasso}
\end{figure}

\subsection{Stationary Density Validation}

Figure~\ref{fig:density} shows stationary densities computed analytically
via~\cref{eq:fp} for both the true and recovered generators, isolating
coefficient error without Monte Carlo variance contamination.  Total
variation distances are $\mathrm{TV}=0.0050$ (OU), $0.0092$ (double-well),
and $0.0093$ (multiplicative), all below 0.01.

\begin{figure}[tp]
  \centering
  \includegraphics[width=\textwidth]{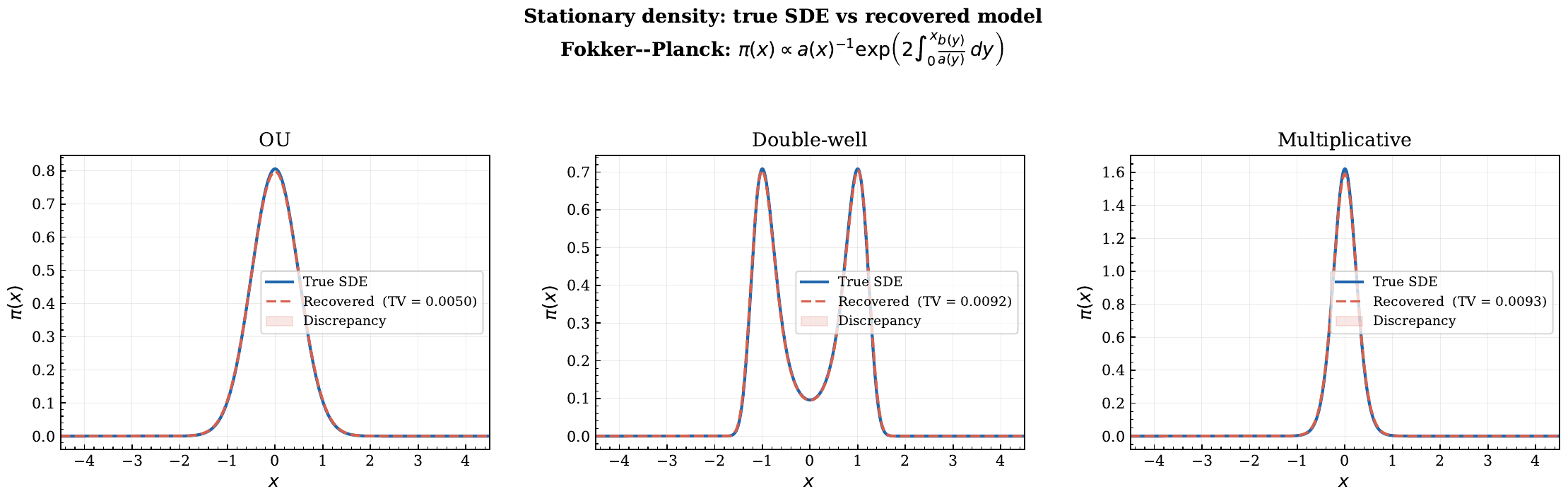}
  \caption{Stationary density: true SDE vs.\ recovered model.  Densities
    computed analytically via~\cref{eq:fp}.  Blue solid: true; red dashed:
    recovered.  Shaded region quantifies pointwise discrepancy.
    \textit{Left (OU):} Gaussian $\mathcal{N}(0,0.245)$ reproduced with
    $\mathrm{TV}=0.0050$.
    \textit{Centre (double-well):} Bimodal density with peaks at
    $x\approx\pm1$ faithfully captured; $\mathrm{TV}=0.0092$.
    \textit{Right (multiplicative):} Unimodal heavy-tailed density
    reproduced with $\mathrm{TV}=0.0093$, demonstrating the effectiveness
    of the bias correction.  Discrepancy regions are visually negligible
    in all panels.}
  \label{fig:density}
\end{figure}

\subsection{Autocorrelation and Relaxation Timescales}

Figure~\ref{fig:autocorr} shows empirical autocorrelation functions from
long simulations of both true and identified systems, over a common
$1.2$-second lag window for all three panels.  The OU recovered
relaxation rate $\hat{\theta}=0.980$ (error~2.0\%) closely matches the
analytic reference $e^{-\tau}$ and the window is long enough relative to
the OU relaxation time ($1/\theta=1$) to show visible decay.  The
double-well window is short relative to the Kramers timescale
($\approx33$~s), so only the fast intra-well relaxation is visible: both
the true and recovered ACFs decay slightly, to $\approx0.87$--$0.89$ by
the end of the window, confirming that the identified generator captures
the local (intra-well) relaxation dynamics; probing inter-well mixing
would require a much longer lag window.  The multiplicative system
autocorrelation matches across all displayed lags, confirming that
$\hat{a}(x)=0.250+0.251x^2$ correctly encodes position-dependent mixing
rates.

\begin{figure}[tp]
  \centering
  \includegraphics[width=\textwidth]{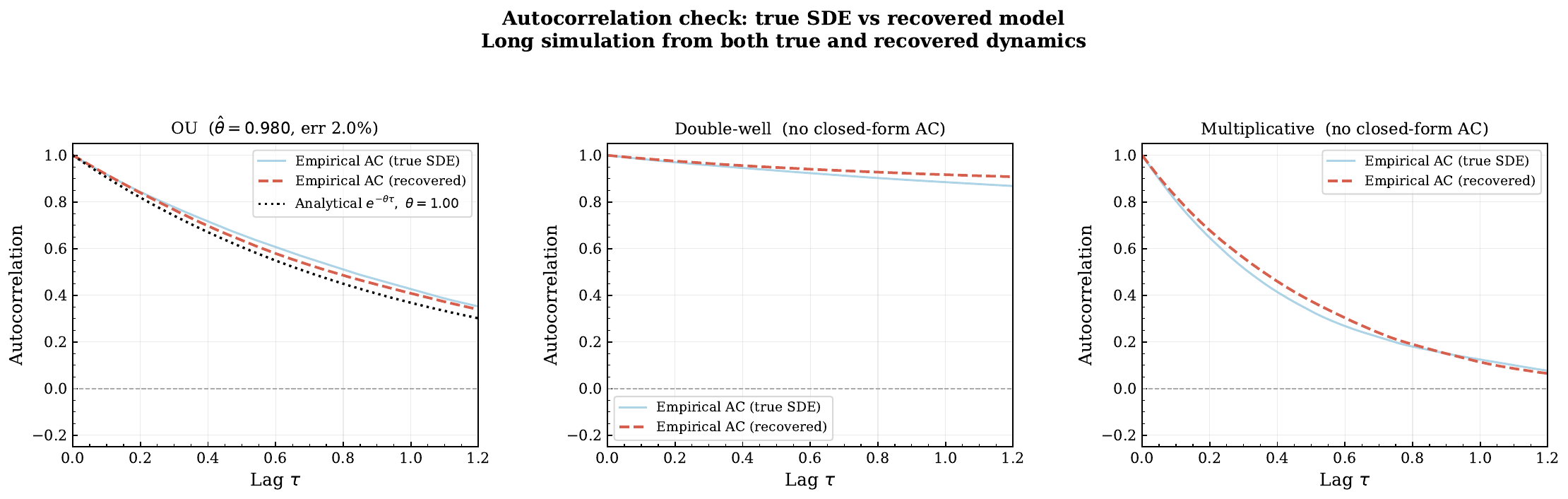}
  \caption{Autocorrelation: true SDE vs.\ recovered model, over a common
    $1.2$-second lag window ($\dt=0.002$, 600 lags) in all three panels.
    Blue solid: true SDE; red dashed: recovered model; dotted (left panel
    only): analytic $e^{-\theta\tau}$.
    \textit{Left (OU):} Recovered relaxation rate $\hat{\theta}=0.980$
    (spectral gap error 2.0\%); visible decay over the window since
    $1/\hat\theta\approx1$.
    \textit{Centre (double-well):} Window short relative to the Kramers
    timescale of $\approx33$~s, so only fast intra-well relaxation is
    visible; both curves decay modestly, to $\approx0.87$--$0.89$ by the
    end of the window, demonstrating faithful local relaxation dynamics.
    \textit{Right (multiplicative):} State-dependent diffusion correctly
    reproduces the position-dependent mixing rate over the same window.}
  \label{fig:autocorr}
\end{figure}

\subsection{Summary of Coefficient Recovery}

Table~\ref{tab:coefficients} provides a complete quantitative summary.
Every nonzero coefficient passes a 15\% relative error threshold, with the
largest error 3.9\% for the multiplicative drift.  All inactive
coefficients are set to exactly zero by the ElasticNetCV~+~STLSQ pipeline;
there are no false positives in any of the six sub-problems.

\begin{table}[tp]
\caption{Complete summary of recovered drift and diffusion coefficients.
All scientifically significant parameters pass a 15\% tolerance
(\checkmark).  Zero entries are set exactly to zero by
ElasticNetCV~$+$~STLSQ\@; no false positives appear.}
\label{tab:coefficients}
\begin{center}
\small
\begin{tabular}{llrrlrrl}
\toprule
System & Term
  & $\hat{c}_k$ & $c_k^{\mathrm{true}}$ & Drift err.
  & $\hat{d}_k$ & $d_k^{\mathrm{true}}$ & Diff.\ err. \\
\midrule
Ornstein    & $1$   & $0.000$  & $0.000$  & ---              & $0.490$ & $0.490$ & $0.0\%$\,\checkmark \\
Uhlenbeck   & $x$   & $-0.980$ & $-1.000$ & $2.0\%$\,\checkmark & $0.000$ & $0.000$ & --- \\
            & $x^2$ & $0.000$  & $0.000$  & ---              & $0.000$ & $0.000$ & --- \\
            & $x^3$ & $0.000$  & $0.000$  & ---              & $0.000$ & $0.000$ & --- \\
            & $x^4$ & $0.000$  & $0.000$  & ---              & $0.000$ & $0.000$ & --- \\
\midrule
Double      & $1$   & $0.000$  & $0.000$  & ---              & $0.250$ & $0.250$ & $0.1\%$\,\checkmark \\
Well        & $x$   & $+0.988$ & $+1.000$ & $1.2\%$\,\checkmark & $0.000$ & $0.000$ & --- \\
            & $x^2$ & $0.000$  & $0.000$  & ---              & $0.000$ & $0.000$ & --- \\
            & $x^3$ & $-0.976$ & $-1.000$ & $2.4\%$\,\checkmark & $0.000$ & $0.000$ & --- \\
            & $x^4$ & $0.000$  & $0.000$  & ---              & $0.000$ & $0.000$ & --- \\
\midrule
Multiplicative & $1$   & $0.000$  & $0.000$  & ---           & $0.250$ & $0.250$ & $0.0\%$\,\checkmark \\
               & $x$   & $-1.922$ & $-2.000$ & $3.9\%$\,\checkmark & $0.000$ & $0.000$ & --- \\
               & $x^2$ & $0.000$  & $0.000$  & ---           & $0.251$ & $0.250$ & $0.6\%$\,\checkmark \\
               & $x^3$ & $0.000$  & $0.000$  & ---           & $0.000$ & $0.000$ & --- \\
               & $x^4$ & $0.000$  & $0.000$  & ---           & $0.000$ & $0.000$ & --- \\
\bottomrule
\end{tabular}
\end{center}
\end{table}

\subsection{Empirical Convergence Rate}

Figure~\ref{fig:convergence} tests the $1/\!\sqrt{T}$ CLT rate of
Theorem~\ref{thm:clt} on all three benchmark systems: for OU and the
multiplicative system the theorem applies, while for the Euler-simulated double well
(outside the theorem) the observed slope is only consistent with that rate, using dense
(unregularised) OLS over the full five-term library
$\Theta(x)=[1,x,x^2,x^3,x^4]$ rather than the sparse, single-nonzero-term
estimator used elsewhere in the paper; $\|\hat c-c^*\|_2$ is therefore an
$\ell_2$ norm over all five coefficients, not the relative error on the
single active drift coefficient reported in
Table~\ref{tab:coefficients}, and the two are not directly comparable in
magnitude. Fixing $T=100$ and
varying the number of trajectories $R\in\{2,4,8,15,30,60,120\}$, the
$\ell_2$ coefficient error decreases with slope
$-1/2$ on a log--log scale up to $R\approx30$ ($T_{\mathrm{eff}}\approx3{,}000$)
for OU (fitted slope $-0.46$), after which the variance-dominated decrease gives way to a floor
of about $0.10$.  The double well follows a similar slope ($-0.49$) with no
clear floor by $R=120$ (error $0.047$); the multiplicative system has a steeper
initial slope ($-0.67$) but a much higher floor ($\approx0.25$).  We suspect
this reflects ill-conditioning of the higher monomials, since its stationary law is
concentrated near $|x|\lesssim0.5$, but we have not tested that, nor whether the
finite-$T$ bias below also contributes there.  For OU, as
\cref{rmk:discrete_vs_continuum} notes, this floor is not the
discretisation gap between $c^*(\dt)$ and the continuum coefficients
$c_0^*$, since the data here is itself Euler--Maruyama-simulated at the
same $\dt$ used in the fit.  We attribute it to a per-trajectory finite-$T$ bias
and tested this directly.  The cross term $\E[A_jZ_j]$ is nonzero for spatial
kernels as well (Remark~\ref{rmk:spatial_consistency_mechanism}), so the ratio
estimator is biased at finite $T$ within each trajectory, and pooling more
trajectories at fixed $T=100$ does not average that bias away.  In a
supplementary experiment on the OU process at fixed $T_{\mathrm{eff}}=12{,}000$ (stationary
initial conditions, same dense-OLS fit, 40 seeds per setting) we varied the
split between $R$ and $T$ and measured the \emph{signed} error of the mean fit,
which separates bias from variance.  The mean error of the $x^3$ coefficient
(true value $0$; the other coefficients are smaller) is
\begin{center}
\small
\begin{tabular}{@{}lcccccc@{}}
\toprule
$T$ (with $R=12000/T$) & 25 & 50 & 100 & 200 & 400 & 800 \\
\midrule
bias of $\hat c_{x^3}$ & $0.219$ & $0.167$ & $0.092$ & $0.058$ & $0.033$ & $0.012$ \\
standard error (40 seeds) & $0.004$ & $0.006$ & $0.006$ & $0.004$ & $0.004$ & $0.006$ \\
seed-to-seed s.d.\ of $\hat c$ (norm) & $0.068$ & $0.070$ & $0.076$ & $0.068$ & $0.063$ & $0.081$ \\
\bottomrule
\end{tabular}
\end{center}
At fixed $T_{\mathrm{eff}}$ the sampling variability is essentially constant
(about $0.07$), while the bias is many standard errors from zero and shrinks
by a factor of about $18$ as $T$ grows $32$-fold, roughly as $T^{-0.85}$ over
this range.  Together these account for the floor: at $R=120$, $T=100$ the
bias ($0.09$) and the standard deviation ($0.08$) combine to an rms error of
$0.12$.  The data are compatible with an asymptotic $1/T$ law but do not
establish the exact rate, and we have not tested other systems.

\begin{figure}[tp]
  \centering
  \includegraphics[width=\textwidth]{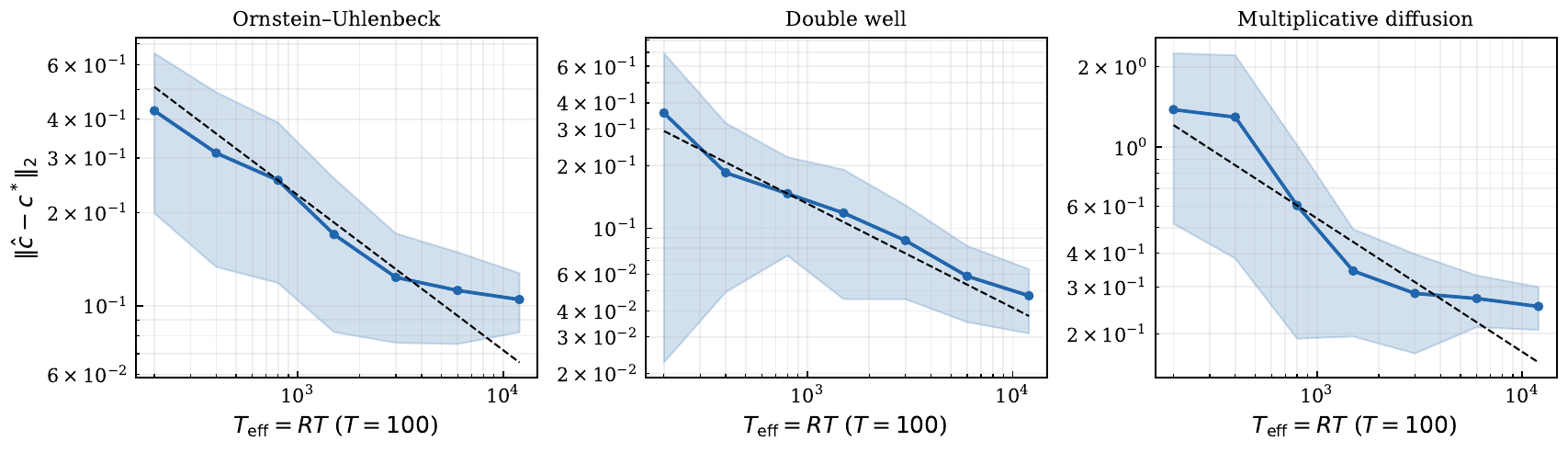}
  \caption{Empirical convergence of the drift estimator, fit by dense OLS over
    the full five-term library $[1,x,x^2,x^3,x^4]$ (no sparsification), on the
    three benchmark systems.  $\|\hat{c}-c^*\|_2$ vs.\ $T_{\mathrm{eff}}=R\cdot T$
    on a log--log scale ($T=100$, $\dt=0.002$, 20 independent seeds per $R$,
    $R\in\{2,\dots,120\}$); this $\ell_2$ norm over all five coefficients is not
    directly comparable to the single-coefficient relative errors of
    Table~\ref{tab:coefficients}.  Mean with a $\pm1$ standard-deviation band;
    dashed line: slope $-1/2$ reference.  Theorem~\ref{thm:clt} covers the OU and
    multiplicative panels (at fixed $R$, $T\to\infty$) but not $R\to\infty$ at
    fixed $T$, which is what is varied here; the double-well panel is outside the
    theorem, and its slope is only consistent with $-1/2$.  OU: the slope holds up to $T_{\mathrm{eff}}\approx3{,}000$ and then
    flattens to a floor of about $0.10$.  Double well: no clear floor within
    $R\le120$.  Multiplicative: the error is much larger and flattens near
    $0.25$ (we suspect ill-conditioning of the higher-order monomials, as the
    stationary law is concentrated near $|x|\lesssim0.5$; untested).  The OU floor is
    consistent with a per-trajectory finite-$T$ bias that pooling trajectories at fixed
    $T=100$ cannot remove (see text); the supplementary bias test was run for OU
    only, so no cause is established for the other two panels.}
  \label{fig:convergence}
\end{figure}

\subsection{Single-Trajectory Coefficient Estimation}
\label{sec:single_traj}

The method is formulated for identification from a \emph{single}
discretely sampled trajectory, yet
Figure~\ref{fig:convergence} and the main benchmark results of
\cref{sec:results} pool $R=120$ independent realisations. We close this
gap directly: Figure~\ref{fig:singletraj} repeats the OU coefficient-
recovery experiment with $R=1$ (a genuinely single trajectory,
re-simulated with an independent seed at each trajectory length $T$) and
compares it against $R=8$ pooled trajectories, both plotted against the
same effective sample size $T_{\mathrm{eff}}=RT$. The single-trajectory
curve falls roughly as $T_{\mathrm{eff}}^{-1/2}$ for all three systems (fitted
log--log slopes $-0.42$, $-0.46$ and $-0.55$ for OU, double well and
multiplicative diffusion).  The pooled $R=8$ curve lies \emph{above} it at small
$T_{\mathrm{eff}}$ (e.g.\ median relative error $0.36$ against $0.12$ for OU at
$T_{\mathrm{eff}}=50$, i.e.\ $T=6$ per trajectory), the finite-$T$ bias of
Remark~\ref{rmk:spatial_consistency_mechanism} at work, and the two come within
a factor of about $1.6$ of each other at $T_{\mathrm{eff}}\ge800$ (in either direction).
Splitting a fixed budget of data into several short trajectories is
therefore not free.  This supports the single-trajectory
\emph{coefficient-estimation} claim with known support on all three systems; it
does not by itself establish finite-sample sparse support recovery for the full
library.

\begin{figure}[tp]
  \centering
  \includegraphics[width=\textwidth]{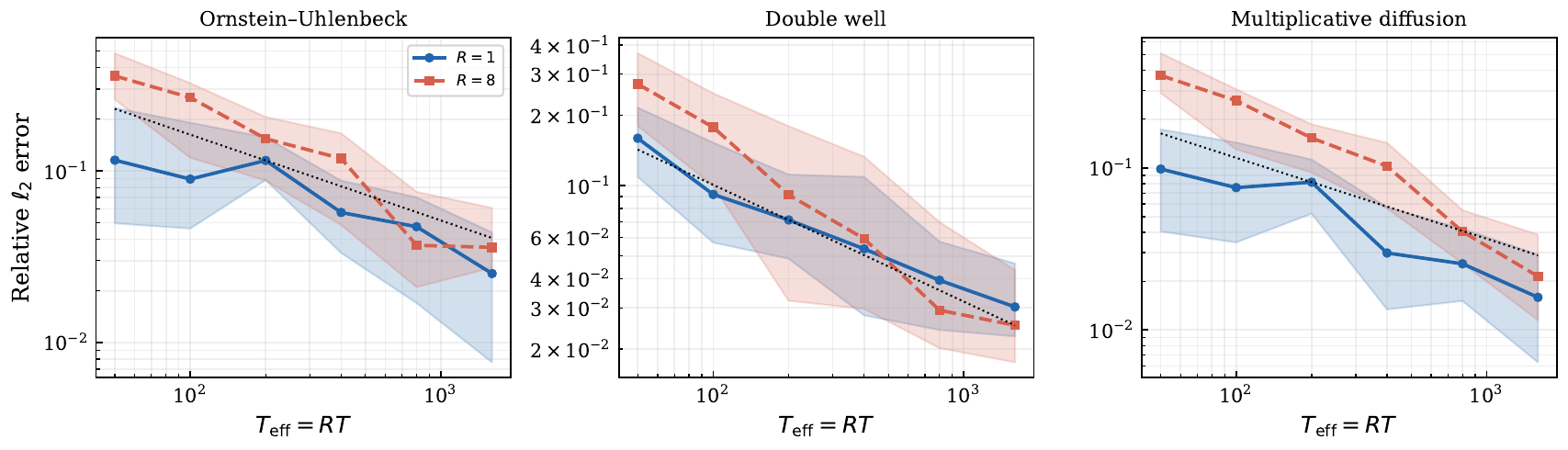}
  \caption{Single-trajectory ($R=1$) versus pooled ($R=8$) coefficient
    recovery with known support on the three benchmark systems (active terms
    $[x]$ for OU and multiplicative, $[x,x^3]$ for the double well), spatial-kernel
    estimator, stationary initial conditions, plotted against effective sample size
    $T_{\mathrm{eff}}=R\cdot T$ ($\dt=0.002$, $M=40$ centres, 30 seeds per point,
    median with IQR band; error is the relative $\ell_2$ error of the active
    coefficients).  The $R=1$ curve falls roughly as $T_{\mathrm{eff}}^{-1/2}$
    (dotted reference); the $R=8$ curve lies above it at small
    $T_{\mathrm{eff}}$ and comes within a factor of about $1.6$ of it (in
    either direction) at $T_{\mathrm{eff}}\ge800$.}
  \label{fig:singletraj}
\end{figure}

\paragraph{Model selection with $R=1$.}
This experiment fits only the known active library terms (one for OU and the
multiplicative system, two for the double well), so it deliberately isolates the \emph{estimation} question ---
does the coefficient converge at the right rate from one trajectory? ---
without engaging sparse support recovery.  Accordingly, Figure~\ref{fig:singletraj}
is evidence for single-trajectory estimation, not yet for finite-sample support
selection over the full sparse library. For the full multi-term libraries used
elsewhere in the paper, ElasticNetCV's regularisation path is selected by
trajectory-grouped cross-validation (\cref{sec:algorithm}), which is
unavailable by construction when $R=1$: there is only one trajectory to
group. The practical substitute, which we recommend for single-trajectory
use of the full pipeline, is to replace trajectory-grouped folds with
\emph{contiguous time-block} folds: partition the single trajectory into
contiguous windows of length $\gg$ the system's mixing time and treat
each block as a fold, exactly as block cross-validation is used elsewhere
for autocorrelated time series~\cite{mangan2017}. This preserves the
within-fold autocorrelation structure that trajectory-grouping was
designed to protect, at the cost of needing a trajectory somewhat longer
than the mixing time to form enough blocks for a meaningful CV split.  A full
multi-term $R=1$ support-recovery benchmark is therefore a separate numerical
validation of the sparse-selection component rather than a consequence of
Figure~\ref{fig:singletraj}.  We perform that validation below using the
double-well system with five contiguous time blocks for model selection and
12 independent single-trajectory runs at each trajectory length.
\begin{table}[tp]
  \centering
  \caption{Full sparse-pipeline single-trajectory recovery on the double-well
    system ($\dt=0.002$).  Each row summarizes 12 independent $R=1$ runs
    with contiguous block cross-validation.  Exact drift support counts runs
    recovering the true support; diffusion support was recovered exactly in all
    48 runs.  Drift F1 and the two $L^2$ errors are \emph{medians} over the 12
    runs (the corresponding mean drift F1 is $0.63,\,0.83,\,1.00,\,1.00$).
    The table is a finite-sample validation of sparse model selection,
    not a non-asymptotic support-recovery theorem.}
  \label{tab:singletraj_sparse}
  {\small
  \begin{tabular}{@{}ccccc@{}}
    \toprule
    $T$ & Drift support & Drift F1 & Drift $L^2$ & Diffusion $L^2$ \\
    \midrule
    50  & 3/12  & 0.571 & 1.154 & $1.89\times10^{-3}$ \\
    100 & 8/12  & 1.000 & 0.218 & $1.67\times10^{-3}$ \\
    200 & 12/12 & 1.000 & 0.092 & $7.25\times10^{-4}$ \\
    400 & 12/12 & 1.000 & 0.066 & $3.95\times10^{-4}$ \\
    \bottomrule
  \end{tabular}}
\end{table}

The table shows where the single-trajectory pipeline works and where it does
not.  At $T=50$ only 3 of 12 runs recover the exact drift support (the median
F1 of $0.57$ is consistent with spurious extra terms, and the median drift
$L^2$ error is $1.15$); at $T=100$ the figure is 8 of 12, with a
median F1 of $1.000$ but a mean of $0.83$; reliable support recovery
(12 of 12) requires $T\gtrsim200$.  For comparison, the inter-well
(Kramers) time of this system is $\approx33$, so $T=200$ is about six
transitions' worth of data; we have not tested whether that is the controlling
scale.

\subsection{Additional Single-Trajectory Benchmarks}
\label{sec:benchmark_table}

To test the method beyond the three systems above, we run it on eight further
scalar SDEs of deliberately simple form (Table~\ref{tab:sde_benchmark}): a
shifted OU process, single-well and tilted double-well cubic drifts, a pure
cubic drift, an asymmetric cubic drift, and three systems with
state-dependent (multiplicative) diffusion.  Each fit uses \emph{one}
Euler--Maruyama trajectory ($R=1$; $\dt=0.005$, $T=1000$, i.e.\
$2\times10^{5}$ observations) with no pooling across trajectories.  The
library is kept small, $\Theta(x)=[1,x,x^2,x^3]$, for both drift and diffusion
(the diffusion is \emph{not} restricted to a constant), and the same pipeline
as in \cref{sec:algorithm} is used: $M=40$ Gaussian kernels equally spaced over
the 0.5\%--99.5\% quantile range of the observed path with $h=2.15\times$ the
centre spacing, finite-step drift-squared correction of the diffusion
target, and five-block contiguous time-block cross-validation in place of
trajectory-grouped folds (\cref{sec:single_traj}).  Half of these eight
systems (the single-well, tilted double-well, asymmetric-cubic, and both
multiplicative-diffusion Duffing/double-well drifts) deliberately activate
$x$ and $x^3$ simultaneously, the most collinear monomial pair under a
roughly Gaussian stationary density; model selection on this table therefore
uses the elastic-net generalisation of \cref{eq:lasso}, not plain LASSO, for
the reason given in \cref{sec:algorithm}.  To report statistics
without pooling, every system is fitted on 10 independent single
trajectories.  Errors are the occupancy-weighted relative $L^2$ error of the
recovered function,
$\bigl[\tfrac1N\sum_n(\hat f(X_n)-f(X_n))^2\big/\tfrac1N\sum_nf(X_n)^2\bigr]^{1/2}$,
reported as the median over the 10 fits; ``supp.'' counts fits whose
identified support equals the true support exactly.  Exact drift support is
recovered in all ten fits for only two of the eight systems; the lowest rates
are 5/10 (single-well Duffing), 6/10 (tilted double well) and 7/10
(multiplicative-noise OU), even though the median relative $L^2$ errors are
all at most $0.11$.  Diffusion support is exact in all 80 fits.  We make no
claim that support recovery is reliable at $T=1000$ for these collinear or
weakly identified libraries.

\begin{table}[tp]
  \centering
  \caption{Single-trajectory ($R=1$) identification of eight additional scalar
    SDEs, $\dt=0.005$, $T=1000$, library $[1,x,x^2,x^3]$.  Median relative
    $L^2$ error over 10 independent single-trajectory fits, and the number of
    fits (out of 10) that recover the exact sparse support.}
  \label{tab:sde_benchmark}
  \resizebox{\textwidth}{!}{\input{table_sde_benchmark.tex}}
\end{table}

\subsection{Hyperparameter Robustness}

Figure~\ref{fig:hyperparams} examines the sensitivity of the
function-space reconstruction error to the kernel bandwidth $h$ and
number of centres $M$ on the double-well system.  Across a $7\times6$ grid spanning
$h\in[0.08,0.43]$ and $M\in[10,100]$, all mean absolute relative errors
remain within 2--11\%, so moderate misspecification of either
hyperparameter does limited damage.  The paper default
$(h=0.22,\,M=50)$ lies in the blue band, within a factor of two of the grid
minimum (4.8\% against 2.4\%).  This default value is not comparable to the
2.7\% double-well drift error of Figure~\ref{fig:recovery}: the heatmap uses
dense, unregularised OLS and a function-space mean absolute relative error over
$[-2.5,2.5]$, whereas Figure~\ref{fig:recovery} reports relative coefficient
errors of the sparse (elastic-net$+$STLSQ) pipeline.  Errors increase primarily with bandwidth: overly large $h$
smooths out spatial variation in the drift, while varying $M$ at fixed
$h$ has comparatively little effect once $M\geq20$.

\begin{figure}[tp]
  \centering
  \includegraphics[width=0.62\textwidth]{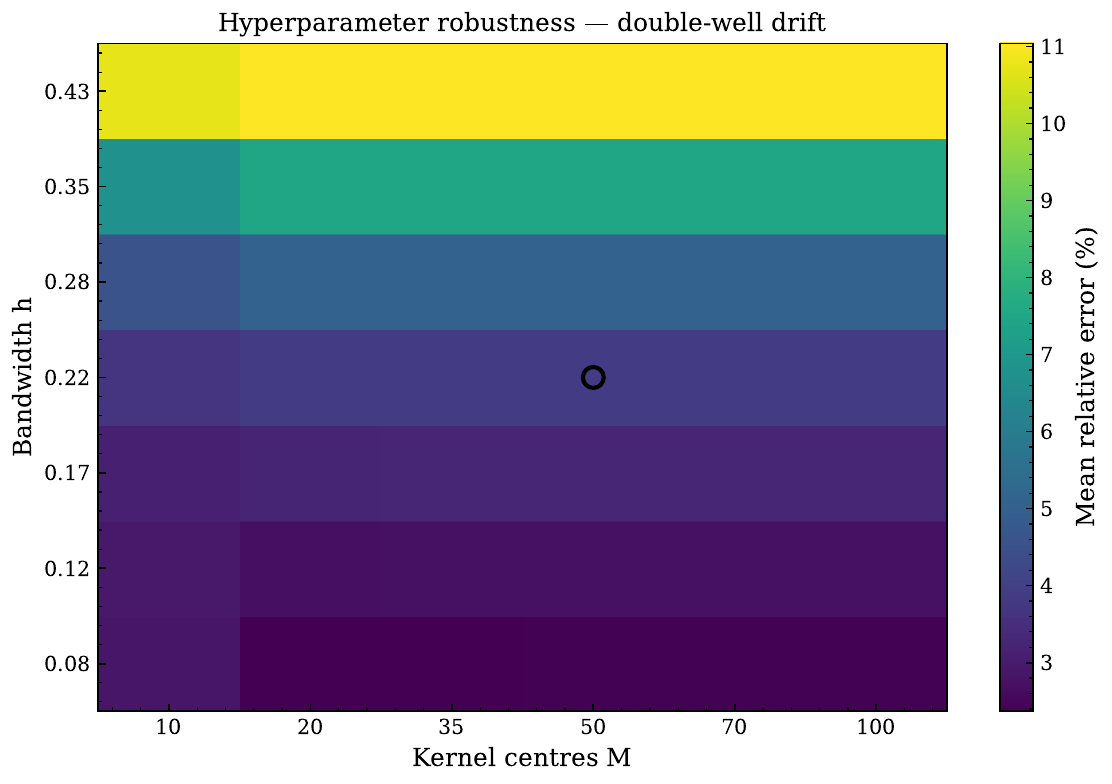}
  \caption{Hyperparameter robustness on the double-well drift
    ($R=120$ trajectories, OLS, function-space MAE over $[-2.5,2.5]$).
    Heatmap of mean relative error (\%) across a
    $7\times6$ grid of $(h,M)$ values; open circle marks the paper default
    $(h,M)=(0.22,50)$.  Errors range from 2.4\% (at $h=0.08$, $M=20$) to
    11.0\% (at $h=0.43$, $M=100$).  Error increases with bandwidth and is comparatively
    flat across $M$ once $M\geq20$.}
  \label{fig:hyperparams}
\end{figure}

\subsection{Finite-Time-Step Bias Correction}

Table~\ref{tab:bias_correction} provides a direct comparison of the
diffusion coefficient estimates before and after the drift-squared bias
correction~\cref{eq:bias_corr} for the multiplicative system, from a
dedicated run isolating the effect of the correction. (It uses the same
settings as the main pipeline ($R=120$ trajectories, $\dt=0.002$, $M=50$
centres, $h=0.27$) but an independent random realisation, so it gives a
slightly different point estimate for the same coefficient, 0.249 here
against 0.251 in Table~\ref{tab:coefficients}, both well within sampling
variability of the true value 0.250.) For the
constant term $d_1$, the OLS uncorrected estimate is already accurate
(0.0\% error) because the drift-squared bias enters predominantly through
the spatially varying term.  For the quadratic term $d_{x^2}$, the OLS
uncorrected estimate of 0.259 (3.5\% error) is reduced to 0.251 (0.5\%
error) by the correction with plain OLS, and to 0.249 (0.5\%) with the
sparse regressor: a sevenfold improvement that is attributable to the
correction, not to the change of regressor.  Both
corrected estimates pass the 15\% tolerance threshold with large margin.

\begin{table}[tp]
  \centering
  \caption{Finite-time-step bias correction for the multiplicative
    diffusion system ($\dt=0.002$, $R=120$; an independent realisation of the
    main-pipeline run).  Recovered coefficient values and relative errors:
    plain OLS without the correction, plain OLS \emph{with} the correction,
    and the sparse ElasticNetCV$+$STLSQ regression with the correction (the
    main pipeline), alongside the true values.  The middle column isolates the
    correction from the regressor: the correction alone reduces the
    $d_{x^2}$ error from 3.5\% to 0.5\%, and the sparse regressor adds no
    material further change.}
  \label{tab:bias_correction}
  \resizebox{\textwidth}{!}{%
  \begin{tabular}{lcccccccc}
    \toprule
    & \multicolumn{4}{c}{Coefficient value} & \multicolumn{3}{c}{Relative error (\%)} \\
    \cmidrule(lr){2-5}\cmidrule(lr){6-8}
    Coefficient & OLS & OLS + corr. & ElasticNet + corr. & True & OLS & OLS + corr. & ElasticNet + corr. \\
    \midrule
    Constant term $d_1$       & 0.250 & 0.250 & 0.250 & 0.250 & 0.0 & 0.0 & 0.1 \\
    Quadratic term $d_{x^2}$  & 0.259 & 0.251 & 0.249 & 0.250 & 3.5 & 0.5 & 0.5 \\
    \bottomrule
  \end{tabular}%
  }
\end{table}

\subsection{Observation Noise: What the Weak Form Does and Does Not Buy}
\label{sec:noise_scaling}

Two statements about observation noise are easily conflated, so we separate
them.  Write $Y_n=X_{t_n}+\eta_n$ with $\eta_n\sim\mathcal N(0,\nu)$ i.i.d.

\emph{Per point.}  The Kramers--Moyal target $(Y_{n+1}-Y_n)/\dt$ has
observation-noise variance $2\nu/\dt^2$ (plus process noise $\sigma^2/\dt$),
which diverges as $\dt\to0$.

\emph{Averaged.}  Any estimator that averages over the visits to a region of
state space divides this by the number of visits $m\propto N=T/\dt$.  The
weak-form statistics $\widetilde B_j=\sum_nK_j(Y_n)(Y_{n+1}-Y_n)$ and
$\widetilde A_{jk}=\dt\sum_nK_j(Y_n)f_k(Y_n)$ are, up to normalisation, exactly a
kernel-weighted average of the increment ratios $(Y_{n+1}-Y_n)/\dt$ regressed on
the library.  The random part of $\widetilde B_j$ has variance of order
$2\nu\sum_nK_j(Y_n)^2\propto\nu N$ while $\widetilde A_j\propto T$, so the
zero-mean noise fluctuation of $\hat c$ (the fluctuation that would remain if the
$O(\nu/\dt)$ bias below were known and removed) is of order
$\sqrt{2\nu}/(\dt\sqrt{m})\propto\sqrt{\nu/(T\dt)}$, i.e.\ $\dt^{-1/2}$ at fixed $T$.  A binned Kramers--Moyal
average over the same visits has the same order.  We therefore make
\emph{no} claim that the weak form has a smaller averaged variance than an
equally averaged Kramers--Moyal estimator; what it avoids is the binning (and
the pointwise differencing) itself, and it comes with the exact
Gaussian-moment identities of Theorem~\ref{thm:noise} that allow the bias to
be corrected.

Under observation noise the bias, not the variance, is what dominates.  The
uncorrected statistic carries an $O(\nu/\dt)$ bias
(Theorem~\ref{thm:noise}(i), confirmed in Figure~\ref{fig:noise_bias}), the
same $1/\dt$ divergence we criticised for Kramers--Moyal in
\cref{sec:noiseproblem}, and the sampling variability of that bias grows with it: the bias term
$-\sum_nK_j(Y_n)\eta_n$ is multiplied by the visit count of the kernel region,
which fluctuates from path to path through the (noisy) argument of $K_j$, so its
fluctuation scales with the bias itself rather than with the $\dt^{-1/2}$ law
above.  In a Monte Carlo check (OU,
$\sigma_\eta=0.2$, $T=200$, a single kernel at $x=1$, 300 repetitions) the
standard deviation of the \emph{uncorrected} estimate rose from $0.39$ to
$1.36$ to $4.8$ as $\dt$ fell from $0.064$ to $0.016$ to $0.004$, roughly as
$\dt^{-0.9}$, steeper than the $\dt^{-1/2}$ of the zero-mean part and tracking the
bias, not a favourable $\sqrt{\dt}$ law.  The weak form's
robustness to observation noise thus rests on the Gaussian moment correction of
\cref{sec:noiserobust}, validated next, and not on any variance scaling.

\paragraph{Regression-level validation of the measurement-noise correction.}
To test the bias removal directly, we apply the complete noisy-data
pipeline to OU trajectories over
$\dt\in\{0.01,0.002\}$,
$\sigma_\eta\in\{0,0.01,0.03,0.05,0.08,0.12\}$,
and $T\in\{25,50,100,200\}$, with 20 independent repetitions per point.
At $T=200$ and $\sigma_\eta=0.05$, the mean signed drift bias at
$\dt=0.002$ is $-5.11$ for the naive estimator but only $+0.006$ after
Gaussian correction; at $\dt=0.01$ the corresponding values are $-1.06$
and $-0.032$.  The corrected signed bias remains close to zero over the
moderate-noise range $\sigma_\eta\le0.08$ at both step sizes, while the
naive bias grows rapidly with $\sigma_\eta^2/\dt$. The reduction is
interpreted as removal of the dominant singular bias term; residual
errors-in-variables bias and finite-sample variability are not claimed to
vanish identically at fixed nonzero noise.
At the largest tested noise level $\sigma_\eta=0.12$, finite-sample error
becomes visibly larger, especially for $\dt=0.002$; we therefore interpret
this experiment as a regression-level validation of the population correction,
not as a finite-sample uniform-in-noise guarantee.

\begin{figure}[tp]
  \centering
  \includegraphics[width=0.78\textwidth]{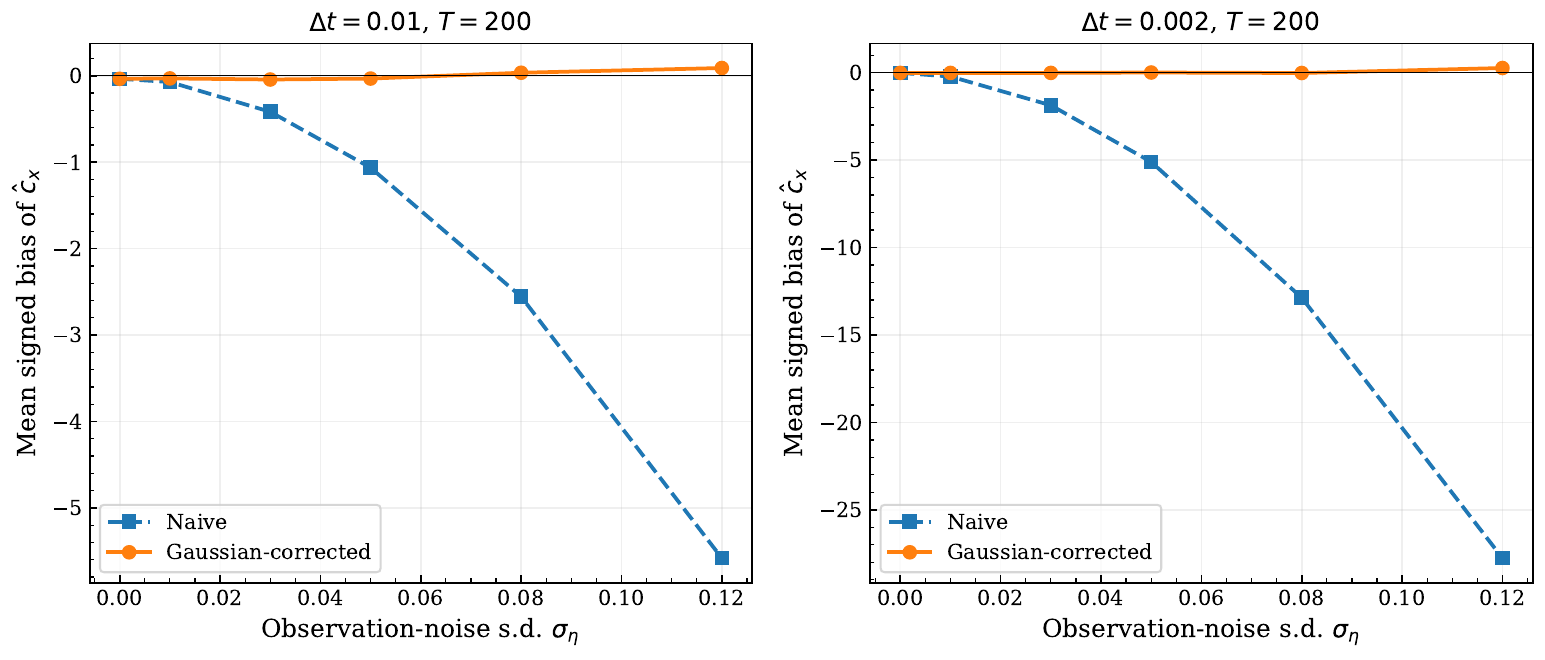}
  \caption{Regression-level validation of the Gaussian observation-noise
    moment correction on the OU process at $T=200$ (20 independent repetitions per
    point).  The signed drift bias of the naive estimator grows rapidly in
    magnitude with observation-noise standard deviation, with stronger
    degradation at the finer sampling interval.  The Gaussian moment correction substantially suppresses the leading
    bias throughout the moderate-noise regime; degradation at
    $\sigma_\eta=0.12$ reflects finite-sample error and the residual
    errors-in-variables term rather than the re-emergence of the leading
    $\sigma_\eta^2/\dt$ bias.}
  \label{fig:noise_bias}
\end{figure}

The fixed-noise trajectory-length audit at
$(\dt,\sigma_\eta)=(0.002,0.05)$ provides an additional check of the
asymptotic direction.  In the 20-repetition audit shown in
Figure~\ref{fig:noise_trajectory}, the mean signed bias of the corrected
drift coefficient changes from approximately $-0.254$ at $T=25$ to
$+0.006$ at $T=200$, while the naive bias remains strongly displaced, from
about $-5.91$ to $-5.11$.  The corresponding mean absolute corrected error
falls from approximately $0.461$ to $0.206$.  These values are consistent
with substantial reduction of sampling error after the singular
observation-noise term has been removed; they do not establish exact
fixed-noise coefficient consistency.

\begin{figure}[tp]
  \centering
  \includegraphics[width=0.78\textwidth]{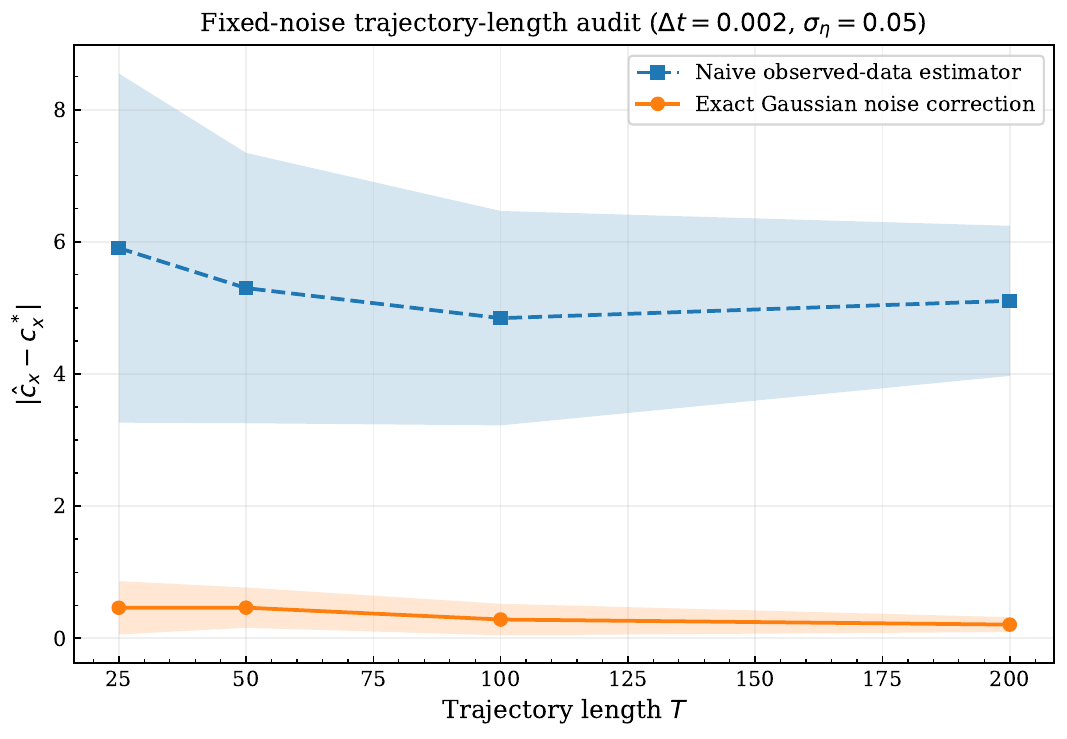}
  \caption{Fixed-noise trajectory-length audit for the OU drift at
    $\dt=0.002$ and $\sigma_\eta=0.05$, with 20 independent repetitions per
    trajectory length.  Curves show mean absolute coefficient error with one
    standard-deviation bands.  The Gaussian moment correction substantially
    suppresses the noise-induced displacement of the naive estimator, while
    the remaining corrected error decreases with $T$.}
  \label{fig:noise_trajectory}
\end{figure}

These direct regression experiments test the observation-noise correction
inside the identification pipeline.  The theorem establishes exact Gaussian moment identities and a leading-order
coefficient-bias characterization; these finite-sample experiments validate
that leading-order correction rather than asserting exact fixed-noise
coefficient recovery.

\section{Relation to Existing Methods}
\label{sec:related}

Table~\ref{tab:comparison} compares the proposed method with those in
\cref{sec:related_intro}. Several of them already share individual properties
with it: gEDMD and the stochastic occupation kernel method are derivative-free,
unbiased, state-weighted SDE estimators, and Langevin regression already
applies a finite-step bias correction. The present method differs in the
combination: joint drift and diffusion recovery from one shared design matrix,
an explicit sparse polynomial-library output, and a quantified diagnosis of why
temporal weighting fails. EDMD~\cite{williams2015} and
gEDMD~\cite{klus2020} are related through the generator; EDMD alone does
not target the diffusion term, while gEDMD does but typically projects
onto a general dictionary of observables rather than solving for a sparse
polynomial generator. The occupation kernel
method~\cite{wells2024} is nonparametric (RKHS-valued) rather than
sparse-symbolic. Neural SDE methods~\cite{kidger2021} trade
interpretability for flexibility and produce no generator amenable to
spectral or Kramers-rate analysis. Martingale estimating
functions~\cite{bibby1995,kessler1999} establish the underlying
unbiasedness principle but are formulated for individual scalar
diffusion-coefficient estimation rather than joint sparse regression over
a candidate library.

\begin{table}[tp]
\caption{Comparison of SDE identification methods. ``Joint'' indicates
  that drift and diffusion are identified from a single shared design
  matrix/regression pass rather than separate procedures.}
\label{tab:comparison}
\begin{center}
\resizebox{\textwidth}{!}{%
\small
\begin{tabular}{lccccc}
\toprule
Method & Handles SDEs & Identifies $a(x)$ & Symbolic \& sparse
       & Derivative-free & Joint \\
\midrule
SINDy~\cite{brunton2016}              & $\times$   & $\times$   & \checkmark & $\times$     & --- \\
Stoch.\ SINDy~\cite{boninsegna2018,wanner2024} & \checkmark & \checkmark & \checkmark & $\times$     & $\times$ \\
Weak SINDy~\cite{messenger2021}       & $\times$   & $\times$   & \checkmark & \checkmark   & --- \\
EDMD~\cite{williams2015}              & \checkmark & partial    & $\times$   & \checkmark   & $\times$ \\
gEDMD~\cite{klus2020}                 & \checkmark & \checkmark & partial    & \checkmark   & \checkmark \\
Occupation kernel~\cite{wells2024}    & \checkmark & \checkmark & $\times$   & \checkmark   & $\times$ \\
Langevin regression~\cite{callaham2021} & \checkmark & \checkmark & \checkmark & $\times$   & $\times$ \\
Mart.\ est.\ functions~\cite{bibby1995,kessler1999} & \checkmark & \checkmark & $\times$ & \checkmark & --- \\
Neural SDE~\cite{kidger2021}          & \checkmark & \checkmark & $\times$   & \checkmark   & \checkmark \\
\textbf{Proposed}                     & \checkmark & \checkmark & \checkmark & \checkmark   & \checkmark \\
\bottomrule
\end{tabular}}
\end{center}
\end{table}

\section{Discussion and Future Directions}
\label{sec:discussion}

\subsection{What the Symbolic Model Enables}

An explicit symbolic generator $\hat{\Lop}$ supports analyses that a
black-box surrogate does not. Applying the algorithm in a rolling window gives a
time series of generators, and the estimated spectral gap should close ahead of
a bifurcation such as the double-well to single-well transition, which would
serve as an early-warning indicator (we have not tested this). The leading
eigenfunctions of $\hat{\Lop}$ give a slow manifold without further simulation,
for the double well the signed well occupancy, which can be used for model
reduction. The regression systems are unchanged for non-reversible dynamics,
where the spectrum of the non-symmetric generator describes state-space
circulation.

\subsection{Limitations}

As with any library-based method, the library must span the true drift and
diffusion. Terms missing from it yield the best kernel-weighted approximation
(Corollary~\ref{cor:l2proj}), not the exact generator. The bandwidth $h$ and
number of centres $M$ need tuning, although Figure~\ref{fig:hyperparams} shows
errors of 2--11\% across a $7\times6$ grid spanning a factor of five in $h$ and
ten in $M$. A selection rule based on leave-one-out cross-validation over
$(h,M)$ would reduce this dependence. The number of diffusion parameters grows
as $d(d+1)/2$ with the state dimension $d$, which becomes expensive in high
dimension unless the diffusion is structured (diagonal or low-rank, say).
The measurement-noise theorem assumes i.i.d. Gaussian observation noise and gives exact moment identities together with a leading-order coefficient-bias characterization.  The remaining coefficient bias is an errors-in-variables term of order $O(\sigma_\eta^2)$ under the stated smoothness assumptions.  The regression-level validation is strongest in the moderate-noise regime; finite-sample error grows at the largest tested noise level, especially for the finer sampling interval.  Exact fixed-noise coefficient recovery would require an additional correction of the noisy design, such as deconvolution.  For non-Gaussian observation noise, the present paper does not claim even the Gaussian exact moment correction; the corresponding extension would require a different moment or deconvolution argument.

\subsection{Future Work}

Six problems remain open. (i)~A central-limit theorem for the diffusion
coefficients, built on the quadratic-variation fluctuation of
\cref{eq:qv_split,eq:bias_expand}, would complete the distribution theory
beyond the drift result of Theorem~\ref{thm:clt}. (ii)~A finite-sample
support-recovery analysis, including non-asymptotic guarantees for the
single-trajectory block-CV setting, would complement the identifiability
conditions. (iii)~The method needs extending to several dimensions with coupled,
non-diagonal diffusion tensors, as molecular and financial applications require.
(iv)~Library selection from overcomplete dictionaries, guided by physical
constraints or symmetry, would reduce the dependence on a user-chosen polynomial
library. (v)~Measurement-noise moment or deconvolution corrections for
non-Gaussian observations need a different argument from the Gaussian one used
here. (vi)~We have not compared directly against gEDMD~\cite{klus2020}, the
stochastic occupation kernel method~\cite{wells2024} or Langevin
regression~\cite{callaham2021} on shared benchmarks, in particular in the
higher-dimensional regime where binning-based stochastic SINDy degrades.

\section{Conclusion}
\label{sec:conclusion}

We presented a weak-form, state-weighted method for identifying sparse
stochastic generators from discretely sampled trajectories. It recovers
symbolic drift and diffusion functions through a shared kernel-library design
matrix, and the resulting generator gives relaxation timescales, stationary
densities and metastable escape rates.

The central point is the difference between temporal windowing and
state-dependent weighting. Fixed-width temporal windows keep a non-vanishing
within-window correlation between the regression design and the stochastic
residual, which leaves a persistent coefficient-level bias. State-dependent
kernels reuse each regression equation over an ever-growing number of visits to
the same region of state space, so the signal grows linearly, the noise only like
the square root, and the martingale strong law applies to the stochastic term.
For the OU process, Proposition~\ref{prop:ou_temporal} shows that the windowed
temporal estimator converges to $0$ for every window shape and width, while the
spatial estimator converges to the target.  Proposition~\ref{prop:ou_temporal}
is an OU statement; for other systems only the nonzero limit of
Theorem~\ref{thm:endogeneity} is claimed.

At fixed sampling interval, the target is the coefficient vector of the
Euler--Maruyama discrete-time model. The drift and diffusion estimators are
strongly consistent for it under the stated ergodicity, moment and rank
assumptions.  The drift estimating equation has
the $T^{-1/2}$ central-limit rate, and smooth drift-based functionals such as
the OU spectral gap inherit that rate.  A separate diffusion-coefficient CLT
is not claimed; the quadratic-variation fluctuation that controls its
finite-sample variance is identified explicitly and is left for future work.
Recovery of the coefficients of the underlying continuous-time generator
requires the additional $\dt\to0$ limit.

For additive observation noise, exact Gaussian moment identities are established
for i.i.d. Gaussian errors.  These identities remove the persistent
$\sigma_\eta^2/\dt$ contribution generated by the observation-noise
increment, leaving an $O(\sigma_\eta^2)$ errors-in-variables remainder in the
coefficient equation under the stated smoothness assumptions.  We therefore
do not claim exact fixed-noise coefficient consistency; an additional
deconvolution-type correction of the noisy design would be required for that
stronger statement.  The result is deliberately not extended to non-Gaussian
observation noise.

The experiments recover the coefficients accurately on linear, bistable and
state-dependent-diffusion benchmarks, and the dynamical quantities computed from
the recovered generators match the true ones.  The single-trajectory
double-well experiment of \cref{tab:singletraj_sparse} applies the full
sparse pipeline with contiguous block-CV: exact drift support is recovered
in every run at $T=200$ and $400$ but only in 8 of 12 runs at $T=100$ and
3 of 12 at $T=50$; on the eight further systems at $T=1000$ it is recovered in
5 to 10 of 10 fits depending on the system.  Diffusion support is recovered in
every tested run.  The single-trajectory claim is therefore supported for
coefficient estimation (the OU experiment), and for sparse support recovery
only when the trajectory is long enough: $T\gtrsim200$ on the double well, with
no guarantee for shorter records or for the weakly identified systems above.
A non-asymptotic support-recovery theory remains open.

\section*{Acknowledgments}

The authors thank Dr.~Rahul~S and Anil~N.~Bhat for helpful discussions. We would further wish to thank the Quantum and Nano Devices Laboratory, PES University, for providing us the required computation resources.

\section*{Statements and Declarations}

\textbf{Funding.}
The authors declare that no funds, grants, or other support were received
during the preparation of this manuscript.

\textbf{Competing Interests.}
The authors have no relevant financial or non-financial interests to disclose.

\textbf{Data Availability.}
The complete implementation, including simulation environments, the
spatial Gaussian kernel projection pipeline, and figure-generation scripts,
is available at \url{https://github.com/eshwarRA/Weak-Stochastic-SINDy/}.
No additional datasets were generated or analysed beyond those produced
by the numerical experiments described in the paper.

\bibliographystyle{cas-model2-names}
\bibliography{references}

\clearpage

\appendix
\section{Theoretical Properties}
\label{sec:theory}

\subsection{Standing Assumptions}
\label{sec:assumptions_app}

\paragraph{Provenance of the arguments.}
None of the proofs below is copied from another paper.  The arguments for
Theorems~\ref{thm:consistency} and~\ref{thm:clt} are standard applications of
the ergodic theorem, the martingale strong law and central limit
theorem~\cite{hall1980,meyn2009} and of Slutsky's lemma and the delta
method~\cite{vandervaart1998}, in the spirit of the martingale
estimating-function theory of~\cite{bibby1995,kessler1999}; we do not claim
novelty for the consistency statement itself.  The Gaussian moment identities
of \cref{sec:noiserobust} rest on Stein's identity~\cite{stein1981}.  The results we
regard as specific to this paper are the endogeneity analysis of
temporal windows (Theorem~\ref{thm:endogeneity} and
Proposition~\ref{prop:ou_temporal}), the finite-step correction and the
noise-corrected statistics.

The standing assumptions used throughout this appendix are
Assumptions~\ref{ass:ergodic}--\ref{ass:em_ergodic} stated in
\cref{sec:assumptions} of the main text.  Assumption~\ref{ass:em_ergodic}
is verified for the OU and multiplicative benchmarks but fails for the
Euler-simulated cubic double well (Remark~\ref{rmk:em_ergodicity}), so the
theorems below do not cover that benchmark.

\subsection{Consistency of the Weak Estimator}
\label{sec:consistency}

\begin{theorem}[Strong consistency]
\label{thm:consistency}
Under Assumptions~\ref{ass:ergodic}--\ref{ass:em_ergodic}, as
$T\to\infty$ with $\dt=T/N$ fixed,
\begin{equation}
  \frac{1}{T}A_{jk} \as \bar{A}_{jk}
  := \int K_j(x)\,f_k(x)\,\mu_{\dt}(dx), \qquad
  \frac{1}{T}B_j \as \bar{B}_j
  := \int K_j(x)\,b(x)\,\mu_{\dt}(dx).
  \label{eq:ergodic_conv}
\end{equation}
If $\bar{A}$ has full column rank, then $\hat{c}\as c^*$ as $T\to\infty$.
\end{theorem}

\begin{proof}
This is the standard argument for ergodic estimating equations (compare the
consistency results for martingale estimating functions in
\cite{bibby1995,kessler1999}); we write it out only to fix notation.  It uses
Birkhoff's ergodic theorem for geometrically ergodic chains~\cite{meyn2009} and
the $L^2$ martingale strong law~\cite{hall1980}.

\textit{Step~1: Ergodic convergence of $A_{jk}$.}
$g_{jk}(x):=K_j(x)f_k(x)$ is bounded and Lipschitz by
Assumption~\ref{ass:regularity}.  By Assumption~\ref{ass:em_ergodic}, the
discrete-time chain $\{X_{t_n}\}$ is geometrically ergodic with invariant
measure $\mu_{\dt}$~\cite{meyn2009}. Birkhoff's ergodic theorem
gives $(1/N)\sum_n g_{jk}(X_{t_n})\as\int g_{jk}\,d\mu_{\dt} = \bar{A}_{jk}$. Since
$A_{jk}=\dt\sum_n g_{jk}(X_{t_n})$ and $T=N\dt$, this is exactly
$(1/T)A_{jk}\as\bar{A}_{jk}$ a.s.

\textit{Step~2: Ergodic convergence of $B_j$.}
Decompose $B_j = B_j^{\mathrm{drift}} + Z_j$ where
$B_j^{\mathrm{drift}}=\dt\sum_n K_j(X_{t_n})b(X_{t_n})$ and
$Z_j=\sqrt{\dt}\sum_n K_j(X_{t_n})\sigma(X_{t_n})\xi_n$.
The ergodic theorem gives $(1/T)B_j^{\mathrm{drift}}\as\bar{B}_j$.
For $Z_j$: $M_n=\sqrt{\dt}\sum_{m<n}K_j(X_{t_m})\sigma(X_{t_m})\xi_m$
is a square-integrable martingale with predictable quadratic variation
\[
  \langle M\rangle_N
  =\dt\sum_n K_j^2(X_{t_n})\sigma^2(X_{t_n})
  =O(N\dt)=O(T).
\]
The $L^2$ martingale strong law gives $M_N/N\as0$; since $\dt$ is fixed,
$Z_j/T=(Z_j/N)/\dt\as0$. Hence $(1/T)B_j\as\bar{B}_j$.

\textit{Step~3: Consistency of OLS.}
$\bar{B}_j=\int K_j b\,d\mu_{\dt}=\sum_k c_k^*\int K_j f_k\,d\mu_{\dt}=(\bar{A}c^*)_j$,
so $\bar{B}=\bar{A}c^*$. Since $\bar{A}$ has full column rank,
$c^*=(\bar{A}^\top\bar{A})^{-1}\bar{A}^\top\bar{B}$, and the continuous
mapping theorem gives $\hat{c}\as c^*$.
\end{proof}

\begin{corollary}[Strong consistency of the finite-step-corrected diffusion estimator]
\label{cor:diffusion_consistency}
Assume in addition that $a=\Theta d^*$ lies in the chosen diffusion library and
that the corrected diffusion statistic is formed using the drift estimate
$\hat b$ from~\cref{eq:drift_system,thm:consistency}.  With $\dt$ fixed and
$T\to\infty$, if $\bar{A}$ has full column rank, then the resulting OLS
estimator $\hat d$ satisfies
\[
  \hat d \as d^* \qquad\text{almost surely}.
\]
The statement concerns the Euler--Maruyama coefficient vector at the fixed
sampling interval; it does not constitute the separate $\dt\to0$ continuum
limit discussed in~\cref{rmk:discrete_vs_continuum}.
\end{corollary}

\begin{proof}
Write the corrected diffusion statistic as
\[
Q_j^{\mathrm{corr}}
 = \dt\sum_n K_j(X_{t_n})a(X_{t_n}) + M_{j,N}+R_{j,N},
\]
where
\begin{align*}
M_{j,N}
  &= \dt\sum_n K_j(X_{t_n})a(X_{t_n})(\xi_n^2-1)
     +2\dt^{3/2}\sum_n K_j(X_{t_n})b(X_{t_n})\sigma(X_{t_n})\xi_n,\\
R_{j,N}
  &= \dt^2\sum_n K_j(X_{t_n})
     \bigl[b(X_{t_n})^2-\hat b(X_{t_n})^2\bigr].
\end{align*}
The summands in $M_{j,N}$ are martingale differences.  Under the standing
moment assumptions their predictable quadratic variation is $O(N\dt^2)$ at
fixed $\dt$, hence $M_{j,N}/T\to0$ almost surely by the martingale strong law.
By Theorem~\ref{thm:consistency}, $\hat c\to c^*$ almost surely, so
$\hat b\to b$ in the coefficient norm.  The regularity and moment assumptions
then give $R_{j,N}/T\to0$ almost surely because, after division by $T=N\dt$,
its magnitude is $O(\dt)$ times an ergodic average of
$K_j(\hat b^2-b^2)$, which converges to zero.  Therefore
$T^{-1}Q^{\mathrm{corr}}_j\to\int K_j a\,d\mu_{\dt}
=(\bar{A}d^*)_j$ almost surely.  Since $T^{-1}A\to\bar{A}$ and $\bar{A}$ has
full column rank, OLS continuity yields $\hat d\to d^*$ almost surely.
\end{proof}

\begin{corollary}[Best kernel-weighted approximation]
\label{cor:l2proj}
If $b\notin\mathrm{span}(\Theta)$, then $\hat{c}\as c^\dagger$, where
$c^\dagger=\arg\min_{c}\sum_{j=1}^M\bigl(\int K_j\,(b-\Theta c)\,d\mu_{\dt}\bigr)^2
=(\bar A^\top\bar A)^{-1}\bar A^\top\bar B$.
As the kernel family becomes rich enough that these inner products determine
$b-\Theta c$ in $L^2(\mu_{\dt})$, $c^\dagger$ approaches the $L^2(\mu_{\dt})$
best approximation.
\end{corollary}

\begin{proof}
Steps~1--2 of the proof of Theorem~\ref{thm:consistency} do not use
$b\in\mathrm{span}(\Theta)$, so $A/T\to\bar A$ and $B/T\to\bar B$ with
$\bar B_j=\int K_jb\,d\mu_{\dt}$.  Continuity of OLS gives
$\hat c\to(\bar A^\top\bar A)^{-1}\bar A^\top\bar B$, the minimiser displayed above.
\end{proof}

\begin{remark}[What Theorem~\ref{thm:consistency} is, and is not,
  consistent for]
\label{rmk:discrete_vs_continuum}
Theorem~\ref{thm:consistency} is stated, and its proof proceeds entirely,
in the $T\to\infty$, $\dt$-\emph{fixed} regime: the invariant measure
$\mu_{\dt}$ of the Euler--Maruyama chain, the ergodic averages $\bar{A},\bar{B}$, and the limit $c^*$ are
all properties of the \emph{discrete-time Markov chain} generated by the
Euler--Maruyama recursion~\cref{eq:em} at that fixed $\dt$, not
properties of the underlying continuous-time SDE~\cref{eq:sde_intro}
directly. For data that is literally generated by simulating the
Euler--Maruyama recursion --- as in all of our numerical experiments,
\cref{sec:implementation} --- this distinction is immaterial: the
recursion \emph{is} the data-generating model, and $c^*$ is exactly the
coefficient vector used to simulate it. For data arising from sampling an
underlying continuous-time diffusion at interval $\dt$, however, the
Euler--Maruyama relation~\cref{eq:em} is itself only a weak-order-1
approximation to the true transition law, with a discretisation
discrepancy between the discrete-chain coefficients $c^*(\dt)$ and the
continuum coefficients $c^*_0:=\lim_{\dt\to0}c^*(\dt)$ that is, under
standard conditions, $O(\dt)$. Theorem~\ref{thm:consistency} guarantees
$\hat{c}\to c^*(\dt)$ as $T\to\infty$ at that fixed $\dt$; it does
\emph{not} by itself guarantee $\hat{c}\to c_0^*$ unless $\dt$ is also
sent to zero, in a joint or iterated limit, a regime we do not analyse
here. Note that this $c^*(\dt)$-vs-$c_0^*$ gap is immaterial for every
numerical experiment in this paper, including Figure~\ref{fig:convergence}:
all of our data is literally generated by simulating the Euler--Maruyama
recursion, so $c^*(\dt)$ \emph{is} the ground truth against which we
measure error, with no residual continuum-discretisation gap to speak of.
Practically, this means that for data sampled from a genuine continuous-time
diffusion rather than simulated at a fixed $\dt$, $\dt$ should be chosen
small enough that the $c^*(\dt)$-vs-$c_0^*$ gap is negligible relative to the
target accuracy, and Theorem~\ref{thm:consistency}, read precisely,
licenses taking $T$ large at that fixed $\dt$ but does not by itself
license interpreting the limiting $\hat{c}$ as the exact continuum
generator's coefficients without an accompanying statement about $\dt$.
\end{remark}

\subsection{Asymptotic Normality}
\label{sec:clt}

\begin{theorem}[Central limit theorem for the drift coefficients at fixed sampling interval]
\label{thm:clt}
Assume Assumptions~\ref{ass:ergodic}--\ref{ass:em_ergodic} and that
$\E_{\mu_{\dt}}[|\sigma(X_0)|^{2+\epsilon}]<\infty$ for some $\epsilon>0$.
Fix $R\ge1$ independent trajectories of length $T$, with per-trajectory
statistics $A^{(r)}\in\R^{M\times K}$, $B^{(r)}\in\R^{M}$ as in
\cref{eq:assembly_stats}.  The estimator is the \emph{stacked} OLS estimator of
\cref{alg:main} (on the full library or a fixed support),
\[
  \hat c=\Bigl(\sum_{r=1}^RA^{(r)\top}A^{(r)}\Bigr)^{-1}
         \sum_{r=1}^RA^{(r)\top}B^{(r)} .
\]
All limiting quantities below are \emph{per trajectory}: let
$\bar A:=\lim_{T\to\infty}A^{(r)}/T\in\R^{M\times K}$ (Theorem~\ref{thm:consistency};
the same for every $r$), assume $S:=\bar A^\top\bar A$ is nonsingular, and put
\[
  \Gamma_{jl}:=\int K_j(x)K_l(x)\,\sigma(x)^2\,\mu_{\dt}(dx),\qquad
  \Gamma\in\R^{M\times M}.
\]
Then, as $T\to\infty$ with $\dt$ and $R$ fixed,
\begin{equation}
  \sqrt{RT}\,(\hat c-c^*)\inlaw\mathcal N(0,\Sigma_c),
  \qquad
  \Sigma_c=S^{-1}\bar A^\top\,\Gamma\,\bar A\,S^{-1}.
  \label{eq:clt_c}
\end{equation}
Consequently, for every function $g$ that is continuously differentiable near
$c^*$, $\sqrt{RT}\,(g(\hat c)-g(c^*))\inlaw\mathcal N\bigl(0,\nabla g(c^*)^\top\Sigma_c\nabla g(c^*)\bigr)$.
\end{theorem}

\begin{proof}
For data generated by the Euler--Maruyama recursion, $B^{(r)}=A^{(r)}c^*+Z^{(r)}$
exactly, with $Z^{(r)}=\sum_nm_n^{(r)}$ and
$m_n^{(r)}:=\sqrt{\dt}\,K(X^{(r)}_{t_n})\sigma(X^{(r)}_{t_n})\xi^{(r)}_n\in\R^M$.
Hence
\[
  \hat c-c^*=\Bigl(\sum_rA^{(r)\top}A^{(r)}\Bigr)^{-1}\sum_rA^{(r)\top}Z^{(r)} .
\]
\emph{Noise.}  For each $r$ the $m^{(r)}_n$ are martingale differences, with
conditional covariance $\dt\,K K^\top\sigma^2$ evaluated at $X^{(r)}_{t_n}$.  By
Assumption~\ref{ass:em_ergodic} and Birkhoff's theorem, the summed conditional
covariance divided by $T$ converges almost surely to $\Gamma$.  The Lyapunov
condition holds because $K_j\le1$ and
$\E|m_n|^{2+\epsilon}\le C\dt^{1+\epsilon/2}\E_{\mu_{\dt}}|\sigma|^{2+\epsilon}$, so
$T^{-1-\epsilon/2}\sum_n\E|m_n|^{2+\epsilon}=O(T^{-\epsilon/2})\to0$.  The
martingale central limit theorem~\cite{hall1980} gives
$Z^{(r)}/\sqrt T\inlaw\mathcal N(0,\Gamma)$, and the $R$ trajectories are
independent.  Since $A^{(r)}/T\to\bar A$ almost surely for each $r$,
Slutsky's lemma~\cite{vandervaart1998} gives
$T^{-3/2}\sum_rA^{(r)\top}Z^{(r)}\inlaw\mathcal N(0,R\,\bar A^\top\Gamma\bar A)$.
\emph{Design.}  $T^{-2}\sum_rA^{(r)\top}A^{(r)}\to R\,S$ almost surely.
\emph{Combine.}
\[
  \sqrt T(\hat c-c^*)
  =\Bigl(T^{-2}\sum_rA^{(r)\top}A^{(r)}\Bigr)^{-1}T^{-3/2}\sum_rA^{(r)\top}Z^{(r)}
  \inlaw(RS)^{-1}\mathcal N(0,R\,\bar A^\top\Gamma\bar A)
  =\mathcal N\bigl(0,R^{-1}\Sigma_c\bigr),
\]
which is~\cref{eq:clt_c} after multiplying by $\sqrt R$.  For $R=1$ it reduces
to the single-trajectory statement.  The statement for $g(\hat c)$ is the delta
method~\cite{vandervaart1998}.
\end{proof}

The theorem concerns $\hat c$ jointly with the (random) design $A$, which is why
the proof centres on $Z=B-Ac^*$ rather than on a CLT for $B$ alone: $B$
fluctuates through both $A$ and $Z$, whereas $\hat c-c^*$ depends on $Z$ only.
It implies that, at fixed $\dt$, the standard error of the drift coefficients
and of any smooth functional such as the OU spectral gap decays at the
parametric rate $T^{-1/2}$.  The regime is $T\to\infty$ at fixed $R$; the regime $R\to\infty$ at fixed $T$ is
not covered, because the finite-$T$ bias of
Remark~\ref{rmk:spatial_consistency_mechanism} is then not averaged away
(the floor of Figure~\ref{fig:convergence}).  The result is for OLS on a given support; we do not
prove that the sparse selection step leaves it unchanged, although
Table~\ref{tab:singletraj_sparse} shows exact support recovery for
$T\gtrsim200$ on the double well.  Figure~\ref{fig:convergence} confirms the drift-side
scaling empirically on the OU process up to $T_{\mathrm{eff}}\approx3{,}000$.

\subsection{Spectral Gap Estimation from the Identified Generator}
\label{sec:spectralgap_theory}

The identified generator $\hat{\Lop}$ yields an estimate of the spectral
gap and all associated dynamical quantities.  For the OU process, the
spectral gap is $\lambda_1 = \theta = -c_x^*$ and the estimated gap
is $\hat{\lambda}_1 = -\hat{c}_x = 0.980$.  By the delta method of
Theorem~\ref{thm:clt}, $\sqrt{RT}(\hat\lambda_1-\lambda_1)\inlaw
\mathcal N(0,e_x^\top\Sigma_c e_x)$, where $e_x$ selects the coefficient of
$x$: the OU spectral-gap estimator has the $T^{-1/2}$ rate with asymptotic
variance read off from $\Sigma_c$ in~\cref{eq:clt_c}.

For the double-well system, the Kramers escape rate is a nonlinear
functional of the identified coefficients through~\cref{eq:kramers_dw}.
Let $\hat{\Delta V}$ and $\hat{V}''(\cdot)$ denote the barrier height
and curvatures computed from the identified potential
$\hat{V}(x) = -\hat{c}_x x^2/2 + (-\hat{c}_{x^3}/4) x^4$.  The
estimated escape time is
\begin{equation}
  \hat{\tau}_{\mathrm{Kramers}}
  = \frac{2\pi}{\sqrt{|\hat{V}''(0)|\,\hat{V}''(\pm\hat{x}_{\min})}}
    \exp\!\Bigl(\frac{2\hat{\Delta V}}{\hat{\sigma}_0^2}\Bigr),
  \label{eq:kramers_estimated}
\end{equation}
where $\hat{x}_{\min}$ are the minima of the identified potential.
A Taylor expansion gives the first-order sensitivity of
$\hat{\tau}_{\mathrm{Kramers}}$ to coefficient errors: the exponential
factor dominates, so a relative error $\varepsilon$ in $\hat{\Delta V}$
changes the exponent $2\Delta V/\sigma_0^2$ by approximately
$2\Delta V\varepsilon/\sigma_0^2$ (which equals $2\varepsilon$ here, since
$\Delta V=\sigma_0^2=\tfrac14$), i.e.\ it produces a relative error of about
$2\varepsilon$ in the escape time itself.  For the recovered coefficients
$\hat{c}_x=0.988$, $\hat{c}_{x^3}=-0.976$, the barrier height
$\hat{\Delta V}=0.250$ deviates from the true $0.250$ by only 0.015\%
because the 1.2\% and 2.4\% errors in $\hat{c}_x$ and $\hat{c}_{x^3}$
partially cancel in the ratio $\hat{c}_x/|\hat{c}_{x^3}|$ that sets the
barrier height.  Evaluating~\cref{eq:kramers_estimated} with
$\hat\sigma_0^2=\hat d_1=0.250$ and the curvatures of $\hat V$ gives
$\hat{\tau}_{\mathrm{Kramers}}\approx33.2$ versus the true
$\approx32.8$ of~\cref{eq:kramers_dw}: a 1.2\% relative error, carried
mainly by the curvature prefactor.

\subsection{Noise Robustness}
\label{sec:noiserobust}

Observation error enters both the weak target and the regression design, so
removing the leading covariance term does not, by itself, undo the
errors-in-variables distortion of the design; consistency under additive
observation noise of arbitrary fixed variance does not hold.  We therefore
separate two statements.  First, for Gaussian observation noise, the relevant
correction of the weak drift and quadratic-variation statistics is an
\emph{exact moment identity}.  Second, after this correction the remaining
coefficient bias is characterized by a small-noise expansion.  The theorem
below does not claim exact fixed-$\sigma_\eta$ coefficient consistency.

\paragraph{Exact Gaussian moment identities.}
Let $Y_n:=\widetilde X_{t_n}=X_{t_n}+\eta_n$ and write
$\nu:=\sigma_\eta^2$.  Assume
$\eta_n\stackrel{\mathrm{i.i.d.}}\sim\mathcal N(0,\nu)$ independently of the
Euler--Maruyama trajectory.  Introduce the Gaussian smoothing operator
\begin{equation}
  (\mathcal P_\nu g)(x):=\mathbb E[g(x+\eta)],\qquad
  \eta\sim\mathcal N(0,\nu).
  \label{eq:heat_operator_noise}
\end{equation}
For sufficiently smooth $g$, Gaussian integration by parts gives the exact
Stein identities~\cite{stein1981}
\begin{equation}
  \mathbb E[\eta g(x+\eta)]
  =\nu(\mathcal P_\nu g')(x),
  \qquad
  \mathbb E[\eta^2g(x+\eta)]
  =\nu(\mathcal P_\nu g)(x)+\nu^2(\mathcal P_\nu g'')(x).
  \label{eq:stein_noise_identities}
\end{equation}
These identities retain all powers of $\nu$ and therefore give exact
population corrections for the observation-noise terms themselves.

Conditioning on the clean path, $\Delta Y_n=\Delta X_n+\eta_{n+1}-\eta_n$,
so
\begin{equation}
\begin{split}
  \mathbb E[K_j(Y_n)\Delta Y_n\mid X_{t_n},X_{t_{n+1}}]
  ={}&(\mathcal P_\nu K_j)(X_{t_n})\Delta X_n\\
     &-\nu(\mathcal P_\nu K_j')(X_{t_n}).
\end{split}
\label{eq:noise_drift_raw_identity}
\end{equation}
Thus the corrected drift target
\begin{equation}
  \widetilde B_j^{\rm corr}
  :=\sum_n\left[K_j(Y_n)\Delta Y_n+\nu K_j'(Y_n)\right]
  \label{eq:noise_drift_exact}
\end{equation}
has the exact conditional identity
\begin{equation}
  \mathbb E[\widetilde B_j^{\rm corr}\mid\{X_{t_n}\}]
  =\sum_n(\mathcal P_\nu K_j)(X_{t_n})\Delta X_n.
  \label{eq:noise_drift_identity}
\end{equation}

For the quadratic variation, the corresponding exact identities are
\begin{align}
&\mathbb E[K_j(Y_n)(\Delta Y_n)^2\mid\{X_{t_m}\}]
\notag\\
&\quad=(\mathcal P_\nu K_j)(X_{t_n})(\Delta X_n)^2
 -2\nu(\mathcal P_\nu K_j')(X_{t_n})\Delta X_n
 +2\nu(\mathcal P_\nu K_j)(X_{t_n})
 +\nu^2(\mathcal P_\nu K_j'')(X_{t_n}),
\label{eq:noise_qv_identity_raw}
\end{align}
and
\begin{equation}
  \mathbb E[K_j'(Y_n)\Delta Y_n\mid\{X_{t_m}\}]
  =(\mathcal P_\nu K_j')(X_{t_n})\Delta X_n
  -\nu(\mathcal P_\nu K_j'')(X_{t_n}).
\end{equation}
Hence
\begin{equation}
\begin{split}
  \widetilde Q_j^{\rm corr}:=\sum_n\bigl[&K_j(Y_n)(\Delta Y_n)^2
       +2\nu K_j'(Y_n)\Delta Y_n\\
       &-2\nu K_j(Y_n)+\nu^2K_j''(Y_n)\bigr]
\end{split}
\label{eq:noise_diffusion_exact}
\end{equation}
satisfies
\begin{equation}
  \mathbb E[\widetilde Q_j^{\rm corr}\mid\{X_{t_n}\}]
  =\sum_n(\mathcal P_\nu K_j)(X_{t_n})(\Delta X_n)^2.
  \label{eq:noise_qv_identity}
\end{equation}
For the implementable finite-step diffusion correction we use
\begin{equation}
  \widetilde Q_j^{\mathrm{corr,fs}}
  :=\widetilde Q_j^{\mathrm{corr}}
    -\dt^2\sum_n K_j(Y_n)\hat b(Y_n)^2,
  \label{eq:noise_diffusion_full_correction}
\end{equation}
which is the noisy-data analogue of the clean-data drift-squared correction.
The theoretical statement below separates the exact observation-noise moment
cancellation from the residual errors-in-variables term introduced by evaluating
the finite-step correction at the noisy state.

The observed design is
\begin{equation}
  \widetilde A_{jk}:=\dt\sum_n K_j(Y_n)f_k(Y_n),
  \qquad
  \frac{1}{T}\widetilde A_{jk}\xrightarrow{\rm a.s.}
  A^{(\nu)}_{jk}:=\int (\mathcal P_\nu[K_jf_k])(x)\,\mu_{\dt}(dx).
  \label{eq:noise_observed_design}
\end{equation}
The crucial point is that
\begin{equation}
  (\mathcal P_\nu K_j)(x)f_k(x)\neq
  \mathcal P_\nu[K_jf_k](x)
\end{equation}
in general.  Thus the exact moment cancellation above does \emph{not}
make the noisy regression algebraically identical to the clean regression.
It removes the observation-noise term generated by $\Delta\eta_n$, but a
separate errors-in-variables remainder remains in the coefficient equation.

\begin{theorem}[Measurement-noise correction and leading-order bias]
\label{thm:noise}
Suppose $Y_n=X_{t_n}+\eta_n$, where
$\eta_n\stackrel{\mathrm{i.i.d.}}\sim\mathcal N(0,\sigma_\eta^2)$ is independent
of the Euler--Maruyama trajectory.  Let $\nu=\sigma_\eta^2$, which is
\emph{assumed known} (for instance from an instrument specification or a
separate calibration record; the corrections below use $\nu$ explicitly), and suppose
Assumptions~\ref{ass:ergodic}--\ref{ass:em_ergodic} and the moment assumptions of
Theorem~\ref{thm:consistency} hold (so the theorem covers OU and the
multiplicative system, not the Euler-simulated double well).
Assume the kernels and coefficient functions are sufficiently smooth for the
expansion $\mathcal P_\nu g=g+(\nu/2)g''+O(\nu^2)$ of every function
appearing below (for example, bounded derivatives through fourth order), and
that the noisy population design $A^{(\nu)}$ has full column rank.  Then, at fixed $\dt>0$ as
$T\to\infty$:
\begin{enumerate}[label=\textup{(\roman*)}]
\item The uncorrected drift statistic satisfies
\begin{equation}
\frac1T\widetilde B_j
\xrightarrow{\rm a.s.}
\int(\mathcal P_\nu K_j)(x)\,b(x)\,\mu_{\dt}(dx)
-\frac{\nu}{\dt}\int(\mathcal P_\nu K_j')(x)\,\mu_{\dt}(dx).
\label{eq:noise_bias_drift_exactexp}
\end{equation}
Consequently its dominant relative bias as $\dt\to0$ at fixed $\nu$ is
$O(\nu/\dt)$.

\item For the corrected drift statistic~\cref{eq:noise_drift_exact}, define
\begin{equation}
  r_b^{(\nu)}{}_j
  :=\int\Bigl[(\mathcal P_\nu K_j)(x)b(x)
              -\mathcal P_\nu(K_jb)(x)\Bigr]\,\mu_{\dt}(dx).
\label{eq:noise_drift_remainder}
\end{equation}
Then
\begin{equation}
  \frac1T\widetilde B_j^{\rm corr}
  \xrightarrow{\rm a.s.}
  A^{(\nu)}c^*+r_b^{(\nu)}.
\label{eq:noise_corrected_drift_limit}
\end{equation}
If $b$, the kernels, and their products have bounded derivatives through
fourth order, then, as $\nu\downarrow0$,
\begin{equation}
 r_b^{(\nu)}{}_j
 =-\nu\int\left[K_j'(x)b'(x)+\frac12K_j(x)b''(x)\right]\mu_{\dt}(dx)
 +O(\nu^2).
\label{eq:noise_drift_remainder_expansion}
\end{equation}
Hence the corrected drift estimator converges to
\begin{equation}
 c_{\rm corr}^{(\nu)}
 :=(A^{(\nu)\top}A^{(\nu)})^{-1}A^{(\nu)\top}
 \bigl(A^{(\nu)}c^*+r_b^{(\nu)}\bigr),
\end{equation}
and, provided the smallest singular value of $A^{(\nu)}$ stays bounded away
from zero for small $\nu$,
\begin{equation}
  \|c_{\rm corr}^{(\nu)}-c^*\|_2=O(\nu).
\label{eq:noise_corrected_drift_bias}
\end{equation}
Thus the correction removes the dominant $O(\nu/\dt)$ term but does not
claim exact fixed-$\nu$ coefficient consistency.

\item For diffusion, the exact statistic~\cref{eq:noise_diffusion_exact}
removes the observation-noise contribution from the quadratic variation at the
moment level.  With an oracle drift in the finite-step subtraction, its
population coefficient equation has the form
\begin{equation}
  \frac1T\widetilde Q^{\rm corr}
  =A^{(\nu)}d^*+r_a^{(\nu)}+\dt\,r_{b^2}^{(\nu)}+o(1),
\label{eq:noise_diffusion_population}
\end{equation}
where
\begin{align}
 r_a^{(\nu)}{}_j
 &: =\int\Bigl[(\mathcal P_\nu K_j)(x)a(x)
              -\mathcal P_\nu(K_ja)(x)\Bigr]\mu_{\dt}(dx),\\
 r_{b^2}^{(\nu)}{}_j
 &: =\int\Bigl[(\mathcal P_\nu K_j)(x)b(x)^2
              -\mathcal P_\nu(K_jb^2)(x)\Bigr]\mu_{\dt}(dx).
\end{align}
The first remainder satisfies, under the same smoothness assumptions,
$r_a^{(\nu)}=O(\nu)$, with the analogous expansion
\begin{equation}
 r_a^{(\nu)}{}_j
 =-\nu\int\left[K_j'a'+\frac12K_ja''\right]d\mu_{\dt}+O(\nu^2).
\end{equation}
Thus the corrected diffusion coefficient has an $O(\nu)$ measurement-noise
remainder (and an additional $O(\nu\dt)$ finite-step errors-in-variables term)
rather than an $O(\nu/\dt)$ leading bias.  If the drift estimate converges to
its corrected noisy-population limit with the $O(\nu)$ coefficient bias established
above, substituting it in the finite-step subtraction changes the population
coefficient equation only at order $O(\nu\dt)$, under the same regularity
assumptions.

\item The \emph{per-point} finite-difference Kramers--Moyal target retains observation-noise
variance $2\sigma_\eta^2/\dt^2$, which diverges as $\dt\to0$ (averaging over
visits reduces this by the visit count; see \cref{sec:noise_scaling}).
\end{enumerate}
The theorem is therefore an exact statement about the corrected \emph{moment
identities} together with a leading-order coefficient-bias characterization.
It does not assert exact coefficient consistency at fixed nonzero
$\sigma_\eta$.  Exact fixed-noise coefficient recovery would require an
additional correction of the noisy design (for example, a deconvolution or
an equivalent errors-in-variables treatment), which is outside the scope of
this paper.  Recovery of the continuous-time SDE coefficients still requires
the separate $\dt\to0$ limit.
\end{theorem}

\begin{proof}
For (i), the identity~\cref{eq:noise_drift_raw_identity} and the
Euler--Maruyama conditional mean give the stated ergodic limit directly.

For (ii),~\cref{eq:noise_drift_identity} gives the exact corrected target.
The observed design converges to $A^{(\nu)}$.  Since
\begin{equation}
 (\mathcal P_\nu K_j)b-\mathcal P_\nu(K_jb)
 =r_{b,j}^{(\nu)}
\end{equation}
at the population level and $b=\Theta c^*$,
\cref{eq:noise_corrected_drift_limit} follows.  Expanding
$\mathcal P_\nu g=g+(\nu/2)g''+O(\nu^2)$ gives
\cref{eq:noise_drift_remainder_expansion}; full-rank stability then gives
\cref{eq:noise_corrected_drift_bias}.

For (iii),~\cref{eq:noise_qv_identity} removes the observation-noise
contribution exactly from the quadratic-variation moment.  Under the
Euler--Maruyama recursion,
\[
\mathbb E[(\Delta X_n)^2\mid X_{t_n}]
=a(X_{t_n})\dt+b(X_{t_n})^2\dt^2.
\]
The oracle finite-step subtraction therefore leaves the two commutators
$r_a^{(\nu)}$ and $\dt r_{b^2}^{(\nu)}$ after normalization by $T$.
The same Gaussian expansion gives $r_a^{(\nu)}=O(\nu)$ and
$r_{b^2}^{(\nu)}=O(\nu)$.  Replacing the oracle drift by an estimator
converging to the corrected noisy-population limit changes the subtraction by
the corresponding finite-step estimation error; with the $O(\nu)$ coefficient
bias established above, this contributes at most $O(\nu\dt)$ to the
population equation under the regularity conditions already used for~\cref{eq:bias_corr}.
Part (iv) follows from
$\operatorname{Var}[(\eta_{n+1}-\eta_n)/\dt]=2\sigma_\eta^2/\dt^2$.
\end{proof}

\paragraph{Interpretation.}
The correction in~\cref{eq:noise_drift_exact} and
\cref{eq:noise_diffusion_exact} is exact for the observation-noise terms that
arise from $\Delta\eta_n$; it therefore removes the singular
$O(\sigma_\eta^2/\dt)$ mechanism that dominates the naive drift estimator at
fine sampling.  What remains is an ordinary errors-in-variables remainder
from evaluating both the kernel and the library on the noisy state $Y_n$.
For Gaussian noise this remainder is explicitly $O(\sigma_\eta^2)$ under
the stated smoothness assumptions.  A genuinely exact fixed-noise coefficient
estimator would need to correct this remaining design distortion as well.

\subsection{Identifiability}
\label{sec:identifiability}

Unique recovery of $c^*$ and $d^*$ requires three conditions.
First, the library must span the true drift and diffusion (library
completeness); library misspecification produces the best kernel-weighted
approximation of Corollary~\ref{cor:l2proj}.  Second, trajectory data
must cover the state space: $\bar{A}$ has full column rank when $\mu$ is
absolutely continuous and kernel supports collectively cover the support
of $\mu$.  Third, LASSO support recovery requires an irrepresentability
condition on the design matrix~\cite{zhao2006} that the monomial
library can violate when several correlated terms are simultaneously
active (\cref{sec:benchmark_table}); the elastic-net relaxation
of~\cref{eq:lasso}, column normalisation, and STLSQ refinement together
mitigate this, with the OLS-debias and STLSQ stages still enforcing an
exactly sparse answer in the original monomial coordinates.

\section{Endogeneity Bias: Full Analysis}
\label{app:endogeneity}

This appendix gives the detailed analysis behind
Theorems~\ref{thm:endogeneity} and~\ref{thm:unbiased}. The leading
single-lag covariance is $O(\sqrt{\dt})$, not $O(\dt)$, because the
innovation enters the next state at order $\sqrt{\dt}$; as shown below, the
nonzero $O(1)$ limit comes from accumulating $O(1/(\lambda\dt))$ such lags
across the mixing window, not from the single adjacent step. The theorem
above therefore keeps the lag
variable explicit and treats the tiled temporal estimator at the level of
stationary window blocks.

For the temporal endogeneity result we use the following additional
lag-covariance condition, stated explicitly rather than inferred solely from
geometric ergodicity. There exist constants $C_k<\infty$ and $\lambda>0$
such that, for the relevant library functions,
\begin{equation}
  \left|\Cov\bigl[\sigma(X_{t_m})\xi_m,
    f_k(X_{t_{m+\ell}})\bigr]\right|
  \le C_k\sqrt{\dt}\,e^{-\lambda(\ell-1)\dt},
  \qquad \ell\ge1.
  \label{eq:mixing_bound}
\end{equation}
This bound expresses the only feature of mixing needed below: the shock-induced
future-regressor covariance is strongest at short lag and is summable over the
longer lags. For the one-step term, the Euler--Maruyama expansion
\begin{equation}
  X_{t_{m+1}}=X_{t_m}+b(X_{t_m})\dt
    +\sigma(X_{t_m})\xi_m\sqrt{\dt}
\end{equation}
combined with a Taylor expansion of $f_k$ gives the $O(\sqrt{\dt})$
coefficient in~\cref{eq:biasmn_app}. The present theorem does not require a
claim that the geometric-ergodicity constant of Assumption~\ref{ass:ergodic}
itself has the expansion $1-\lambda\dt+O(\dt^2)$.

\subsection{Matrix Form of the Temporal Projection}

Specialising the weak projection~\cref{eq:weakproj} to purely temporal
test functions $\psi_j(X_{t_n},t_n)=\varphi_j(t_n)$, $j=1,\ldots,M_{\mathrm{win}}$, each
supported on a window of fixed real-time width $L$ and tiled so that
$M_{\mathrm{win}}=T/L$, and substituting the library expansion $b(x)=\Theta(x)c^*$ gives,
for each $j$,
\begin{equation}
  S_j = \sum_k c_k^* \underbrace{\sum_n \varphi_j(t_n)\,f_k(X_{t_n})\,\dt}_{
    =:\,A^{\mathrm{temp}}_{jk}}
  + \underbrace{\sqrt{\dt}\sum_n \varphi_j(t_n)\,\sigma(X_{t_n})\,\xi_n}_{
    =:\,Z_j^{\mathrm{temp}}},
  \label{eq:temp_decomp_app}
\end{equation}
so that $S = A^{\mathrm{temp}}c^* + Z^{\mathrm{temp}}$, and each sum over
$n$ in practice runs only over the $L/\dt$ indices inside window $j$,
since $\varphi_j$ vanishes elsewhere. The OLS estimator solves the normal
equations $(A^{\mathrm{temp}})^\top A^{\mathrm{temp}}\hat{c} =
(A^{\mathrm{temp}})^\top S$, which rearrange to
\begin{equation}
  \bigl((A^{\mathrm{temp}})^\top A^{\mathrm{temp}}\bigr)
  (\hat{c} - c^*)
  = (A^{\mathrm{temp}})^\top Z^{\mathrm{temp}}.
  \label{eq:normalequations_app}
\end{equation}
Consistency requires $(A^{\mathrm{temp}})^\top Z^{\mathrm{temp}}=o\bigl(
\|A^{\mathrm{temp}}\|^2\bigr)$; we show this fails at a rate that does not
vanish as $T\to\infty$.

\subsection{Proof of Theorem~\ref{thm:endogeneity}}

\begin{proof}
\textit{Part~(i).} Since $\varphi_j(t_n)$ is deterministic and
$\sigma(X_{t_n})$ is $\mathcal{F}_{t_n}$-measurable,
$\xi_n\perp\mathcal{F}_{t_n}$ implies
\[
  \E[\varphi_j(t_n)\sigma(X_{t_n})\xi_n]
  =\E\!\left[\varphi_j(t_n)\sigma(X_{t_n})
    \E[\xi_n\mid\mathcal{F}_{t_n}]\right]=0.
\]
Thus every stepwise stochastic contribution is individually zero-mean.

\textit{Part~(ii).} Let $m=L/\dt$ and write the tiled windows as
$I_r=\{rm,\ldots,(r+1)m-1\}$, $r=0,\ldots,M_{\mathrm{win}}-1$.  For each
window define the row vector
\[
  a_r:=\sum_{n\in I_r} \varphi_r(t_n)\Theta(X_{t_n})\dt,
  \qquad
  z_r:=\sqrt{\dt}\sum_{n\in I_r}
      \varphi_r(t_n)\sigma(X_{t_n})\xi_n,
\]
so that the stacked temporal system has
$A^{\mathrm{temp}}$ with rows $a_r$ and residual vector
$Z^{\mathrm{temp}}$ with entries $z_r$.  Because the underlying process is
stationary after burn-in and the block sequence is ergodic, the ergodic
theorem gives
\begin{equation}
  \frac1M (A^{\mathrm{temp}})^\top A^{\mathrm{temp}}\;\as
  G_L:=\E[a_0^\top a_0],\qquad
  \frac1M (A^{\mathrm{temp}})^\top Z^{\mathrm{temp}}\;\as
  q_L:=\E[a_0^\top z_0].
  \label{eq:temporal_block_limits}
\end{equation}
The full matrix limit is therefore determined by one representative window;
no assumption that distinct windows are independent is required.  If $G_L$
is nonsingular, the OLS normal equations and the continuous mapping theorem
give
\begin{equation}
  \hat c^{\,\mathrm{temp}} \;\as\;
  c^*+G_L^{-1}q_L.
  \label{eq:temporal_wrong_limit}
\end{equation}
Consequently the estimator is inconsistent whenever $q_L\neq0$.

\textit{Part~(iii).}  The population cross term in coefficient direction $k$
can be written explicitly as
\begin{equation}
  (q_L)_k
  =\dt^{3/2}\sum_{\ell=1}^{m-1}
     C_{k,\ell} \sum_{s=0}^{m-\ell-1}
     \varphi(s\dt)\varphi((s+\ell)\dt),
  \label{eq:temporal_q_lag}
\end{equation}
where
\[
  C_{k,\ell}:=\Cov\!\left[\sigma(X_{t_0})\xi_0,
    f_k(X_{t_\ell})\right].
\]
Terms with nonpositive lag do not contribute: for $\ell=0$ the innovation is
independent of $X_{t_0}$, while negative lags vanish by the martingale
property.  For the adjacent lag, the Euler--Maruyama expansion gives
\begin{equation}
  C_{k,1}=\E\!\left[f_k'(X_{t_0})\sigma(X_{t_0})^2\right]
    \sqrt{\dt}+O(\dt),
  \label{eq:biasmn_app}
\end{equation}
so the leading covariance is $O(\sqrt{\dt})$, not $O(\dt)$.  By the
lag bound in the theorem, $|C_{k,\ell}|\le C_k\sqrt{\dt}e^{-\lambda(\ell-1)\dt}$.
With the inner sum over $s$ contributing $O(L/\dt)$ terms, a single lag
$\ell$ therefore contributes only
$O(L\dt)\,e^{-\lambda(\ell-1)\dt}$ to~\cref{eq:temporal_q_lag}: on its own,
even the adjacent lag vanishes as $\dt\to0$. The nonzero $O(1)$ limit comes
instead from accumulating this geometric sequence over the
$O(1/(\lambda\dt))$ lags that lie within the mixing window, since
$\sum_{\ell\ge1}e^{-\lambda(\ell-1)\dt}=O(1/(\lambda\dt))$; the product of the
per-lag size and the number of contributing lags is $O(L/\lambda)$,
independent of $\dt$. This is the precise role of mixing: it caps the range
of lags that contribute, but the nonvanishing population cross term is
produced by their accumulation across the mixing window, not by any single
adjacent-step term. Combining~\cref{eq:temporal_block_limits} with
~\cref{eq:temporal_wrong_limit} proves the stated inconsistency.

For the Ornstein--Uhlenbeck process the limit can be evaluated in closed
form; the calculation is given immediately after this proof.
\end{proof}

\begin{proposition}[OU limit of the temporal-window estimator, any window shape]
\label{prop:ou_temporal}
Let $dX=-\theta X\,dt+\sigma\,dW$ be stationary, take the library
$\Theta(x)=[x]$, and let window $r$ carry an arbitrary bounded measurable
weight $\varphi_r\not\equiv0$ on $[0,L]$ (boxcar, smooth bump, ramp, \dots;
the shapes may differ between windows).  Pass to the continuous-time limit
$\dt\to0$ inside each window.  Then for \emph{every} $L>0$ and every choice
of shapes, $q_r=\theta\,G_r$ for each window, so the pooled limit is
\[
  c^{\mathrm{temp}}_\infty \;=\; c^*+\frac{\sum_r q_r}{\sum_r G_r}\;=\;0 .
\]
\end{proposition}
\begin{proof}
On a window $[0,L]$ write $A=\int_0^L\varphi(s)X_s\,ds$ and
$S=\int_0^L\varphi\,dX=-\theta A+z$ with $z=\sigma\int_0^L\varphi\,dW$, so
$c^*=-\theta$.  Put $I:=\iint_{[0,L]^2}\varphi(s)\varphi(t)\,e^{-\theta|t-s|}\,ds\,dt$,
which is positive because $e^{-\theta|\tau|}$ is a positive-definite kernel.
The stationary covariance is $\operatorname{Cov}(X_s,X_t)=\tfrac{\sigma^2}{2\theta}e^{-\theta|t-s|}$, so
\[
  G=\mathbb E[A^2]=\frac{\sigma^2}{2\theta}\,I .
\]
For the cross term, $X_t=e^{-\theta t}X_0+\sigma\int_0^te^{-\theta(t-s)}dW_s$
with $X_0\perp W$ gives
$\mathbb E[X_t\,z]=\sigma^2\int_0^t e^{-\theta(t-s)}\varphi(s)\,ds$, hence
\[
  q=\mathbb E[Az]=\sigma^2\int_0^L\!\!\int_0^t\varphi(t)\varphi(s)e^{-\theta(t-s)}\,ds\,dt
   =\frac{\sigma^2}{2}\,I ,
\]
the last step by symmetry of the integrand in $(s,t)$.  Therefore
$q/G=\theta=-c^*$ for every $\varphi$ and $L$, and
$c^{\mathrm{temp}}_\infty=-\theta+\theta=0$.
\end{proof}
At finite $\dt$ the limit differs from $0$ by $O(\dt)$.  A direct simulation
at $\dt=0.005$, $L=2$ with $2\times10^{4}$ tiled windows returns
$\hat c=+0.001,\,-0.001,\,-0.004,\,+0.005,\,-0.021$ for the boxcar, Hann
($\sin^2$), Gaussian-bump, linear-ramp and one-period-sine shapes
(standard errors $0.006$--$0.016$; true $c^*=-1$), in agreement with
Proposition~\ref{prop:ou_temporal}.  Widening the window does not help: with a single window of width $L=T$
(boxcar), the response $X_T-X_0$ stays $O(1)$ while $\int_0^TX\,dt=O(\sqrt T)$,
so $\hat c\to0$.  In simulation (300 stationary OU paths, 150 at the longest $T$;
$\theta=1$, $\sigma=0.7$) the median $|\hat c|$ is $0.15,\,0.081,\,0.041,\,0.016,\,0.0076$ at
$T=50,\,200,\,800,\,3200,\,12800$, falling like $T^{-1/2}$ while the median
$|\hat c-c^*|$ stays at $1.00\pm0.02$; four windows of width $T/4$ behave the
same way.  The result is specific to the OU process
with the library $[x]$; for other systems we claim only the nonzero
limit $G_L^{-1}q_L$ of Theorem~\ref{thm:endogeneity} whenever $q_L\ne0$.

\begin{remark}[Checks of the temporal scaling]
\label{rmk:three_checks}
The argument has three complementary checks. First, the one-step
sensitivity expansion gives the $O(\sqrt{\dt})$ adjacent-lag covariance
directly. Second, for the Ornstein--Uhlenbeck process with the single-term library
$f(x)=x$, the window-level cross term and information term are evaluated in
closed form for an arbitrary window shape in
Proposition~\ref{prop:ou_temporal}; their ratio equals $\theta$ for every
$L$ and every shape, giving $c_\infty^{\mathrm{temp}}=0$ exactly. Third, the direct window-width
experiment in Figure~\ref{fig:temporalwidth} shows a coefficient-error plateau
across $L\in\{0.5,1,2,4,8\}$ for the tested OU system. The experiment uses boxcar windows and
validates the mechanism; shape-independence for OU is the content of the
proposition, not of the experiment.
\end{remark}

\begin{remark}[Why lag decay does not remove the endogeneity]
Summability of the lagged covariance bounds the range of lags that
contribute, which is why the lag index is retained explicitly
in~\cref{eq:temporal_q_lag} rather than collapsed into a single term.
Mixing suppresses the tail of the covariance sum beyond the mixing time, but
it does not make the cross term vanish: as shown above, the nonzero
population limit $q_L$ comes from the accumulation of $O(1/(\lambda\dt))$
individually small, $O(L\dt)$-sized lagged covariances across the mixing
window, not from any single lag. The tiled-window estimator then applies the
same finite-window normal calculation to an ever-growing number of ergodic
blocks. Consequently the sample cross term and the sample information
matrix both have linear-in-$T$ block averages, while their ratio approaches
the nonzero population quantity $G_L^{-1}q_L$ whenever $q_L\ne0$. Thus the
correct conclusion is a persistent coefficient-level asymptotic bias, not an
unbounded growth of the bias with $T$.
\end{remark}

\subsection{Proof of Theorem~\ref{thm:unbiased}}

\begin{proof}
For each fixed $n$, the random variable $K_j(X_{t_n})\sigma(X_{t_n})$
is $\mathcal{F}_{t_n}$-measurable, and
$\E[\xi_n\mid\mathcal{F}_{t_n}]=0$. Therefore
\begin{align}
  \E[K_j(X_{t_n})\sigma(X_{t_n})\xi_n]
  &=\E\!\left[\E\!\left[K_j(X_{t_n})\sigma(X_{t_n})\xi_n
    \mid\mathcal{F}_{t_n}\right]\right]\\
  &=0.
\end{align}
Summing over $n$ gives $\E[Z_j^{\mathrm{spatial}}]=0$ and hence
$\E[B_j]=\E[\sum_nK_j(X_{t_n})b(X_{t_n})\dt]
=\E[(Ac^*)_j]$. No conditioning on the full trajectory is used or implied;
indeed, such a conditioning statement would be false in general because
future states contain the past innovations that also enter the residual.
\end{proof}

\begin{remark}[Why the spatial estimator stays consistent: the
  growing-sample argument, not a cross-step cancellation]
\label{rmk:spatial_consistency_mechanism}
Theorem~\ref{thm:unbiased} alone --- like Part~(i) of
\cref{thm:endogeneity} --- establishes only that $\E[Z_j]=0$, which holds
for \emph{both} weight families and therefore cannot by itself explain why
one is consistent and the other is not. The cross term
$\E[A_{jk}Z_j]$ for a spatial kernel is, in general, also nonzero: writing
the analogous lag expansion, a shock $\xi_m$ at a step where the
trajectory is near centre $x_j$ generically correlates with
$K_j(X_{t_{m+\ell}})f_k(X_{t_{m+\ell}})$ at nearby lags, by the same
sensitivity mechanism as~\cref{eq:biasmn_app}. We do \emph{not} claim this
cross-step covariance vanishes or is small; instead, consistency is
established, as in the proof of Theorem~\ref{thm:consistency}, by a
sample-size argument that sidesteps the cross-covariance question
entirely. The $M$ kernel centres $x_1,\ldots,x_M$ are \emph{fixed} and
independent of $T$. As the ergodic trajectory lengthens, it revisits the
neighbourhood of each $x_j$ an unboundedly growing number of times, so the
quadratic variation of the martingale $Z_j^{\mathrm{spatial}}$,
\begin{equation}
  \langle Z^{\mathrm{spatial}}_j\rangle_N
  = \dt\sum_n K_j^2(X_{t_n})\sigma^2(X_{t_n}) = O(N\dt) = O(T),
\end{equation}
grows linearly in $T$, and the $L^2$ martingale strong law (used already
in the proof of Theorem~\ref{thm:consistency}, Step~2) gives
$Z_j^{\mathrm{spatial}}/N\to0$ almost surely. Combined with
$A_j/T\to\bar{A}_j$ (Theorem~\ref{thm:consistency}, Step~1), this gives
$A_jZ_j/N^2\to0$ regardless of the sign or magnitude of any particular
cross-step covariance, because the strong law requires only that the
quadratic variation grow sub-quadratically in $N$ (here linearly), not
that cross terms vanish or decay at any particular rate. This is precisely
the property that fixed-width temporal windows lack: a window's own
quadratic variation $\langle Z_j^{\mathrm{temp}}\rangle$ is $O(L/\dt)$,
\emph{fixed} as $T\to\infty$, so the martingale strong law never applies
to a single window's own sample --- its relative fluctuation never
shrinks --- and the same nonzero window-level population cross term is
repeated over the growing number $M_{\mathrm{win}}=T/L$ of ergodic windows instead of
being averaged away within a single fixed equation. The essential contrast,
stated most precisely, is therefore not about predictability or about whether
cross-step covariances are nonzero (they are, in both cases), but about
whether the signal-to-noise ratio of each individual regression equation grows
without bound as $T\to\infty$: it does for a fixed, revisited set of spatial
kernels (signal $\propto T$, noise $\propto\sqrt T$), it does not for a tiled
set of fixed-width temporal windows (bounded sample), and it does not for a
window widened with $T$ when the weighted regressor has mean zero (signal and
noise both $O(\sqrt T)$).
\end{remark}

\end{document}

%% file: table_sde_benchmark.tex
\begin{tabular}{@{}llccccc@{}}
\toprule
SDE & $b(x)$ & $a(x)$ & Drift err. & Diff.\ err. & Drift supp. & Diff.\ supp. \\
\midrule
Shifted OU & $0.5-1.5x$ & $0.25$ & $0.049$ & $8.0\times10^{-4}$ & 8/10 & 10/10 \\
Duffing (single well) & $-x-0.5x^{3}$ & $0.25$ & $0.111$ & $2.1\times10^{-3}$ & 5/10 & 10/10 \\
Tilted double well & $0.5+x-x^{3}$ & $0.4$ & $0.077$ & $2.7\times10^{-3}$ & 6/10 & 10/10 \\
Pure cubic & $-x^{3}$ & $0.4$ & $0.057$ & $2.6\times10^{-3}$ & 10/10 & 10/10 \\
Asymmetric cubic & $-x+0.5x^{2}-0.5x^{3}$ & $0.5$ & $0.102$ & $2.2\times10^{-3}$ & 9/10 & 10/10 \\
OU, multiplicative noise & $-1.2x$ & $0.1+0.1x^{2}$ & $0.037$ & $6.9\times10^{-3}$ & 7/10 & 10/10 \\
Double well, multiplicative noise & $x-x^{3}$ & $0.15+0.05x^{2}$ & $0.043$ & $3.7\times10^{-3}$ & 10/10 & 10/10 \\
Duffing, multiplicative noise & $-x-x^{3}$ & $0.1+0.1x^{2}$ & $0.075$ & $4.4\times10^{-3}$ & 9/10 & 10/10 \\
\bottomrule
\end{tabular}